\documentclass[11pt,reqno]{amsart}
\usepackage{amssymb}
\usepackage{amsfonts}
\usepackage{amsmath}
\usepackage{stmaryrd}
\usepackage{physics}
\usepackage{braket}
\usepackage{dsfont}
\usepackage{bbold}
\usepackage{graphicx}
\usepackage{relsize}
\usepackage{makecell}
\usepackage[table,xcdraw]{xcolor}
\usepackage{enumerate}
\usepackage[pagebackref, colorlinks = true, linkcolor = blue, urlcolor  = blue, citecolor = red]{hyperref}
\usepackage[margin=1in]{geometry}
\usepackage{enumitem}
\usepackage{mathtools}
\usepackage{graphbox}
\usepackage{comment}
\usepackage{float}
\usepackage{hhline}
\usepackage[capitalize]{cleveref}
\usepackage{accents}

\renewcommand{\epsilon}{\varepsilon}
\renewcommand{\phi}{\varphi}
\newcommand{\overbar}[1]{\mkern 1.5mu\overline{\mkern-1.5mu#1\mkern-1.5mu}\mkern 1.5mu}

\newcommand{\DUC}{\mathsf{DUC}}
\newcommand{\CDUC}{\mathsf{CDUC}}
\newcommand{\DOC}{\mathsf{DOC}}

\newcommand{\TCP}{\mathsf{TCP}}

\newcommand{\M}[1]{\mathcal{M}_{#1}(\mathbb{C})}

\newcommand{\C}[1]{\mathbb{C}^{#1}}

\newcommand{\Msa}[1]{\mathcal{M}^{sa}_{#1}(\mathbb{C})}

\newcommand{\Wg}{\operatorname{Wg}}
\newcommand{\Mob}{\operatorname{M\ddot{o}b}}
\renewcommand{\ring}[1]{\accentset{\circ}{#1}}

\newtheorem{theorem}{Theorem}[section]

\newtheorem*{definition*}{Definition}
\newtheorem{proposition}[theorem]{Proposition}
\newtheorem{corollary}[theorem]{Corollary}
\newtheorem{lemma}[theorem]{Lemma}
\newtheorem{remark}[theorem]{Remark}

\newtheorem{conjecture}[theorem]{Conjecture}
\newtheorem*{conjecture*}{Conjecture}

\newcommand\vertarrowbox[3][6ex]{
  \begin{array}[t]{@{}c@{}} #2 \\
  \left\uparrow\vcenter{\hrule height #1}\right.\kern-\nulldelimiterspace\\
  \makebox[0pt]{\scriptsize#3}
  \end{array}
}

\theoremstyle{definition}

\definecolor{darkgreen}{rgb}{0,0.392,0}

\newtheorem{innercustomgeneric}{\customgenericname}
\providecommand{\customgenericname}{}
\newcommand{\newcustomtheorem}[2]{
  \newenvironment{#1}[1]
  {
   \ifdefined\crefalias\crefalias{innercustomgeneric}{#2}\fi
   \renewcommand\customgenericname{#2}
   \renewcommand\theinnercustomgeneric{##1}
   \innercustomgeneric
  }
  {\endinnercustomgeneric}
  \ifdefined\crefname\crefname{#2}{#2}{#2s}\fi
}

\newcustomtheorem{customthm}{Theorem}
\newcustomtheorem{customlemma}{Lemma}

\newcommand{\diag}{\mathrm{diag}}
\newcommand{\id}{\mathrm{id}}
\newcommand{\la}{\langle}
\newcommand{\ra}{\rangle}
\newcommand{\Om}{\Omega}

\author{Ion Nechita}
\email{nechita@irsamc.ups-tlse.fr}
\address{Laboratoire de Physique Th\'eorique, Universit\'e de Toulouse, CNRS, UPS, France}

\author{Sang-Jun Park}
\email{spark@irsamc.ups-tlse.fr}
\address{Laboratoire de Physique Th\'eorique, Universit\'e de Toulouse, CNRS, UPS, France}

\title{Random covariant quantum channels}

\begin{document}
\begin{abstract}
The group symmetries inherent in quantum channels often make them tractable and applicable to various problems in quantum information theory. In this paper, we introduce natural probability distributions for covariant quantum channels. Specifically, this is achieved through the application of ``twirling operations'' on random quantum channels derived from the Stinespring representation that use Haar-distributed random isometries. We explore various types of group symmetries, including unitary and orthogonal covariance, hyperoctahedral covariance, diagonal orthogonal covariance (DOC), and analyze their properties related to quantum entanglement based on the model parameters. In particular, we discuss the threshold phenomenon for positive partial transpose and entanglement breaking properties, comparing thresholds among different classes of random covariant channels. Finally, we contribute to the PPT$^2$ conjecture by showing that the composition between two random DOC channels is generically entanglement breaking.
\end{abstract}

\maketitle

\tableofcontents

\section{Introduction}

Quantum channels are essential mathematical constructs in quantum theory, modeling the most general evolution open quantum systems can undergo. Generalizing the unitary quantum dynamics described by Schr\"odinger's equation, quantum channels have been studied extensively, from different perspectives \cite{nielsen2010quantum,wilde2017quantum,watrous2018theory}. In physics, quantum channels describe the evolution of quantum states as they propagate through space and/or time, and model usually the noise coming from the environment. In quantum information theory, quantum channels model the transmission of quantum information between distant parties, akin to a classical communication channel but with quantum properties. Mathematically, quantum channels are represented by completely positive, trace preserving maps. 

The study of \emph{covariant quantum evolutions} was initiated by Holevo \cite{holevo1993note,holevo1996covariant}. This type of covariance encodes usually a symmetry present in the system Hamiltonian and can be mathematically formulated as follows: 
$$\forall g \in G, \qquad \Phi\Big(\pi_A(g) \cdot \pi_A(g)^* \Big) = \pi_B(g) \Phi(\cdot) \pi_B(g)^*.$$
Above, $\pi_A$ and $\pi_B$ are unitary representations of the physical symmetry group $G$. A channel $\Phi : \M{d} \to \M{d}$ having the property above is called a \emph{$(\pi_A, \pi_B)$-covariant channel}. This notion is not only important in physically-relevant settings where symmetry is present, but also as a tool in several areas of quantum information theory, such as quantum cloning \cite{keyl1999optimal}, programmability \cite{GBW21}, quantum Shannon theory \cite{WH02, Kin03, Hol05, DHS06, KW09, DTW16, WTM17,  BCLY20}, etc. The importance of this notion stems also from the fact that group symmetries inherent in quantum channels often make them tractable and applicable to various problems in quantum information theory, especially related to quantum entanglement \cite{VW01, COS18, Christandl2018,  siudzinska2018quantum, BCS20, singh2021diagonal, singh2022ppt, PJPY23, PY23}. Their general structure and mathematical properties have been extensively studied in recent years \cite{Al14, MSD17, MS22, LY22}.

Another way to probe the structure of quantum channels is to consider their generic behavior, that is to study \emph{random quantum channels}. Broadly speaking, random quantum channels model \emph{typical} behavior of quantum operations \cite{bruzda2009random}, hence it is of crucial importance to understand their physical and mathematical properties, such as: entanglement, information transmission capacity, spectral properties, etc. A systematic study of the different models of random quantum channels and of their properties can be found in \cite{kukulski2021generating}. Moreover, random constructions for quantum channels have played a crucial role in many areas of quantum Shannon theory, mainly as a \emph{source of (counter-)examples} \cite{HW08, Has09, BCN16, CN16, CP22}; this situation is parallel to classical information theory, where random constructions of codes and channels are ubiquitous.  

\smallskip
In this work, we consider \emph{random covariant quantum channels}. We introduce and study \emph{probability distributions on the set of covariant quantum channels}, for different choices of group symmetry. 
We find in which asymptotic regimes such channels are \emph{PPT} or \emph{entanglement breaking}. In particular, we explain the \textit{threshold phenomenon} for these two properties. We show that in the case of diagonal symmetry, random covariant channels satisfy the \emph{PPT$^2$ conjecture}.
Our work sets up the framework for the study of random covariant channels, by defining the relevant probabilistic models in the most common settings of group symmetry. We expect the models and the results in this paper to find further applications in quantum information theory, in the study of problems having a specific type of symmetry and also as a rich source of examples of channels with a tractable form. 

We shall consider several cases of (random) covariant channels, corresponding to the standard representations of: 
\begin{itemize}
    \item the unitary group $\mathcal U_d = \{U \in \mathcal M_d(\mathbb C) \, : \, U^*U = I_d\}$
    \item the orthogonal group $\mathcal O_d = \{U \in \mathcal M_d(\mathbb R) \, : \, O^{\top}O = I_d\}$
    \item the hyperoctahedral group $\mathcal H_d = \{\pm 1\} \wr \mathcal S_d$ of signed permutation matrices
    \item the diagonal unitary group $\mathcal{DU}_d = \{\operatorname{diag}(e^{\mathrm i \theta_1}, \ldots, e^{\mathrm i \theta_d}) \, : \, \theta_1, \ldots, \theta_d \in \mathbb R\}$
    \item the diagonal orthogonal group $\mathcal{DO}_d = \{\operatorname{diag}(\epsilon_1 \ldots, \epsilon_d) \, : \, \epsilon_1, \ldots, \epsilon_d  = \pm 1\}$.
\end{itemize} 

\begin{figure}[htb]
    \centering
    \includegraphics[scale=1]{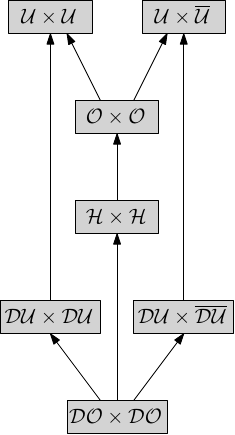} \qquad \qquad \qquad \qquad \includegraphics[scale=1]{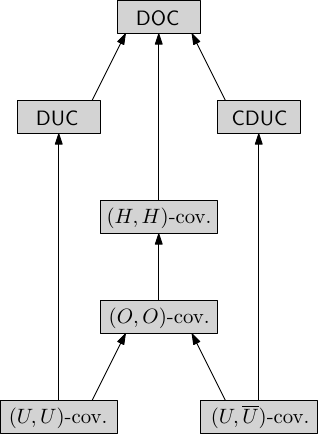}
    \caption{The groups we consider (left) and the corresponding class of covariant channels (right). Edges correspond to set inclusion relations, meaning that the source is contained in the target. Top elements correspond to the largest groups (reps.~channel sets), while bottom elements correspond to the smallest groups (resp.~channel sets).}
    \label{fig:groups-channels-inclusion}
\end{figure}

We depict in \cref{fig:groups-channels-inclusion} the groups (left) and the corresponding classes of covariant channels (right). Edges correspond to set inclusion relations, meaning that the source is contained in the target. Top elements correspond to the largest groups (reps.~channel sets), while bottom elements correspond to the smallest groups (resp.~channel sets). The notation is self-evident, with the exception of (see \cref{sec:DOC}): 
\begin{align*}
    \text{DUC channels} &\quad \leftrightarrow \quad (\mathcal{DU}_d \times \mathcal{DU}_d) \text{ covariance}\\
    \text{CDUC channels} &\quad \leftrightarrow \quad (\mathcal{DU}_d \times \overline{\mathcal{DU}}_d) \text{ covariance}\\
    \text{DOC channels} &\quad \leftrightarrow \quad (\mathcal{DO}_d \times \mathcal{DO}_d) \text{ covariance}.
\end{align*}

The contributions of our work are two-fold: 
\begin{itemize}
\item \emph{Entanglement theory}: we find the range of parameters $s$ for which random covariant channels are asymptotically PPT, respectively EB. In particular, we specify the \textit{threshold} $s_0$ such that, in the limit $d\to\infty$, the probability that random covariant channels are PPT/EB tends to $1$ (resp. $<1$) whenever whenever $s\gg s_0$ (resp. $s\ll s_0$).

\item \emph{PPT$^2$ conjecture}: we generalize the main result of \cite{singh2022ppt}. Furthermore, we show that random DOC channels satisfy the PPT$^2$ conjecture, under very general conditions. 
\end{itemize}

We recall that a quantum channel is \emph{PPT} (positive partial transpose) if the composition of the channel with the (non-physical) transpose map is still a quantum channel. On the other hand, a quantum channel is called \emph{entanglement breaking} (EB) if, when applied to one system of a bipartite quantum state, the output bipartite state is separable \cite{horodecki2003entanglement}.
Our results on the PPT and EB properties of random covariant tables are informally presented in the table below. The rows of the table are ordered, starting with the largest groups on the top, and finishing with the smallest group at the bottom. The last row of the table corresponds to no symmetry whatsoever, hence to random channels before the twirling operation. Here $s_0$ denotes the threshold for corresponding properties. For example, $s_0\sim \text{const.}$ for PPT means that we have the property with probability $1$ (resp. $<1$) whenever $d\to \infty$ and $s=s(d)\to \infty$ (resp. $s$ is fixed). Note that the two last columns are devoted to the entanglement breaking (EB) property of channels: since deciding whether a channels is EB is an $\mathsf{NP}$-hard problem, it is difficult in general to obtain sharp thresholds for this property. All the results present in the table are new, with the exception of the last row.  

\medskip 

\bgroup
\def\arraystretch{1.2}
  \begin{center}

    \label{tab-RanCov}
    \begin{tabular}{|c||c|c|c|} 
\hline
    Group symmetry & PPT property  & EB property & Not EB property \\ \hline 
\hline
    $(U,U)$-cov. & Always & Always & Never  \\
\hline
    $(U,\overline{U})$-cov. & $s_0\sim$ const. & $s_0\sim$ const. & $s_0\sim$ const. \\
\hline
    $(O,O)$-cov. & $s_0\sim$ const. & $s_0\sim$ const. & $s_0\sim$ const.  \\
\hline
    $(H,H)$-cov. & $s_0\sim$ const. & $s_0\sim$ const. & $s_0\sim$ const. \\
\hline
    DUC & $s_0\sim$ const. & $s\gtrsim d^{2+\epsilon}$ & $s=O(1)$  \\
\hline
    CDUC & $s_0\sim 4d$ & $s\gtrsim d^{2+\epsilon}$ & $s\leq (4-\epsilon)d$    \\
\hline
    DOC & $s_0\sim 4d$ & $s\gtrsim d^{2+\epsilon}$ & $s\leq (4-\epsilon)d$   \\
\hline
    No symmetry & $s_0\sim 4d^2$ \cite{Aub12} & $s \gtrsim d^3\log^2 d$ \cite{ASY14} & $s \lesssim d^3$ \cite{ASY14} \\
\hline
    \end{tabular}    
\end{center}
\egroup

\medskip

The formal statements for the unitary and orthogonal cases are given in \cref{thm-RanUUCovEB}; the ones for the hyperoctahedral case are given in \cref{thm-RanHHCovEB}. We present below the statement of a formal result, corresponding to the CDUC, DUC, and DOC rows of the table above; see \cref{thm:DOC-PPT,thm-DOCEB} for the details.

\begin{customthm}{A}
Let $\Phi_\DUC \sim \mu^\DUC_{d,s}$, $\Phi_\CDUC \sim \mu^\CDUC_{d,s}$, and $\Phi_\DOC \sim \mu^\DOC_{d,s}$ be, respectively, random DUC, CDUC, and DOC channels. Then:
\begin{enumerate}
    \item In the regime $s$ is fixed and $d\to \infty$, we have $\lim_{d\to \infty}\mathbb{P}(\Phi_{\DUC} \text{ is PPT })=0$.     
    
    \item In the regime $s\gtrsim d^\epsilon$ for some $\epsilon>0$, $\Phi_\DUC$ is almost surely PPT as $d\to \infty$.
    
    \item In the regime $s/d\to 0$, $\Phi_{\CDUC}$ and $\Phi_{\DOC}$ are almost surely \textit{not} PPT as $d\to \infty$. 
    
    \item Consider the regime $s\sim cd$. 
    \begin{itemize}
        \item If $c<4$, then $\Phi_\CDUC$ and $\Phi_\DOC$ are almost surely \textit{not} PPT as $d\to \infty$.

        \item If $c>4$, then both $\Phi_\CDUC$ and $\Phi_\DOC$ have Choi matrices whose partial transpositions converge, almost surely as $d \to \infty$, in moments, to a probability measure supported on the positive real line.
    \end{itemize}

    \item In the regime $s\gtrsim d^{1+\epsilon}$ for some $\epsilon>0$, $\Phi_\CDUC$ and $\Phi_\DOC$ are almost surely PPT as $d\to \infty$.

    \item In the regime $s\gtrsim d^{2+\epsilon}$ for some $\epsilon>0$, $\Phi_\DUC$, $\Phi_\CDUC$, and $\Phi_\DOC$ are almost surely EB as $d\to \infty$.
\end{enumerate}
\end{customthm}

Note that in item (5) above we show that the smallest eigenvalue of the partial transpose of the Choi matrix of the random channels we consider converges to a positive value. We also obtain results about the \emph{realignment criterion} \cite{chen2003matrix,rudolph2005further} applied to the Choi matrix of DUC, CDUC, and DOC channels in \cref{prop:DOC-realignment}.

\medskip

Our second contribution concerns the \emph{PPT$^2$ conjecture} for DOC channels. Motivated by the theory of quantum repeaters, this conjecture \cite{PPTsq, Christandl2018} has received a lot of attention in the recent years. The conjecture states that the composition of two arbitrary {PPT} quantum channels is entanglement breaking. Recent works establish the conjecture in several restricted scenarios, or with weaker conclusions \cite{Lami2015entanglebreak,Kennedy2017,Rahaman2018,Christandl2018,Chen2019,hanson2020eventually,girard2020convex,jin2020investigation, singh2022ppt}. Moreover, in the case of random quantum channels (without any covariance property), it has been proven \cite{collins2018ppt} that two \emph{independent} random quantum channels that are PPT satisfy the conjecture.

Regarding the quantum channels with diagonal orthogonal symmetry, it has been shown in \cite{singh2022ppt} that the PPT$^2$ conjecture holds for \emph{any} two DUC / CDUC channels, while the same question for DOC channels has been left open. We tackle this question from both \textit{deterministic} and \textit{random} perspectives. The results can be summarized as follows; see \cref{cor-PPTDOC,thm:PPTsquared-DOC} for the details.

\begin{customthm}{B} \label{thmB}
\begin{enumerate}
    \item The PPT$^2$ conjecture holds for the composition between {\emph{any}} DOC and (C)DUC channels. That is, if $\Phi_1,\Phi_2:\M{d}\to \M{d}$ are two PPT channels such that $\Phi_1$ is DOC and $\Phi_2$ is DUC or CDUC, then both $\Phi_1\circ \Phi_2$ and $\Phi_2\circ \Phi_1$ are entanglement breaking.
    
    \item Consider two \emph{random} DOC channels $\Phi_i \sim \mu^\DOC_{d,s_i}$, $i=1,2$, in the asymptotic regime where $d \to \infty$ and $s_i\gtrsim d^{t_i}$ for some constants $t_i>0$. Then, almost surely as $d\to \infty$, the composition $\Phi_1 \circ \Phi_2$ is entanglement breaking.  In particular, in the asymptotic regime above, any two random DOC channels satisfy the PPT$^2$ conjecture. 
\end{enumerate}
\end{customthm}

Note that \cref{thmB} (1) {generalizes} the main result of \cite{singh2022ppt} {by allowing one of the channels to belong to the larger class of DOC channels}. Moreover, it is worth mentioning that the hypotheses of \cref{thmB} (2) are very weak when compared to those of \cite{collins2018ppt}: the channels need not be PPT, nor independent.

\medskip

Our paper is organized as follows. In \cref{sec:background} we recall the theory of covariant quantum channels and that of random quantum channels; we also give the general recipe for sampling random covariant quantum channels. \cref{sec:conjugate-unitary-orthogonal-covariant,sec:hyperoctahedral-covariant,sec:DOC} deal, respectively, with the three main classes of symmetry groups: unitary/orthogonal, hyperoctahedral, and diagonal unitary/orthogonal. These sections present the general structure theorems for the corresponding covariant channels, and then gather our main results regarding the PPT and EB properties of the random covariant quantum channels. Finally, \cref{sec:PPT2} contains our result about DOC channels satisfying the PPT$^2$ conjecture.

\section{Random covariant quantum channels}\label{sec:background}

In this section we recall two of the fundamental notions discussed in this paper: that of symmetric (covariant) quantum channels, and that of random quantum channels. In the rest of the paper, we shall combine these two points of view by analyzing the properties of \emph{random covariant quantum channels}.

\subsection{Covariant quantum channels}

For a compact group $G$, consider (finite-dimensional) {\it unitary representations} $\pi:G\to \mathcal U_{d}$, $\pi_A:G\to \mathcal U_{d_A}$, and $\pi_B:G\to \mathcal U_{d_B}$ of $G$. Then we call
\begin{enumerate}
    \item a matrix $X\in \M{d}$ \textit{$\pi$-invariant} if $\pi(x)X\pi(x)^*=X$ for all $x\in G$,

    \item a linear map $\Phi:\M{d_A}\to \M{d_B}$ \textit{$(\pi_A,\pi_B)$-covariant} if $\Phi(\pi_A(x)Z\pi_A(x)^*)=\pi_B(x)\Phi(Z)\pi_B(x)^*$ for all $x\in G$ and $Z\in \M{d_A}$.
\end{enumerate}
Let us denote by ${\rm Inv}(\pi)$ (resp. ${\rm Cov}(\pi_A,\pi_B)$) the set of all $\pi$-invariant matrices (resp. $(\pi_A,\pi_B)$-covariant linear maps). $\pi$ is called \textit{irreducible} if ${\rm Inv}(\pi)=\C{}I_d$ (which coincides with usual definition of irreducibility in representation theory, thanks to Schur's lemma). If $\pi$ is irreducible, so is the \textit{contragredient representation} $\overline{\pi}:G\to \mathcal U_{d}$ of $\pi$ which is defined by $\overline{\pi}(x)=\overline{\pi(x)}$ for $x\in G$. For unitary representations $\pi_A:G\rightarrow \mathcal{U}_{d_A}$ and $\pi_B:G\rightarrow \mathcal{U}_{d_B}$, the \textit{direct sum representation $\pi_A\oplus \pi_B:G\to \mathcal U_{d_A+d_B}$} (resp. {\it tensor representation} $\pi_A\otimes \pi_B:G\rightarrow \mathcal U_{d_Ad_B}$) is given by $(\pi_A\oplus \pi_B)(x)=\pi_A(x)\oplus \pi_B(x)$ (resp. $(\pi_A\otimes \pi_B)(x)=\pi_A(x)\otimes \pi_B(x)$) for $x\in G$. Two unitary
representations $\pi_1:G\to \mathcal U_{d_1}$ and $\pi_2:G\to \mathcal U_{d_2}$ are said to be \textit{unitarily equivalent} if there exists a unitary matrix $U:\C{d_1}\to \C{d_2}$ such that $\pi_2(x)=U\pi_1(x)U^*$ for all $x\in G$, where we denote by $\pi_1\cong \pi_2$.

The following is a well-known fact in representation theory.

\begin{proposition} [\cite{Simon95}] \label{prop-irrdecomp}
Suppose that a unitary representation $\pi:G\rightarrow \mathcal U_{d}$ has an irreducible decomposition of the form 
\begin{equation}\label{eq-fusion}
    \pi\cong \bigoplus_{i=1}^l \sigma_i\otimes I_{m_i},
\end{equation}
where $\sigma_i:G\to \mathcal U_{n_i}$ ($i=1,2,\ldots, l$) are mutually inequivalent irreducible sub-representations of $\pi$ with multiplicity $m_i$, respectively. Then we have
\begin{equation}\label{eq-fusion2}
    {\rm Inv}(\pi)\cong  \bigoplus_{i=1}^l I_{n_i}\otimes \M{m_i} \subseteq \M{n_i}\otimes \M{m_i}.
\end{equation}
\end{proposition} 

In this section, we are interested in two types of twirling operations: for unitary representations $\pi,\pi_A,\pi_B$ of $G$,
\begin{enumerate}
    \item \textit{$\pi$-twirling} $\mathcal{T}_{\pi}X:=\mathbb{E}_{G}[{\rm Ad}_{\pi(\cdot)}(X)]=\mathbb{E}_{G}[\pi(\cdot)X\pi(\cdot)^*]$ for a matrix $X\in \M{d}$,

    \item \textit{$(\pi_A,\pi_B)$-twirling} $\mathcal{T}_{\pi_A,\pi_B}\Phi:=\mathbb{E}_G[{\rm Ad}_{\pi_B(\cdot)^*}\circ \Phi\circ {\rm Ad}_{\pi_A(\cdot)}]$ for a linear map $\Phi:\M{d_A}\to \M{d_B}$,
\end{enumerate}
where ${\rm Ad}_U: Z\mapsto UZU^*$ and the expectation is defined with respect to the (normalized) Haar measure of the compact group $G$.

Let us collect some useful properties of the twirling operations for the next section. First of all, $\mathcal{T}_{\pi}$ is a \textit{conditional expectation} onto the $*$-subalgebra ${\rm Inv}(\pi)$ of $\M{d}$. Note that for any (finite-dimensional) von Neumann subalgebra $\mathcal{M}$ of $\M{d}$, there is a unique TP conditional expectation of $\M{d}$ onto $\mathcal{M}$ \cite[Lemma 1.5.11]{BrOz}. For example, the map $X\in \M{n}\otimes \M{m} \mapsto \frac{1}{n}I_n\otimes {\rm Tr}_n(X)$ is the unique TP conditional expectation onto $\mathcal{M}=I_n\otimes \M{m}$. This observation allow us to get the following explicit formula of the twirling map $\mathcal{T}_{\pi}$ for the case $\mathcal{M}={\rm Inv}(\pi)$.

\begin{proposition} \label{prop-twirlformula}
In Proposition \ref{prop-irrdecomp}, let $\Pi_i$ be the orthogonal projection from $\C{d}$ onto $\C{n_i}\otimes \C{m_i}$. Then the twirling $\mathcal{T}_{\pi}(X)$ of $X\in \M{d}$ is given by
\begin{equation}
    \mathcal{T}_{\pi}(X)=\bigoplus_{i=1}^l \frac{1}{n_i}I_{n_i}\otimes {\rm Tr}_{n_i}(\Pi_iX\Pi_i).
\end{equation}
In particular, if the irreducible decomposition of $\pi$ is multiplicity-free, i.e., if $m_i\equiv 1$ for all $i=1,2,\cdots,l$, then 
\begin{equation}
    \mathcal{T}_{\pi}(X)=\sum_{i=1}^l \frac{{\rm Tr}(\Pi_i X)}{n_i}\Pi_i.
\end{equation}
\end{proposition}

\begin{proposition} \label{prop-twirlcomposition}
Let $\pi_0,\pi_1:G\to \mathcal U_{d}$ be two unitary representations of $G$ such that ${\rm Ran\,}\pi_1 \subset {\rm Ran\,}\pi_0$. Then
    $$\mathcal{T}_{\pi_0}\circ\mathcal{T}_{\pi_1}=\mathcal{T}_{\pi_1}\circ \mathcal{T}_{\pi_0}=\mathcal{T}_{\pi_0}.$$
\end{proposition}
\begin{proof}
Recall that each $\mathcal{T}_{\pi_i}$ is a TP conditional expectation onto ${\rm Inv}(\pi_i)$. Since ${\rm Inv}(\pi_0)\subset {\rm Inv}(\pi_1)$, we obviously have $\mathcal{T}_{\pi_1}\circ \mathcal{T}_{\pi_0}=\mathcal{T}_{\pi_0}$. Next, observe that
    $$(\mathcal{T}_{\pi_0}\circ\mathcal{T}_{\pi_1})\circ (\mathcal{T}_{\pi_0}\circ\mathcal{T}_{\pi_1})=\mathcal{T}_{\pi_0}\circ (\mathcal{T}_{\pi_1}\circ \mathcal{T}_{\pi_0})\circ\mathcal{T}_{\pi_1}=\mathcal{T}_{\pi_0}\circ\mathcal{T}_{\pi_1}\circ \mathcal{T}_{\pi_1}=\mathcal{T}_{\pi_0}\circ\mathcal{T}_{\pi_1}.$$
Since $\mathcal{T}_{\pi_0}\circ\mathcal{T}_{\pi_1}$ is a unital CPTP map and ${\rm Ran}(\mathcal{T}_{\pi_0}\circ\mathcal{T}_{\pi_1})={\rm Inv}(\pi_0)$, $\mathcal{T}_{\pi_0}\circ\mathcal{T}_{\pi_1}$ is a TP conditional expectation onto ${\rm Inv}(\pi_0)$. By the uniqueness, we conclude that $\mathcal{T}_{\pi_0}\circ\mathcal{T}_{\pi_1}=\mathcal{T}_{\pi_0}$.
\end{proof}

Recall that the (unnormalized) Choi-Jamio{\l}kowski matrix $J(\Phi)$ of a linear map $\Phi:\M{d_A}\to \M{d_B}$ is defined by
\begin{equation}\label{eq:def-Choi-matrix}
    J(\Phi):=d_A [\Phi \otimes \id](\ketbra{\Omega_{d_A}}{\Omega_{d_A}}) = \sum_{i,j=1}^{d_A}\Phi(|i\ra\la j|)\otimes |i\ra\la j|\in \M{d_B}\otimes \M{d_A},
\end{equation}
where $\ket{\Omega_d}$ is the \emph{maximally entangled state}
\begin{equation}\label{eq:def-max-ent-state}
    \ket{\Omega_d} := \frac{1}{\sqrt d} \sum_{i=1}^d \ket{ii} \in \C{d} \otimes \C{d},
\end{equation}
and, for future reference, $F_d$ is the \emph{flip} (or swap) unitary operator
\begin{equation}\label{eq:def-flip}
    F_d := \sum_{i,j=1}^d \ketbra{ij}{ji} \in \mathcal U_{d^2}.
\end{equation}
Importantly, the PPT and entanglement breaking (EB) properties of a channel can be read on its Choi matrix:
\begin{align*}
    \Phi \text{ is PPT } &\iff J(\Phi) \text{ is PPT}\\
    \Phi \text{ is EB } &\iff J(\Phi) \text{ is separable.}
\end{align*}
We recall that a positive semidefinite matrix $X \in \M{d_A}\otimes \M{d_B}$ is called 
\begin{align*}
    \text{ PPT } &\quad \text{ if } \quad X^{\Gamma} = [\id \otimes \top](X) \geq 0\\
    \text{ separable } &\quad \text{ if } \quad X = \sum_i A_i \otimes B_i \quad \text{for positive semidefinite matrices } A_i \in \M{d_A}, \, B_i \in \M{d_B}.    
\end{align*}

The following Proposition \ref{prop-twirling} establishes the connection between group symmetries of quantum channels and those of their Choi matrices.

\begin{proposition} [\cite{PJPY23}] \label{prop-twirling}
Let $\pi_A:G\to \mathcal U_{d_A}$ and $\pi_B:G\to \mathcal U_{d_B}$ be unitary representations of $G$. Then, for any linear map $\Phi:\M{d_A}\to \M{d_B}$,
    $$J(\mathcal{T}_{\pi_A,\pi_B}\Phi)=\mathcal{T}_{{\pi_B}\otimes \overline{\pi_A}}\left ( J(\Phi)\right ).$$
In particular, $\Phi\in {\rm Cov}(\pi_A,\pi_B)$ if and only if $J(\Phi)\in {\rm Inv}(\pi_B \otimes \overline{\pi_A})$. Moreover, if $\Phi$ is PPT / EB, then so is $\mathcal{T}_{\pi_A,\pi_B}\Phi$.

\end{proposition}

\subsection{Random quantum channels}\label{sec:random-quantum-channels}

We shall now endow the set of all quantum channels with a probability measure, laying down the foundation for the study of symmetric (covariant) channels in the subsequent sections. The set of quantum channels (i.e.~completely positive, trace preserving maps) $\Phi:\M{d}\to \M{d}$ is a compact convex subset of the set of linear maps $\M{d}\to \M{d}$. Actually, since quantum channels preserve Hermitian matrices, one can see it as a subset of hermiticity preserving maps $\Msa{d}\to \Msa{d}$. The trace preservation condition adds a number of linear restrictions, easily seen in the Choi-Jamio{\l}kowski picture from Eq.~\eqref{eq:def-Choi-matrix}: a linear map $\Phi$ is trace preserving iff
$$[\Tr \otimes \id](J(\Phi)) = I_d.$$

Let $V$ be a Haar-distributed random isometry of size $ds\times d$ and define the random Stinespring channel $\Phi_V:\M{d}\to \M{d}$ by
\begin{equation}\label{eq:def-Phi_V}
    \Phi_V(\rho):=(\id_d\otimes {\rm Tr}_s)(V\rho V^*).
\end{equation}

The probability distribution of $\Phi_{V}$ depends on the choice of dimension parameters $d$ and $s$, and we shall denote it by $\mu_{d,s}$. Let us recall that a Haar-distributed random isometry $V: \C{d} \to \C{ds}$ can be obtained by truncating a random Haar-distributed unitary matrix $U \in \mathcal U_{ds}$.

According to \cite{KNPPZ21}, we can give an equivalent definition for a random Stinespring channel $\Phi_V$. We first consider a random Wishart matrix $W=GG^*$ of parameter $(d^2, s)$ and set $H:=(\Tr_d\otimes \id_d)W$. Then $H$ is almost surely invertible, and the probability distribution of $J(\Phi_V)$ is the same as that of
\begin{equation} \label{eq-RanChoi}
    \tilde{J}:=(I_d\otimes H^{-1/2})W(I_d\otimes H^{-1/2}).
\end{equation}

Note that the matrix $\tilde J$ is PPT (resp. separable) if and only if $W$ is PPT (resp. separable). We gather below some of the foundational results from \cite{Aub12, ASY14, collins2018ppt}. In our work, we shall prove similar results in the case of covariant channels.

\begin{proposition}
Consider the random Stinespring channel $\Phi_{d,s}$ with parameter $(d,s)$.
\begin{enumerate}
    \item (Threshold phenomenon of PPT property for $\Phi$) For any $\epsilon>0$
    \begin{itemize}
        \item if $s<(4-\epsilon)d^2$, $\lim_{d \to \infty} \mathbb P(\Phi_{d,s} \text{ is PPT})=0$,

        \item if $s>(4+\epsilon)d^2$, $\lim_{d \to \infty} \mathbb P(\Phi_{d,s} \text{ is PPT})=1$.
    \end{itemize}

    \item (Threshold phenomenon of EB property for $\Phi$) There exist constants $C_1,C_2>0$ and a sequence $s_d$ satisfying 
        $$C_1d^3\leq s_d \leq C_2d^3 (\log d)^2$$ 
    such that for any $\epsilon>0$,
    \begin{itemize}
        \item if $s< (1-\epsilon) s_d$, $\lim_{d \to \infty} \mathbb P(\Phi_{d,s} \text{ is EB})=0$,

        \item if $s> (1+\epsilon) s_d$, $\lim_{d \to \infty} \mathbb P(\Phi_{d,s} \text{ is EB})=1$.
    \end{itemize}

    \item (PPT$^2$ conjecture for $\Phi$) Let $\widetilde{\Phi}_{d,s}$ be a random Stinespring channel independent to $\Phi_{d,s}$. For any $\epsilon>0$, if
        $$(4+\epsilon)d^2<s<(1-\epsilon)s_d$$
    so that $\Phi_{d,s}$ and $\widetilde{\Phi}_{d,s}$ are generically PPT and not EB, then $\lim_{d \to \infty} \mathbb P(\Phi_{d,s}\circ \widetilde{\Phi}_{d,s} \text{ is EB})=1$.
\end{enumerate}
\end{proposition}

Let us end this introductory section by providing the general recipe that shall be used to produce \textit{random covariant channels} in the following sections:
\begin{enumerate}
    \item first, we sample a random quantum channel $\Phi_V$, without any symmetry, from the ensemble $\mu_{d,s}$, see \eqref{eq:def-Phi_V}
    \item we then symmetrize $\Phi$ using the twirling operation below, to obtain a random $(\pi_A, \pi_B)$-covariant channel $\Phi_{\pi_A,\pi_B}$:
    \begin{equation} \label{eq-random-twirling}
        \Phi_{\pi_A,\pi_B}(X):=\mathcal T_{\pi_A, \pi_B}(\Phi_V)(X)= \int_G \pi_B(g)^* \Phi_V\Big(\pi_A(g) X \pi_A(g)^* \Big) \pi_B(g) \, \mathrm{d}\mathfrak h(g),
    \end{equation}
    where $\mathfrak h$ is the Haar measure on the group $G$.
\end{enumerate}

\section{Random orthogonal covariant channels}\label{sec:conjugate-unitary-orthogonal-covariant}

In this section we study three basic and important classes of linear maps in quantum information theory, those having a unitary or orthogonal covariance group. Unitary covariance is a very strong form of symmetry, since input states can be diagonalized, the effect of the unitary rotation being transferred to the output of the channel. Such strong symmetry makes very suitable for analysis in various settings related to quantum Shannon theory. We recall below their definition. 

\begin{enumerate}
    \item The \emph{depolarizing maps} characterized by $(U,U)$-covariance:
        $$\Phi(UZU^*)=U\Phi(Z)U^*, \;\;  Z\in \M{d},\;U\in \mathcal U_d,$$

    \item The \emph{transpose depolarizing maps}
    characterized by $(U,\overline{U})$-covariance:
        $$\Phi(UZU^*)=\overline{U}\Phi(Z)U^{\top}, \;\; Z\in \M{d},\; U\in \mathcal U_{d},$$

    \item The \textit{orthogonal covariant maps} 
    characterized by $(O,O)$-covariance:
        $$\Phi(OZO^{\top})=O\Phi(Z)O^{\top}, \;\;  Z\in \M{d},\;O\in \mathcal{O}_d,$$
\end{enumerate}

\begin{proposition}\label{prop:UUO-covariant}
    Let $\Phi : \M{d} \to \M{d}$ be a quantum channel. Then, the $(U,U)$, $(U, \overline U)$, $(O,O)$ twirlings of $\Phi$ are given, respectively, by: 
    \begin{align}
        \mathcal T_{U,U}(\Phi) &= p \cdot \id + (1-p) \cdot \Delta \nonumber\\
        \mathcal T_{U, \overline U}(\Phi) &= q \cdot \top + (1-q) \cdot \Delta \label{eq-OOtwirl}\\
        \mathcal T_{O,O}(\Phi) &=   p'\cdot \id + q' \cdot \top + (1-p'-q')\cdot \Delta, \nonumber
    \end{align}
    where $\id, \top, \Delta$ are, respectively, the identity, the transposition, and the completely depolarizing linear maps. Here the coefficients $p,q,p',q'$ are explicitly related to $\Phi$ by
    \begin{align*}
        p&=\frac{\lambda_1-1}{d^2-1},\\
        q&=\frac{\lambda_2-1}{d^2-1},\\
        p'&=\frac{(d+1)\lambda_1-\lambda_2-d}{d(d+2)(d-1)},\\
        q'&=\frac{(d+1)\lambda_2-\lambda_1-d}{d(d+2)(d-1)},
    \end{align*}
    where the parameters $\lambda_{1,2}$ associated to the channel are defined (see also Figure \ref{fig:lambda-1-2-Phi}):
 \begin{align*}
     \lambda_1(\Phi) &:= d\la \Om_d|J(\Phi)|\Om_d\ra \\
     \lambda_2(\Phi) &:= \Tr(J(\Phi)F_d).
 \end{align*}
 Above we have used maximally entangled state $\ket{\Omega_d}$ and the flip operator $F_d$ defined respectively in \cref{eq:def-max-ent-state} and \cref{eq:def-flip}.

\end{proposition}
\begin{proof}
All the formulas follow from \cref{prop-twirlformula}. Regarding the $(U,U)$-twirling $\mathcal{T}_{U,U}$, note that the representation $U\in \mathcal{U}_d \mapsto U\otimes \overline{U}$ can be decomposed into two irreducible sub-representations whose corresponding projections are
    $$|\Om_d\ra\la \Om_d|, \quad I_{d^2}-|\Om_d\ra\la \Om_d|.$$
Therefore, \cref{prop-twirlformula} and \ref{prop-twirling} implies that
\begin{align*}
    J(\mathcal{T}_{U,U}\Phi)&=\mathcal{T}_{U\otimes \overline{U}}(J(\Phi))\\
    &=\Tr(|\Om_d\ra\la \Om_d|J(\Phi))|\Om_d\ra\la \Om_d|+\Tr((I_{d^2}-|\Om_d\ra\la \Om_d|)J(\Phi))\frac{I_{d^2}-|\Om_d\ra\la \Om_d|}{d^2-1}\\
    &=\frac{\lambda_1}{d}|\Om_d\ra\la \Om_d|+\left(d-\frac{\lambda_1}{d}\right)\frac{I_{d^2}-|\Om_d\ra\la \Om_d|}{d^2-1}\\
    &=\frac{\lambda_1-1}{d^2-1}d|\Om_d\ra\la \Om_d|+ \frac{d^2-\lambda_1}{d^2-1}\frac{I_{d^2}}{d},
\end{align*}
where we used $\Tr(J(\Phi))=d$ since $\Phi$ is TP. Therefore, we have the formula $\mathcal{T}_{U,U}(\Phi)$ as in \cref{eq-OOtwirl}.

The remaining formulas can be derived analogously by noting that the representations $U\in \mathcal{U}_d\mapsto U\otimes U$ and $O\in \mathcal{O}_d\mapsto O\otimes O$ have the irrep decompositions corresponding to the projections (see \cite[Appendix A]{tura2016characterizing} and \cite{Bra37})
    $$\Pi_{\rm sym}:=\frac{I_{d^2}+F_d}{2},\quad \Pi_{\rm anti}:=\frac{I_{d^2}-F_d}{2},$$
and
    $$|\Om_d\ra\la \Om_d|, \quad \Pi_{\rm sym}-|\Om_d\ra\la \Om_d|,\quad \Pi_{\rm anti},$$
respectively.
 \end{proof}

  \begin{figure}
     \centering
     \includegraphics{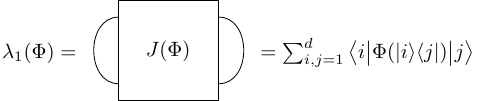} \\ \medskip
     \includegraphics{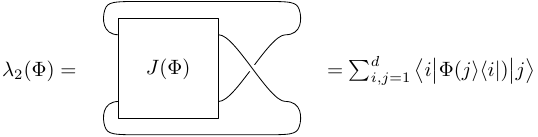}
     \caption{The quantities $\lambda_{1,2}(\Phi)$ associated to a quantum channel $\Phi : \M{d} \to \M{d}$.}
     \label{fig:lambda-1-2-Phi}
 \end{figure}

 The channel parameters $\lambda_{1,2}(\Phi)$ can take values in specific intervals, as follows. 

 \begin{proposition}
     For a channel $\Phi : \M{d} \to \M{d}$, we have
     \begin{align*}
        0 \leq &\lambda_1(\Phi) \leq d^2 \\
        -d \leq &\lambda_2(\Phi) \leq d,
     \end{align*}
     with the inequalities being tight. 
 \end{proposition}
 \begin{proof}
     The map $\Phi \mapsto \mathcal T_{U,U}(\Phi)$ is onto the set of channels of the form $ p \cdot \id + (1-p) \cdot \Delta$. The range of $p$ for which the latter maps are quantum channels is known to be 
     $$p \in \Big[ \frac{-1}{d^2-1}, 1 \Big],$$
     hence the claim for $\lambda_1(\Phi)$ follows. The second claim can be proven in a similar manner, since the allowed range of $q$ is
     $$q \in \Big[ \frac{-1}{d-1}, \frac{1}{d+1} \Big].$$
 \end{proof}

 The PPT and entanglement breaking properties of (conjugate) unitary and orthogonal covariant channels are fully characterized by the parameters $\lambda_{1,2}(\Phi)$, as follows. 

 \begin{proposition}[\cite{VW01}]\label{prop:conditions-PPT-EB-UUOO}
 For the parameters $\lambda_1(\Phi)$ and $\lambda_2(\Phi)$ defined above,
 \begin{enumerate}
     \item $\mathcal{T}_{U,U}(\Phi)$ is PPT iff it is EB iff $0\leq \lambda_1(\Phi)\leq d$;

     \item $\mathcal{T}_{U,\overline{U}}(\Phi)$ is PPT iff it is EB iff $0\leq \lambda_2(\Phi)\leq d$;

     \item $\mathcal{T}_{O,O}(\Phi)$ is PPT iff it is EB iff $0\leq \lambda_1(\Phi),\lambda_2(\Phi)\leq d$.
 \end{enumerate}
 \end{proposition}

 We define now random covariant quantum channels by twirling the random quantum channels introduced in Section \ref{sec:random-quantum-channels}. More precisely, we introduce the following image distributions: 
 \begin{align*}
     \mu^{U,U}_{d,s} &:= (\mathcal T_{U,U})_\# \mu_{d,s}\\
     \mu^{U,\overline U}_{d,s} &:= (\mathcal T_{U,\overline U})_\# \mu_{d,s}\\
     \mu^{O,O}_{d,s} &:= (\mathcal T_{O,O})_\# \mu_{d,s}.
 \end{align*}

Recall that $\mu_{d,s}$ is the probability distribution of the random channel $\Phi_V$ associated to the Haar random isometry $V:\C{d}\to \C{d}\otimes \C{s}$. The distributions above are characterized by the parameters $p,q,p',q'$ from Proposition \ref{prop:UUO-covariant}. In turn, these are related to the isometry $V$ defining the channel $\Phi$ via the following quantities: 
 \begin{align*}
     \lambda_1^{(d,s)} &:= \lambda_1(\Phi_V) = d\la \Om_d|J(\Phi_V)|\Om_d\ra \\
     \lambda_2^{(d,s)} &:= \lambda_2(\Phi_V) = \Tr(J(\Phi_V)F_d),
 \end{align*}
 see Figure \ref{fig:lambda-1-2} for a graphical representations of these scalars. Above, the random channel $\Phi_V$ is defined via the Stinespring dilation from the Haar-distributed random isometry $V : \C{d} \to \C{d} \otimes \C{s}$: 
 $$\Phi_V ( \rho ) = \Tr_s \big[ V\rho V^* \big].$$

 \begin{figure}
     \centering
     \includegraphics{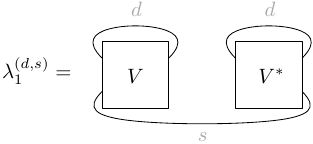} \qquad \qquad \includegraphics{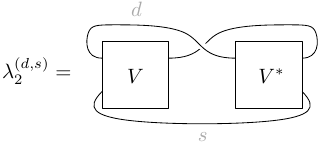}
     \caption{The quantities $\lambda_{1,2}^{(d,s)}$ associated to a Haar-distributed random isometry $V : \mathbb C^d \to \mathbb C^d \otimes \mathbb C^s$.}
     \label{fig:lambda-1-2}
 \end{figure}

 In order to understand the PPT and the entanglement breaking behavior of random (conjugate) unitary and orthogonal covariant channels, we need first to understand the large $d$ limit of the random variables $\lambda_{1,2}^{(d,s)}$. To start with, it follows from \cref{lem:moments-lambda-123} that the random variables $\lambda_{1,2}^{(d,s)}$ have average and variance given by
 $$\mathbb E \Big[ \lambda_{1,2}^{(d,s)} \Big] = 1 \qquad \text{and} \qquad \operatorname{Var}\Big[ \lambda_{1,2}^{(d,s)} \Big] =  \frac{s(d^2-1)}{(ds)^2-1}.$$
 Note that although $\lambda_{1,2}^{(d,s)}$ have the same first two moments, they do not have the same distribution, see \cref{lem:moments-lambda-123} for the exact moments of all orders. In fact, they have very different behavior as $d \to \infty$, as we shall see next. We first recall the definitions of the \emph{normal} (or \emph{Gaussian}) distribution of mean $\mu$ and variance $\sigma^2$:
 $$X \sim \mathcal N(m, \sigma^2) \qquad \iff \qquad \mathrm{d} X  = \frac{1}{\sqrt{2\pi \sigma^2}} \exp\left( \frac{-(x-m)^2}{2\sigma^2} \right) \mathrm{d} x,$$
 and that of the \emph{Gamma} distribution of parameters $(\alpha, \beta)$
 $$Y \sim \Gamma(\alpha, \beta) \qquad \iff \qquad \mathrm{d} Y  = \frac{\beta^\alpha}{\Gamma(\alpha)}y^{\alpha-1}\exp(-\beta y) \mathbf{1}_{y >0}\mathrm{d} y.$$

Note that the Gamma distribution of parameters $(1,\lambda)$ is also called the \textit{exponential} distribution of parameter $\lambda$, for which we denote by $\mathsf{Exp}(\lambda):=\Gamma(1,\lambda)$.

\begin{proposition} \label{prop-UUparaconv}
As $d \to \infty$, the random variables $\lambda_{1,2}^{(d,s)}$ behave as follows:

\begin{enumerate}
    \item If $s$ is fixed, then
    \begin{itemize}
        \item the law of the random variable $s\lambda_1^{(d,s)}$ converges in moments to  the {\rm Gamma distribution} $\Gamma(s,1)$
    
        \item the law of the random variable $s\lambda_2^{(d,s)}$ converges in moments to the {\rm normal distribution} $\mathcal{N}(s,s)$.
    \end{itemize}
    \item If $s=s(d) \to \infty$, then 
        $$\lambda_{1,2}^{(d,s)} \to 1 \text{ in $L^p$ for all $p<\infty$.}$$
    Furthermore, the same convergence holds almost surely if $s\gtrsim d^t$ for some $t>0$.
    \item Moreover, for any $t>0$, we have 
        $$d^{-t}(\lambda_1^{(d,s)}-1)\to 0 \qquad \text{ and } \qquad d^{-t}(\lambda_2^{(d,s)}-1)\to 0$$
    almost surely as $d\to \infty$ or $s\to \infty$ (more precisely, this means that ``$d\to \infty$ and $s=s(d)$ arbitrary'' or ``$s\to \infty$ and $d=d(s)$ arbitrary'').
    \end{enumerate}
    \end{proposition}
\begin{proof}
    See \cref{app:scalar}.
\end{proof}

We plot numerical realizations of the random variables $\lambda_{1,2}^{(d,s)}$ against the theoretical curves in \cref{fig:lambda-1-hist,fig:lambda-2-hist}.

\begin{figure}[htb]
    \centering
    \includegraphics[width=.45\textwidth]{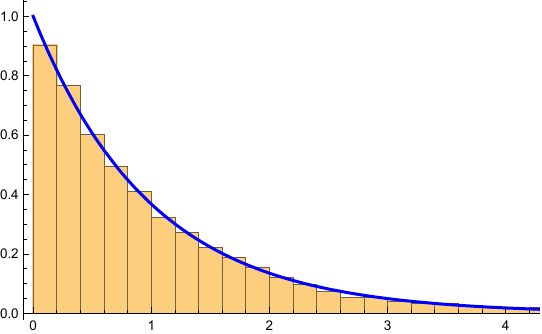} \qquad \includegraphics[width=.45\textwidth]{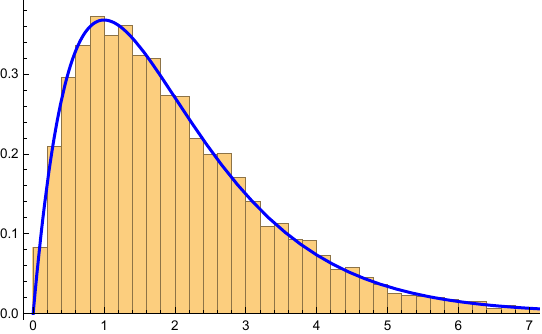} \\
    \includegraphics[width=.45\textwidth]{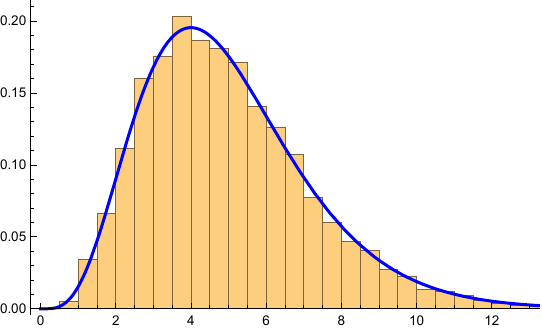} \qquad \includegraphics[width=.45\textwidth]{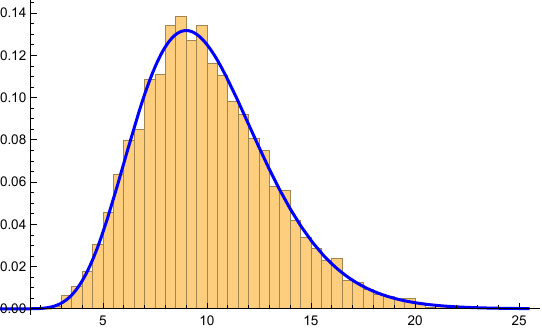}
    \caption{Histograms of $10^4$ realizations of the random variable $s \lambda_1^{(d,s)}$ for $d=100$ and, respectively, $s=1,2,5,10$. In blue, the probability density of the $\Gamma(s,1)$ distribution.}
    \label{fig:lambda-1-hist}
\end{figure}

\begin{figure}[htb]
    \centering
    \includegraphics[width=.45\textwidth]{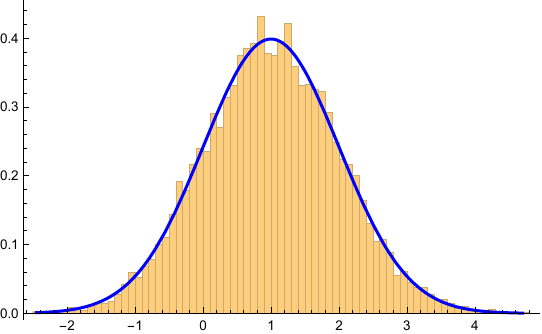} \qquad \includegraphics[width=.45\textwidth]{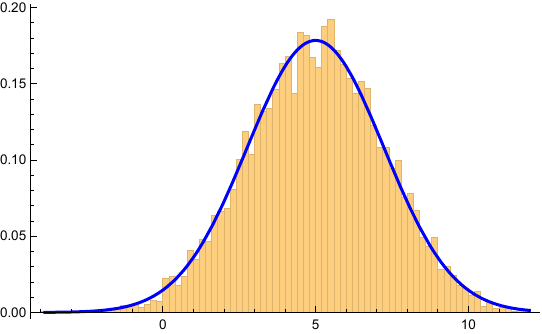}
    \caption{Histograms of $10^4$ realizations of the random variable $s \lambda_2^{(d,s)}$ for $d=100$ and, respectively, $s=1,5$. In blue, the probability density of the $\mathcal N(s,s)$ distribution.}
    \label{fig:lambda-2-hist}
\end{figure}

\begin{theorem} \label{thm-RanUUCovEB}
Let $\Phi_{U,U} \sim \mu^{U,U}_{d,s}$, $\Phi_{U,\overline{U}} \sim \mu^{U,\overline{U}}_{d,s}$, and $\Phi_{O,O}\sim \mu_{d,s}^{O,O}$ be, respectively, random $(U,U)$-covariant, $(U,\overline{U})$-covariant, and $(O,O)$-covariant quantum channels.
\begin{enumerate}
    \item Almost surely, as $d\to \infty$, $\Phi_{U,U}$ is PPT and EB.

    \item If $s$ is fixed, then
        $$\lim_{d \to \infty} \mathbb{P}(\Phi_{U,\overline{U}}\text{ is PPT $\Leftrightarrow$ EB}) =  \lim_{d \to \infty} \mathbb{P}(\Phi_{O,O}\text{ is PPT $\Leftrightarrow$ EB}) = \mathbb{P}(Z \geq -\sqrt{s})<1,$$
    where $Z \sim \mathcal{N}(0,1)$ is a standard normal random variable.
    \item If $s=s(d) \to \infty$ as $d\to \infty$, then 
        $$\lim_{d \to \infty} \mathbb{P}(\Phi_{U,\overline{U}}\text{ is PPT $\Leftrightarrow$ EB}) =  \lim_{d \to \infty} \mathbb{P}(\Phi_{O,O}\text{ is PPT $\Leftrightarrow$ EB}) = 1,$$
    Furthermore, if $s\gtrsim d^t$ for some $t>0$, then almost surely as $d\to \infty$, $\Phi_{U,\overline{U}}$ and $\Phi_{O,O}$  are PPT and EB.
\end{enumerate}
\end{theorem}
We remark that by the statement `$\Phi=\Phi_d$ is \textit{almost surely} PPT/EB as $d\to \infty$', we mean the eventual property of $\Phi_d$: there exists an event set $\Om$ of probability 1 such that for every $\omega\in \Om$, there exists a number $N=N(\omega)\geq 1$ such that $\Phi_d(\omega)$ is PPT/EB for all $d\geq N$. In particular, this implies that $\mathbb{P}(\Phi \text{ is PPT/EB})\to 1$ as $d\to \infty$.
\begin{proof}
We start with the first item, related to the unitary covariant channel $\Phi_{U,U}$ and the random variable $\lambda_1^{(d,s)} \geq 0$. It follows from \cref{prop-UUparaconv}, (3) with $t=1$, that, almost surely,
$$\lim_{d \to \infty} \frac{\lambda_1^{(d,s)}}{d} = 0,$$
which yields (1) by \cref{prop:conditions-PPT-EB-UUOO}.  

Let us now move to the second point, corresponding to the regime $s$ fixed and the channels $\Phi_{U, \overline U}$ and $\Phi_{O,O}$. From the previous point, we have that, almost surely as $d \to \infty$, $\lambda_1^{(d,s)} \leq d$; similarly for $\lambda_2^{(d,s)}$. The only constraint in \cref{prop:conditions-PPT-EB-UUOO} that we need to ensure is $\lambda_2^{(d,s)} \geq 0$, and the conclusion follows from \cref{prop-UUparaconv} (1): if $\tilde Z \sim \mathcal N(s,s)$, then 
$$\mathbb P (\tilde Z \geq 0) = \mathbb P( \sqrt{s}Z+s \geq 0 ) =  \mathbb P( Z \geq -\sqrt s),$$
for a standard normal random variable $Z \sim \mathcal N(0,1)$.

Finally, regarding the proof of (3), the only constraint again is whether $\lambda_2^{(d,s)}\geq 0$. Then the assertions follow directly from \cref{prop-UUparaconv} (2): for the general regime $s(d)\to \infty$, $\lambda_2^{(d,s)}\to 1$ in probability, so we have
    $$\lim_{d\to \infty}\mathbb{P}(\Phi_{O,O} \text{ is PPT})=\lim_{d\to \infty}\mathbb{P}(\lambda_2^{(d,s)}\geq 0)=1,$$
and analogously $\lim_{d\to \infty}\mathbb{P}(\Phi_{U,\overline{U}} \text{ is PPT})=1$. Furthermore, when $s\gtrsim d^t$ for some $t>0$, we have $\lambda_{1,2}^{(d,s)} \to 1$ hence $0 \leq \lambda_{1,2}^{(d,s)} \leq d$ almost surely. 
\end{proof}

\begin{remark}
    One can also show, by slightly modifying the previous arguments, that if $s \to \infty$, then, independently of the behavior of $d=d(s)$, all the three random channels $\Phi_{U,U}$, $\Phi_{U,\overline{U}}$, and $\Phi_{O,O}$ are almost surely PPT and entanglement breaking.
\end{remark}

In conclusion, we have shown that the equivalent PPT and EB properties for large ($d \to \infty$) random (conjugate) unitary and orthogonal covariant quantum channels are as follows: 
\begin{itemize}
    \item random unitary covariant channels are asymptotically PPT and EB (for any $s$)
    \item random conjugate unitary and orthogonal covariant channels admit a phase transition between the regimes $s$ fixed and $s \to \infty$.
\end{itemize}

\section{Random hyperoctahedral covariant channels}\label{sec:hyperoctahedral-covariant}

Here we consider a random behavior of the channels which are covariant with respect to \emph{hyperoctahedral group}. Recall that one can realize the hyperoctahedral group of degree $d$ as the subgroup $\mathcal{H}_d$ of the orthogonal group $\mathcal{O}_d$ which is generated by diagonal orthogonal matrices and permutation matrices. The group $\mathcal H_d$ is also known as the group of \emph{signed permutations} and it has order $2^d\cdot d!$. It is recently shown \cite[Proposition 4.4]{PJPY23} that a linear map $\Phi:\M{d}\to \M{d}$ is covariant with respect to the fundamental representation $\pi_{\mathcal{H}_d}:H\in \mathcal{H}_d\mapsto H\in \mathcal{O}_d$ if and only if
\begin{equation} \label{eq-HHCov}
    \Phi(Z)=p\frac{\rm Tr(Z)}{d}I_d+qZ+r Z^{\top}+s\, {\rm diag}(Z),\;\; Z\in \M{d},
\end{equation}
for some $p,q,r,s\in \C{}$. In other words, the vector space of hyperoctahedral covariant linear maps is spanned by three quantum channels (the completely depolarizing, the identity, and the diagonal conditional expectation) and the transposition map. Note that the formula above covers the (conjugate) unitary and orthogonal covariant quantum channels discussed in Section \ref{sec:conjugate-unitary-orthogonal-covariant}.

Let us simply denote ${\rm Cov}(H,H):={\rm Cov}(\pi_{\mathcal{H}_d},\pi_{\mathcal{H}_d})$ and ${\rm Inv}(H,H):={\rm Inv}(\pi_{\mathcal{H}_d}\otimes \pi_{\mathcal{H}_d})=J({\rm Cov}(H,H))$ (by \cref{prop-twirling}). 
Note that the Choi matrix of the map $\Phi$ given in \eqref{eq-HHCov} is
\begin{equation} \label{eq-HHInv}
    J(\Phi)=p\frac{I_{d^2}}{d}+dq|\Om_d\ra\la \Om_d|+rF_d+ s\Pi_{\rm diag}\in {\rm Inv}(H\otimes H),
\end{equation}
where $|\Om_d\ra=\frac{1}{\sqrt{d}}\sum_{i}|ii\ra$ and $F_d$ are defined in \cref{eq:def-max-ent-state,eq:def-flip}, and 
    $$\Pi_{\diag}:=\sum_{i=1}^d|ii\ra\la ii|$$ 
is the Choi matrix of the ``diagonalization'' (or completely dephasing) map $Z\mapsto \diag(Z)$. From this we can check that the space ${\rm Inv}(H,H)=J({\rm Cov}(H,H))$ is spanned by four mutually orthogonal projections
\begin{align*}
    \Pi_0&=|\Om_d\ra\la \Om_d|,\\
    \Pi_1&=\Pi_{\rm sym}-\Pi_{\rm diag}=\frac{I_{d^2}+F}{2}-\Pi_{\rm diag},\\
    \Pi_2&=\Pi_{\rm anti}=\frac{I_{d^2}-F}{2},\\
    \Pi_3&=\Pi_{\rm diag}-|\Om_d\ra\la \Om_d|.
\end{align*}
Therefore, Proposition \ref{prop-twirlformula} gives the exact formula of $\mathcal{T}_{H\otimes H}:=\mathcal{T}_{\pi_{\mathcal{H}_d}\otimes \pi_{\mathcal{H}_d}}$,
\begin{equation}\label{eq-HHtwirl}
    \mathcal{T}_{H\otimes H}X=\sum_{i=0}^3 {\Tr}(\Pi_iX)\frac{\Pi_i}{d_i},\;\; X\in \M{d}\otimes \M{d},
\end{equation}
where the dimensions read $d_0=1$, $d_1=d_2=\frac{d^2-d}{2}$, and $d_3=d-1$. We can now give the formula for the $(H,H)$-twirling of a quantum channel.

\begin{proposition}\label{prop:HH-covariant}
    Let $\Phi : \M{d} \to \M{d}$ be a quantum channel. Then, the $(H,H)$ twirling of $\Phi$ is given by: 
    $$\mathcal T_{H,H}(\Phi) = p \cdot \Delta + q \cdot \id + r \cdot \top + (1-p-q-r) \cdot \diag $$
    where $\id, \top, \operatorname{diag}, \Delta$ are, respectively, the identity, the transposition, the diagonalization, and the completely depolarizing linear maps. Here the coefficients $p,q,r$ are explicitly related to $\Phi$ by    
    \begin{align*}
        p&=\frac{d-\lambda_3(\Phi)}{d-1}\\
        q&=\frac{\lambda_1(\Phi)-\lambda_3(\Phi)}{d^2-d}\\
        r&=\frac{\lambda_2(\Phi)-\lambda_3(\Phi)}{d^2-d},
    \end{align*}
    where the parameters $\lambda_{1,2,3}(\Phi)$ associated to the channel are defined by (see also \cref{fig:lambda-1-2-Phi,fig:lambda-3}):
 \begin{align*}
     \lambda_1(\Phi) &= d\la \Om_d|J(\Phi)|\Om_d\ra \\
     \lambda_2(\Phi) &= \Tr(J(\Phi)F_d)\\ 
     \lambda_3(\Phi) &:= \Tr(J(\Phi)\Pi_{\diag}).
 \end{align*}
 \end{proposition}
 \begin{proof}
    By Proposition \ref{prop-twirling} and Equation \eqref{eq-HHtwirl}, we have
\begin{align*}
    J(\mathcal{T}_{H,H}\Phi)&=\mathcal{T}_{H\otimes H}(J(\Phi))=\sum_{i=0}^3 \Tr(\Pi_i J(\Phi))\frac{\Pi_i}{d_i}\\
    &=\frac{\lambda_1}{d}|\Om_d\ra\la \Om_d|+\left(\frac{d+\lambda_2}{2}-\lambda_3\right)\frac{I_{d^2}+F_d-2\Pi_{\diag}}{d^2-d}\\
     &\quad\quad\quad +\frac{d-\lambda_2}{2}\frac{I_{d^2}-F_d}{d^2-d}+\left(\lambda_3-\frac{\lambda_1}{d}\right)\frac{\Pi_{\diag}-|\Om_d\ra\la \Om_d|}{d-1}\\
    &=\frac{d-\lambda_3}{d-1}\frac{I_{d^2}}{d}+\frac{\lambda_1-\lambda_3}{d^2-d}d|\Om_d\ra\la \Om_d|+\frac{\lambda_2-\lambda_3}{d^2-d}F_d+\frac{(d+2)\lambda_3-\lambda_1-\lambda_2-d}{d^2-d}\Pi_{\diag}.
\end{align*}
Here we used the fact $\Tr(J(\Phi))=d$ since $\Phi$ is TP.
 \end{proof}

  \begin{figure}
     \centering
     \includegraphics{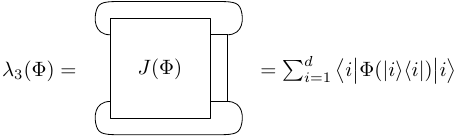} \qquad\qquad
     \includegraphics{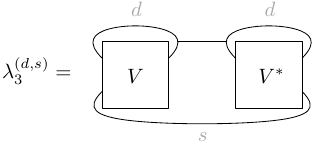}
     \caption{The quantity $\lambda_3(\Phi)$ associated to a quantum channel $\Phi : \M{d} \to \M{d}$ and the corresponding random variable $\lambda_3^{(d,s)} = \lambda_3(\Phi_V)$ for a random channel obtained from a Haar-distributed random isometry $V : \C{d} \to \C{d} \otimes \C{s}$.}
     \label{fig:lambda-3}
 \end{figure}

 One can now state the precise conditions for a $(H,H)$-twirling of a quantum channel to be PPT; the PPT property has been shown in \cite[Theorem 4.1]{PJPY23} to be equivalent to the entanglement breaking property for hyperoctahedral covariant channels.

\begin{proposition} \label{prop-HHCovEB}
Given a quantum channel $\Phi:\M{d} \to \M{d}$, its $(H,H)$-twirling $\mathcal{T}_{H,H}(\Phi)$ is PPT iff it is EB iff
\begin{align} 
    0&\leq \lambda_3(\Phi)\leq d,\nonumber\\
    0&\leq \lambda_1(\Phi),\lambda_2(\Phi) \leq d\lambda_3(\Phi), \label{eqeq:HH-PPT-EB}\\
    2\lambda_3(\Phi)-d&\leq \lambda_1(\Phi),\lambda_2(\Phi) \leq d. \nonumber
\end{align}
\end{proposition}
\begin{proof}
    The conditions in the statement are simply the following inequalities from \cite[Eq.~(4.17)]{PJPY23} reformulated in terms of the parameters $\lambda_{1,2,3}(\Phi)$: 
\begin{align*} 
    0&\leq \,\,\, p\,\,\,\leq \frac{d}{d-1},\\
    \frac{p}{d}-\frac{1}{d-1}&\leq q,r \leq 1-\frac{d-1}{d}p,\\
    -\frac{p}{d}&\leq q,r \leq \frac{p}{d}.
\end{align*}
\end{proof}

Now we consider a random Stinespring channel $\Phi:\M{d}\to \M{d}$ with environment dimension $s$ and define the \emph{random hyperoctahedral covariant channel} 
$$\Phi_{H,H} := \mathcal{T}_{H,H}(\Phi) \sim \mu^{H,H}_{d,s}.$$ 

Since the PPT (and therefore the EB) property of the twirled channel $\mathcal{T}_{H,H}(\Phi)$ depends only on the three parameters $\lambda_{1,2,3}(\Phi)$, we shall consider the corresponding random variables associated to the random channel: 
    $$\lambda_1^{(d,s)}=d\la \Om_d|J(\Phi)|\Om_d\ra,\quad \lambda_2^{(d,s)}=\Tr(J(\Phi)F_d), \quad \lambda_3^{(d,s)}:=\Tr(J(\Phi)\Pi_{\diag}).$$
The asymptotic behavior of the first two variables has been studied in the previous section, \cref{prop-UUparaconv}; we now turn to the study of $\lambda_3^{(d,s)}$. The average and the variance of this random variable are given by (\cref{lem:moments-lambda-123})
$$\mathbb E \Big[ \lambda_3^{(d,s)} \Big] = 1 \qquad \text{and} \qquad \operatorname{Var}\Big[ \lambda_3^{(d,s)} \Big] =  \frac{1}{ds+1}.$$

\begin{proposition}\label{prop:convergence-lambda3}
    The following almost sure convergence holds, independently of the scaling of the parameter $s$:
    $$\lim_{d \to \infty} \lambda_3^{(d,s)} = 1.$$
\end{proposition} 
\begin{proof}
    See \cref{app:scalar}.
\end{proof}

We can now prove the main result of this section, giving the PPT and EB properties of a hyperoctahedral covariant random channel, in the $d \to \infty$ regime.
    
\begin{theorem} \label{thm-RanHHCovEB}
Let $\Phi_{H,H} \sim \mu^{H,H}_{d,s}$ be a random $(H,H)$-covariant quantum channel. Then:
\begin{enumerate}
    \item If $s$ is fixed, 
        $$\lim_{d \to \infty} \mathbb{P}(\Phi_{H,H}\text{ is PPT $\Leftrightarrow$ EB}) = \mathbb{P}(Z \geq -\sqrt{s})<1,$$
    where $Z \sim \mathcal{N}(0,1)$ is a standard normal random variable.
    \item If $s=s(d) \to \infty$ as $d\to \infty$, then 
        $$\lim_{d \to \infty} \mathbb{P}(\Phi_{H,H}\text{ is PPT $\Leftrightarrow$ EB}) = 1.$$
    Furthermore, if $s\gtrsim d^t$ for some $t>0$, then almost surely as $d\to \infty$, $\Phi_{H,H}$ is PPT and EB. 
\end{enumerate}
\end{theorem}
\begin{proof}
Since $\lambda_3^{(d,s)} \to 1$ as $d \to \infty$, the equations \eqref{eqeq:HH-PPT-EB} reduce to the conditions from \cref{prop:conditions-PPT-EB-UUOO} (3). The proof is now the same as \cref{thm-RanUUCovEB} (2) and (3).
\end{proof}

Note that
    $$\mathcal{T}_{U,\overline{U}}\circ \mathcal{T}_{H,H}=\mathcal{T}_{U,\overline{U}}, \quad \mathcal{T}_{U,U}\circ \mathcal{T}_{H,H}=\mathcal{T}_{U,U}, \quad \text{and} \quad \mathcal{T}_{O,O}\circ \mathcal{T}_{H,H}=\mathcal{T}_{O,O},$$
by \cref{prop-twirlcomposition} and \cref{prop-twirling}. Since the twirling operations preserve both PPT and entanglement breaking properties, \cref{thm-RanHHCovEB} directly implies the asymptotic PPT and EB behavior of random depolarizing / transpose depolarizing / orthogonal covariant channels, in the regime where $s \gtrsim d^t$ with $d \to \infty$, recovering part of \cref{thm-RanUUCovEB}.

\begin{corollary} \label{cor-RanOOCovEB}
Suppose $s\gtrsim d^t$ for some constants $t>0$. Then almost surely, the random depolarizing channel $\mathcal{T}_{U,\overline{U}}\Phi$, the random transpose depolarizing channel $\mathcal{T}_{U,U}\Phi$, and the random orthogonal covariant channel $\mathcal{T}_{O,O}\Phi$ are PPT and EB, as $d\to \infty$.
\end{corollary}

In conclusion, we have shown that the equivalent PPT and EB properties for large ($d \to \infty$) random hyperoctahedral covariant quantum channels admit a phase transition between the regimes $s$ fixed and $s \to \infty$, which is completely parallel with random conjugate unitary and orthogonal covariant channels.

\section{Random diagonal orthogonal covariant channels}\label{sec:DOC}

In this section, we consider random quantum channels which are covariant with respect to the actions of the \emph{diagonal unitary group}
$$\mathcal{DU}_d := \big\{ \diag(z) \, : \, z \in \C{d}, \, |z_i| =1 \quad \forall i \in [d]\big\}$$
and the \emph{diagonal orthogonal group}
$$\mathcal{DO}_d := \big\{ \diag(\epsilon) \, : \, \epsilon \in  \{\pm 1\}^d \big\}.$$

Bipartite quantum states that are invariant with respect to the actions of the groups above, as well as channels that are covariant for these actions have appeared in the literature in the last 20 years \cite{DVSS+00,chruscinski2006class,liu2015unitary,lopes2015generic, BSGS22}, but have been studied in detail in \cite{singh2021diagonal}. Since, they have found many applications, see e.g.~\cite{singh2022diagonal,singh2022ppt,singh2022ergodic}. We recall from \cite[Definition 6.3]{singh2021diagonal} that a quantum channel $\Phi:\M{d} \to \M{d}$ is called: 
\begin{itemize}
    \item \emph{conjugate diagonal unitary covariant} ({CDUC}) if 
    $$ \forall Z \in \M{d} \text{ and } U \in \mathcal{DU}_d: \qquad \Phi(UZU^*) = U^*\Phi(Z)U,$$
    \item \emph{diagonal unitary covariant} ({DUC}) if 
    $$ \forall Z \in \M{d} \text{ and } U \in \mathcal{DU}_d: \qquad \Phi(UZU^*) = U\Phi(Z)U^*,$$
    \item \emph{diagonal orthogonal covariant} ({DOC}) if 
    $$ \forall Z \in \M{d} \text{ and } O \in \mathcal{DO}_d: \qquad \Phi(OZO) = O\Phi(Z)O.$$
\end{itemize}
Note that we employ a different naming convention than the one in \cite{singh2021diagonal}, by interchanging the DUC and CDUC terms. In the present work, the emphasis is on the covariance property, hence we prefer to use the ``conjugate'' adjective for the case where $\pi_B = \overline{\pi_A}$.

The structure of linear maps satisfying the covariance conditions above has been obtained in \cite[Eqs.~(36)-(38)]{singh2021diagonal}: a linear map $\Phi:\M{d} \to \M{d}$ is 
\begin{itemize}
    \item {CDUC} if there exists a pair $A,C \in \M{d}$ such that $\diag A = \diag C$ and 
    $$ \forall Z \in \M{d} \qquad \Phi(Z) = \Phi^{(1)}_{A,C}(Z) = \diag(A\ket{\diag(Z)}) + \ring{C}\odot Z^\top.$$
    \item {DUC} if there exists a pair $A,B \in \M{d}$ such that $\diag A = \diag B$ and 
    $$ \forall Z \in \M{d} \qquad \Phi(Z) = \Phi^{(2)}_{A,B}(Z) = \diag(A\ket{\diag(Z)}) + \ring{B}\odot Z.$$
    \item {DOC} if there exists a triple $A,B,C \in \M{d}$ such that $\diag A = \diag B = \diag C$ and 
    $$ \forall Z \in \M{d} \qquad \Phi(Z) = \Phi^{(3)}_{A,B,C}(Z) = \diag(A\ket{\diag(Z)}) + \ring{B}\odot Z + \ring{C}\odot Z^\top.$$
\end{itemize}

Above, the notation $\ring X$, for a matrix $X \in \M{d}$, denotes the matrix $X$ without its diagonal, i.e.
$$\ring X := X - \diag X \qquad \text{ or } \qquad (\ring X)_{ij} = \begin{cases}
    X_{ij} \qquad&\text{ if } i \neq j\\
    0 \qquad&\text{ if } i=j.
\end{cases}$$

Note that the maps $\Phi^{(1,2)}$ can be seen as special cases of the \textsf{DOC} map $\Phi^{(3)}$ by setting, respectively, $B = \diag B$ and $C = \diag C$. A linear \textsf{DOC} map $\Phi^{(3)}_{A,B,C}$ is completely positive (CP) iff 
\begin{itemize}
    \item $A$ is entrywise nonnegative: $A_{ij} \geq 0$ for all $i,j \in [d]$
    \item $B$ is positive semidefinite: $B \geq 0$
    \item $C$ is selfadjoint and $|C_{ij}|^2 \leq A_{ij} A_{ji}$ for all $i,j \in [d]$. 
\end{itemize}
A linear \textsf{DOC} map $\Phi^{(3)}_{A,B,C}$ is trace preserving (TP) iff $A$ is \emph{column stochastic}, i.e.
$$\forall j\in [d] \qquad \sum_{i=1}^d A_{ij}=1.$$

It is important to note at this point that the sets of {CDUC}, {DUC}, and {DOC} linear maps or quantum channels are parametrized by a pair or a triple of $d \times d$ matrices satisfying certain conditions; the number of parameters in these models is growing with the dimension $d$. This fact is in stark contrast with the symmetries investigated in \cref{sec:conjugate-unitary-orthogonal-covariant} (full-unitary and full-orthogonal) and in \cref{sec:hyperoctahedral-covariant} (hyperoctahedral). It is the richness of these models that allows for the existence of PPT symmetric channels that are not entanglement breaking, as we show next. Let us first discuss the PPT property for diagonal symmetric channels. 

\begin{proposition}[\cite{singh2021diagonal}]\label{prop:DOC-PPT}
    For any triple of matrices $A,B,C \in \M{d}$ having the same diagonal, we have
    $$\Phi^{(1)}_{A,C} \circ \top = \Phi^{(2)}_{A,C}, \qquad \Phi^{(2)}_{A,B} \circ \top = \Phi^{(1)}_{A,B}, \qquad \Phi^{(3)}_{A,B,C} \circ \top = \Phi^{(3)}_{A,C,B}.$$
    Hence, we have the following characterizations of the PPT property: 
    \begin{itemize}
        \item A {CDUC} quantum channel $\Phi^{(1)}_{A,C}$ is PPT iff $C$ is positive semidefinite ($C \geq 0$).
        \item A {DUC} quantum channel $\Phi^{(2)}_{A,B}$ is PPT iff $|B_{ij}|^2 \leq A_{ij}A_{ji}$ for all $i,j\in [d]$.
        \item A {DOC} quantum channel $\Phi^{(3)}_{A,B,C}$ is PPT iff $C$ is positive semidefinite ($C \geq 0$) and $|B_{ij}|^2 \leq A_{ij}A_{ji}$ for all $i,j\in [d]$.        
    \end{itemize}
\end{proposition}

To characterize the entanglement breaking property. we recall the following two important definitions from \cite{johnston2019pairwise} and \cite{nechita2021graphical}:

\begin{itemize}
    \item A pair of matrices $A,B \in \M{d}$ having the same diagonal is said to be \emph{pairwise completely positive} (\textsf{PCP}) if there exist matrices $V,W \in \mathcal{M}_{d\times k}(\mathbb{C})$, for some $k \in \mathbb{N}$, such that the following decomposition holds:
    $$A = (V \odot \overbar{V})(W \odot \overbar{W})^*, \qquad B = (V \odot W) (V \odot W)^*.$$
    
    \item A triple of matrices $A,B,C \in \M{d}$ having the same diagonal is said to be \emph{triplewise completely positive} (\textsf{TCP}) if there exist matrices $V,W \in \mathcal{M}_{d \times k}(\mathbb{C})$, for some $k \in \mathbb{N}$, such that the following decomposition holds:
    $$A = (V \odot \overbar{V})(W \odot \overbar{W})^*, \qquad B = (V \odot W) (V \odot W)^*, \qquad C = (V \odot \overbar{W})(V \odot \overbar{W})^*.$$
\end{itemize}

\begin{proposition}[\cite{singh2021diagonal}]\label{prop:DOC-EB}
    For any triple of matrices $A,B,C \in \M{d}$ having the same diagonal, we have
    \begin{itemize}
        \item A CDUC quantum channel $\Phi^{(1)}_{A,C}$ is EB iff the pair $(A,C)$ is pairwise completely positive (\textsf{PCP}).
        \item A DUC quantum channel $\Phi^{(2)}_{A,B}$ is EB iff the pair $(A,B)$ is pairwise completely positive (\textsf{PCP}).
        \item A DOC quantum channel $\Phi^{(3)}_{A,B,C}$ is EB iff the triple $(A,B,C)$ is triplewise completely positive (\textsf{TCP}).    
    \end{itemize}
\end{proposition}

Importantly, for the diagonal covariant channels discussed in this section, the EB property is not equivalent to the PPT property: there exist PPT channels that are not entanglement breaking. Such examples can be found in a simple way starting from the PPT entangled states of size $3 \times 3$ in \cite{kiem2011existence,kye2012classification} or from nonnegative and positive semidefinite matrices of size $d \geq 5$ which are not completely positive \cite{berman2003completely,tura2018separability,yu2016separability,singh2021diagonal}.

\bigskip

Let us now provide the formulas for the diagonal twirlings of a quantum channels. We would like to point out again that our naming convention here is different from the one in \cite{singh2021diagonal}, in the sense that the terms DUC and CDUC are swapped.

\begin{proposition}[{{\cite[Section 6]{nechita2021graphical}}}]
    Let $\Phi:\M{d} \to \M{d}$ be a linear map. Then, its twirlings with respect to the diagonal groups $\mathcal{DU}_d$ and $\mathcal{DO}_d$ are given by
    \begin{align*}
        \mathcal{T}_{\CDUC}\Phi &:=\mathcal T_{\mathcal{DU}_d,\overline{\mathcal{DU}}_d} \Phi = \Phi^{(1)}_{A(\Phi), C(\Phi)}\\
        \mathcal{T}_{\DUC}\Phi &:= \mathcal T_{\mathcal{DU}_d, \mathcal{DU}_d} \Phi = \Phi^{(2)}_{A(\Phi), B(\Phi)}\\
        \mathcal{T}_{\DOC}\Phi &:=\mathcal T_{\mathcal{DO}_d,\mathcal{DO}_d} \Phi = \Phi^{(3)}_{A(\Phi), B(\Phi),C(\Phi)},
    \end{align*}
    where the matrices $A,B,C$ are defined by (See \cref{fig:ABC-Phi} for a graphical representation):
     $$A(\Phi)_{ij}=\la ij|J(\Phi)|ij\ra, \quad B(\Phi)_{ij}=\la ii|J(\Phi)|jj\ra, \quad \text{and} \quad C(\Phi)_{ij}=\la ij|J(\Phi)|ji\ra \qquad \forall i,j \in [d].$$
\end{proposition}

\begin{figure}
    \centering
    \includegraphics[width=\textwidth]{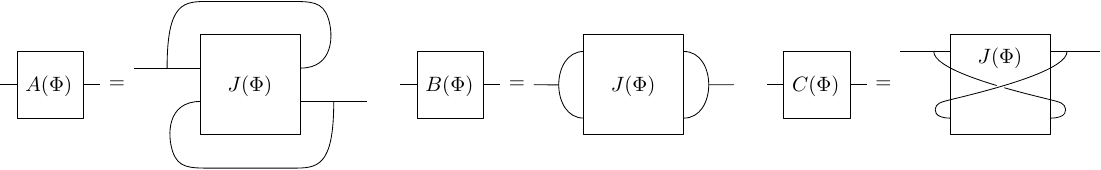}
    \caption{The three matrices $A,B,C$ associated to a linear map $\Phi$ through its Choi matrix $J(\Phi)$.}
    \label{fig:ABC-Phi}
\end{figure}

\bigskip

Let us now consider random diagonal covariant channels. To this end, let $\Phi_V:\M{d}\to \M{d}$ be a random Stinespring channel with environment dimension $s$ and define the following classes of random symmetric channels: 
\begin{itemize}
    \item \emph{random CDUC channels} 
$$\Phi_\CDUC := \mathcal{T}_\CDUC(\Phi_V) \sim \mu^\CDUC_{d,s}.$$
    \item \emph{random DUC channels} 
$$\Phi_\DUC := \mathcal{T}_\DUC(\Phi_V) \sim \mu^\DUC_{d,s}.$$
    \item \emph{random DOC channels} 
$$\Phi_\DOC := \mathcal{T}_\DOC(\Phi_V) \sim \mu^\DOC_{d,s}.$$
\end{itemize}

Since the PPT and EB properties of the random channels we consider here are determined by the associated $A,B,C$ matrices, we shall study the corresponding random matrices, derived from the Haar-distributed random isometry $V:\C{d} \to \C{d} \otimes \C{s}$:
\begin{equation} \label{eq-DOCentry}
    A=\sum_{k=1}^s V^{(k)}\odot \overline{V^{(k)}}, \quad B=\sum_{k=1}^s |{\rm diag\,}V^{(k)}\ra \la {\rm diag\,}V^{(k)}|,\quad C=\sum_{k=1}^s V^{(k)}\odot (V^{(k)})^*,
\end{equation}
where $V$ is represented in the block matrix form $V=\begin{pmatrix} V^{(1)} \\ \vdots \\ V^{(s)}\end{pmatrix}$ for $V^{(k)}\in \M{d}$. Moreover, we can draw the diagrams of $A,B,C$ as in \cref{fig-DOCpara}. Note that these $A,B,C \in \M{d}$ and their distribution depends on the integer parameter $s \geq 1$.

\begin{figure}[htb!] 
    \centering
    \includegraphics[width=\textwidth]{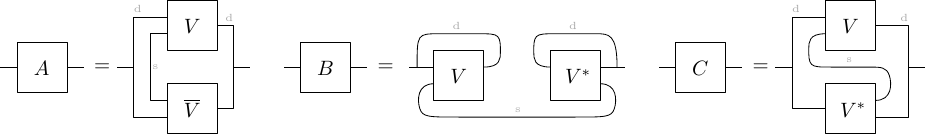}
    \caption{The three random matrix parameters $A$, $B$, and $C$ for random $\CDUC$, $\DUC$, and $\DOC$ channels, defined via a Haar-distributed random isometry $V:\C{d} \to \C{d} \otimes \C{s}$.}
    \label{fig-DOCpara}
\end{figure}

\begin{remark}
    It is useful to connect the random parameters $A,B,C$ discussed in this section to the scalars $\lambda_{1,2,3}^{(d,s)}$ from \cref{sec:conjugate-unitary-orthogonal-covariant,sec:hyperoctahedral-covariant}; graphically, compare \cref{fig:lambda-1-2,fig:lambda-3} with \cref{fig-DOCpara}. We have the following relations:
    \begin{align}
        &\lambda_1(\Phi)=\sum_{i,j=1}^d \la i|B(\Phi)|j\ra = \Tr(J_d B(\Phi)),\nonumber\\
        &\lambda_2(\Phi)=\sum_{i,j=1}^d \la i|C(\Phi)|j\ra = \Tr(J_d C(\Phi)), \nonumber\\ &\lambda_3(\Phi)=\Tr A(\Phi)=\Tr B(\Phi) = \Tr C(\Phi), \label{eq-C-lambda3}
    \end{align}
where $J_d\in \M{d}$ is the matrix in which all entries are equal to 1.
    
\end{remark}

As in the previous two sections, we study now the asymptotic behavior of three random matrix parameters $A,B,C$ in the regime $d \to \infty$. The following results will be used to establish, later in this section, the PPT and the EB properties of random CDUC, DUC, and DOC channels. Since both the entrywise and the spectral distributions of the matrices $A,B,C$ are needed (see e.g.~\cref{prop:DOC-PPT}), the analysis in this section is more involved. We defer the technical proofs of the following three results, \cref{prop-ranDOCmatrix,prop:DOC-entry-limit-fixed,prop:DOC-entry-limit}, to \cref{app:matrix}. 

We start by giving the large-$d$ limit spectrum of the self-adjoint matrix $C \in \Msa{d}$, whose positive semi-definiteness is needed for PPT property of DOC channels. We recall that the \emph{semicircular distribution} with average $m$ and variance $\sigma^2$ is given by 
$$\mathrm{d} \mathrm{SC}_{m,\sigma} = \frac{\sqrt{4\sigma^2 - (x-m)^2}}{2 \pi \sigma}  \mathbf{1}_{[-2\sigma+m, 2\sigma+m]}(x) \, \mathrm{d}x.$$

\begin{proposition} \label{prop-ranDOCmatrix}
The asymptotic spectral behavior of  the random matrix $C$, as $d\to \infty$, is given as follows:
\begin{enumerate}
    \item In the regime $s/d\to 0$,
    \begin{itemize}
        \item in probability, the empirical spectral distribution of $(ds)^{1/2}C$ converges, in moment, to a semicircular distribution with mean $0$ and variance 1,
        \item almost surely, we have the following convergence of the first three moments of $(ds)^{1/2}C$:
            $$\frac{1}{d}\Tr((ds)^{1/2}C)\to 0, \quad \frac{1}{d}\Tr(ds C^2)\to 1, \quad \frac{1}{d}\Tr((ds)^{3/2}C^3)\to 0.$$
    \end{itemize}
    \item In the regime $s\sim cd$ for a constant $c>0$, the empirical spectral distribution of $sC$ converges, in moments, to a semicircular distribution with mean $c$ and variance $c$ almost surely.
    \item In the regime $s/d\to \infty$, the empirical spectral distribution of $dC$ converges, in moments, to the probability mass $\delta_1$ almost surely.
\end{enumerate}
\end{proposition}
\begin{proof}
See \cref{app:matrix}.
\end{proof}

Under some slightly stronger hypothesis, the convergence in moments in item (3) above can be upgraded to a convergence in operator norm. 

\begin{proposition} \label{prop-ranDOCmatrix2}
When $s\gtrsim d^{1+t}$ for some constant $t>0$, $\|dC-I_d\|_\infty \to 0$ almost surely as $d \to \infty$. In particular, the minimum eigenvalue $\lambda_{\min}(dC)$ of $dC$ converges to 1 almost surely as $d\to\infty$.
\end{proposition}
\begin{proof}
See \cref{app:matrix}.
\end{proof}

\begin{remark}
Although it is not needed in our work, we can show the convergence of matrix $B$ in various regimes, using the same method. For example, in the regime $s\sim cd$ and $d\to \infty$, the the empirical spectral distribution of $sB$ converges, almost surely, to a \textit{Mar\u{c}enko-Pastur distribution} of parameter $c$: the distribution ${\rm MP}_c$ given by
\begin{equation*}
\mathrm{d}\mathrm{MP}_c=\max (1-c,0)\delta_0+\frac{\sqrt{(b-x)(x-a)}}{2\pi x} \; \mathbf{1}_{[a,b]}(x) \, \mathrm{d}x,
\end{equation*}
with $a = (1-\sqrt c)^2$ and $b=(1+\sqrt c)^2$.
\end{remark}

We now turn to the distribution of the elements of the matrices $A,B,C$.

\begin{proposition} \label{prop:DOC-entry-limit-fixed}
When $s$ is fixed, we have as $d\to \infty$,
\begin{enumerate}
    \item the law of $dsA_{ij}$ converges, in moments, to the Gamma distribution $\Gamma(s,1)$,
    
    \item both the laws of $|dsB_{ij}|^2$ and $|dsC_{ij}|^2$ converge, in moments, to that of the product of two independent random variables $X\sim \Gamma(s,1)$ and $Y\sim \mathsf{Exp}(1)$,

    \item for any fixed $N\leq d/2$ and pairwise distinct indices $i_1,\ldots, i_N, j_1,\ldots, j_N\in [d]$, the $3N$ random variables 
        $$dsA_{i_1 j_1}, dsA_{j_1 i_1}, |dsB_{i_1 j_1}|^2, dsA_{i_2 j_2}, dsA_{j_2 i_2}, |dsB_{i_2 j_2}|^2, \ldots, dsA_{i_N j_N}, dsA_{j_N i_N}, |dsB_{i_N j_N}|^2$$
    are asymptotically independent.
\end{enumerate}
\end{proposition}
\begin{proof}
See \cref{app:matrix}.
\end{proof}

\begin{proposition}\label{prop:DOC-entry-limit}
When $s\gtrsim d^t$ for some constant $t>0$, we have the following almost sure convergences as $d\to \infty$:
\begin{enumerate}
    \item $\displaystyle\max_{i,j} |dA_{ij}-1|\to 0$ (in particular, $\displaystyle\max_{i,j}dA_{ij}\to 1$ and $\displaystyle\min_{i,j}dA_{ij}\to 1$).
    \item $\displaystyle\max_{i\neq j} d^{1+t'/2}|B_{ij}|\to 0$ whenever $t'<t$.
    \item $\displaystyle\max_{i\neq j} d^{1+t'/2}|C_{ij}|\to 0$ whenever $t'<t$.
\end{enumerate}
\end{proposition}
\begin{proof}
See \cref{app:matrix}.    
\end{proof}

The previous three results allow us to establish the PPT property for diagonal unitary and orthogonal covariant random channels. Note the qualitatively different behaviors for the DUC case versus the CDUC and DOC symmetries.

\begin{theorem}\label{thm:DOC-PPT}
Let $\Phi_\DUC \sim \mu^\DUC_{d,s}$ be a random DUC channel. Then:
\begin{enumerate}
    \item In the regime $s$ is fixed and $d\to \infty$, we have
        $$\lim_{d\to \infty}\mathbb{P}(\Phi_{\DUC} \text{ is PPT })=0.$$

    \item In the regime $s\gtrsim d^t$ for some constant $t>0$, $\Phi_\DUC$ is almost surely PPT as $d\to \infty$.

\end{enumerate}    

\smallskip

   \noindent Let $\Phi_\CDUC \sim \mu^\CDUC_{d,s}$, and $\Phi_\DOC \sim \mu^\DOC_{d,s}$ be, respectively, random CDUC and DOC channels. Then:
 \begin{enumerate}   
 \setcounter{enumi}{2}
    \item In the regime $d\to \infty$ and $s/d\to 0$, almost surely, neither $\Phi_{\CDUC}$ nor $\Phi_{\DOC}$ are PPT.
    
    \item Consider the regime $s\sim cd$. If $c<4$, then almost surely as $d\to \infty$, neither $\Phi_\CDUC$ nor $\Phi_\DOC$ are PPT. If $c>4$, then almost surely as $d \to \infty$, both $\Phi_\CDUC$ and $\Phi_\DOC$ have Choi matrices whose partial transpositions converge,  in moments, to a probability measure supported on the positive real line.

    \item In the regime $s\gtrsim d^{1+t}$ for some constant $t>0$, both $\Phi_\CDUC$ and $\Phi_\DOC$ are almost surely PPT as $d\to \infty$.
\end{enumerate}
\end{theorem}

\begin{proof}
    Let us start with the first and second points. By \cref{prop:DOC-PPT}, we know that 
        $$\Phi_\DUC \text{ is PPT } \iff A_{ij}A_{ji}\geq|B_{ij}|^2 \;\; \forall\, i\neq j.$$
    If $s$ is fixed, then by \cref{prop:DOC-entry-limit-fixed}, the limit probability can be estimated, for any fixed $N\geq 1$, as follows:
    \begin{align*}
        \limsup_{d\to \infty}\mathbb{P}(\Phi_{\DUC}\text{ is PPT })&\leq \limsup_{d\to \infty}\mathbb{P}(A_{2n-1, 2n}A_{2n, 2n-1}\geq |B_{2n-1, 2n}|^2\;\forall\, n=1,\ldots,N)\\
        &=\prod_{n=1}^N \lim_{d\to\infty} \mathbb{P}((dsA_{2n-1, 2n})\cdot (dsA_{2n, 2n-1})\geq |dsB_{2n-1, 2n}|^2)\\
        &=\left(\mathbb{P}(X_1 X_2\geq X_3 Y)\right)^N,
    \end{align*}
    where $X_1,X_2,X_3,Y$ are four independent random variables such that $X_1,X_2,X_3\sim \Gamma(s,1)$ and $Y\sim \mathsf{Exp}(1)$. Moreover, the last equation can be more explicitly written as an integral
        $$\mathbb{P}(X_1 X_2\geq X_3 Y)=\int_{\mathbb{R}^4}\mathbf{1}_{\big\{\substack{x_1x_2\geq x_3y, \\ x_1,x_2,x_3,y\geq 0}\big\}}\frac{(x_1x_2x_3)^{s-1}}{((s-1)!)^3}e^{-(x_1+x_2+x_3+y)}dx_1\, dx_2\, dx_3\, dy<1.$$
    Therefore, we have (1) by taking $N\to \infty$. When $s\gtrsim d^t$ for some $t>0$, we can show that $d^2\min_{i\neq j}(A_{ij}A_{ji}-|B_{ij}|^2)\to 1$, almost surely, from the inequality
        $$\left(\min_{i,j}dA_{ij}\right)^2-\left(\max_{i\neq j}d|B_{ij}|\right)^2\leq d^2\min_{i\neq j}(A_{ij}A_{ji}-|B_{ij}|^2)\leq \left(\max_{i,j}dA_{ij}\right)^2$$
    and Proposition \ref{prop:DOC-entry-limit}, and hence (2) is shown.   

    \medskip 

    Now let us focus on (3) - (5), namely the PPT behaviors of CDUC and DOC channels. Note that for both two channels, their PPT property requires the condition $C\geq 0$, and in particular, this implies that
        $$\Tr C \Tr C^3\geq (\Tr C^2)^2$$
    by Cauchy-Schwarz inequality. However, when $s/d\to 0$ as $d\to\infty$, we have by \cref{prop-ranDOCmatrix} (1)
        $$s^2\left(\Tr C \Tr C^3-(\Tr C^2)^2\right)=\frac{1}{d}\Tr((ds)^{1/2}C)\cdot \frac{1}{d}\Tr((ds)^{3/2}C^3)-\left(\frac{1}{d}\Tr(dsC^2)\right)^2\to -1<0$$
    almost surely, which proves (3). On the other hand, when $s\gtrsim d$ we again have
    \begin{center}
        $d^2\min_{i\neq j}(A_{ij}A_{ji}-|B_{ij}|^2)\to 1$ almost surely.
    \end{center}
    Therefore, it is enough to consider the condition $C\geq 0$ for PPT properties by \cref{prop:DOC-PPT}. Now \cref{prop-ranDOCmatrix} (2) and (3), combined with \cref{prop:DOC-entry-limit}, and \ref{prop-ranDOCmatrix2}, imply the assertions (4) and (5), respectively.
\end{proof}

\begin{remark}
Note that in the case of CDUC and DOC channels, in the regime $d \to \infty$, $s \sim c d$ where $c > 4$, we cannot prove exactly that $\Phi_\CDUC$ and $\Phi_\DOC$ are asymptotically PPT. This is due to the fact that the absence
of negative support of the limiting measure of the partial transpositions of the respective Choi matrices does not exclude the absence of negative outliers. To reach the desired conclusion, one may show the \emph{strong convergence} \cite{haagerup2005new} of the random matrices to their corresponding limit. For the partial transpositions of Wishart random matrices, Aubrun \cite{Aub12} computed very large moments of the random matrices in order to show the norm convergence, as we do in the proof of \cref{prop-ranDOCmatrix2} for the last assertion in \cref{thm:DOC-PPT}.
\end{remark}

The next result describes the asymptotic behavior of \emph{computable cross norm} or \emph{realignment criterion} \cite{chen2003matrix,rudolph2005further} for (the Choi matrix of) random DOC channels. The realignment criterion is known to be, in general, \emph{incomparable} to the PPT criterion. Its form for diagonal invariant bipartite matrices was derived in \cite[Lemma 2.12]{singh2021diagonal}, namely
    $$\Phi_{A,B,C}^{(3)} \text{ is entanglement breaking} \implies \sum_{i,j=1}^d A_{ij}-\|A\|_1\geq \sum_{i\neq j=1}^d \max(|B_{ij}|,|C_{ij}|).$$
The result below can be interpreted as the fact the realignment criterion is \emph{weaker} than the PPT criterion for random DOC channels, from a random perspective. Note that the same phenomenon was proven for random bipartite quantum states in \cite{aubrun2012realigning}.

\begin{proposition}\label{prop:DOC-realignment}
    In every regime $s\gtrsim d^t$ for some constants $t>0$ and $d\to \infty$ , random DOC channels satisfy the realignment criterion almost surely,
    $$\sum_{i,j=1}^d A_{ij}-\|A\|_1\geq \sum_{i\neq j=1}^d \max(|B_{ij}|,|C_{ij}|).$$
The same conclusion holds for random CDUC and DUC channels.
\end{proposition} 
\begin{proof}
This follows from the three almost sure limits:
    $$\frac{1}{d}\sum_{i,j}A_{ij}\to 1,\quad \frac{1}{d}\|A\|_1\to 0, \quad \text{and} \quad \frac{1}{d}\sum_{i\neq j}\max(|B_{ij}|,|C_{ij}|)\to 0.$$
The limits above follow from  \cref{prop:DOC-entry-limit} and the following three inequalities:
\begin{align*}
    &\min_{i,j} dA_{ij}\leq \frac{1}{d}\sum_{i,j} A_{ij}\leq \max_{i,j}dA_{ij},\\
    &\frac{1}{d}\|A\|_1\leq \frac{1}{d}\|I_d\|_2\,\|A\|_2=\frac{1}{\sqrt{d}}\bigg(\sum_{i,j}A_{ij}^2\bigg)^{1/2}\leq \frac{1}{\sqrt{d}}(\max_{i,j}dA_{ij}),\\
    &\frac{1}{d}\sum_{i\neq j}\max(|B_{ij}|,|C_{ij}|)\leq \frac{1}{d}\sum_{i\neq j}(|B_{ij}|+|C_{ij}|)\leq (d-1)(\max_{i\neq j} |B_{ij}|+ \max_{i\neq j}|C_{ij}|).
\end{align*}
\end{proof}

We now turn to the entanglement breaking (EB) property for random DOC channels. 

\begin{theorem} \label{thm-DOCEB}
Under the regime $s\gtrsim d^{2+t}$ for some constants $t>0$ and $d\to \infty$, the random DOC channel $\Phi_{\DOC}$ is almost surely entanglement breaking.
\end{theorem}
\begin{proof}
Let us denote by
\begin{equation} \label{eq-22matrix}
    \begin{pmatrix}  a & b \\ c & d \end{pmatrix}_{i,j\in [d]}:=a|i\ra \la i|+b|i\ra \la j|+c|j\ra\la i| + |j\ra\la j|\in \M{d}
\end{equation}
for simplicity, and let us decompose $(A,B,C)=\sum_{1\leq i<j\leq d} (A^{(ij)},B^{(ij)},C^{(ij)})$ where
\small
    $$A^{(ij)}=\begin{pmatrix}\frac{1}{d-1}A_{ii} & A_{ij} \\ A_{ji} & \frac{1}{d-1}A_{jj} \end{pmatrix}_{i,j\in [d]}, 
    \quad B^{(ij)}=\begin{pmatrix}\frac{1}{d-1}A_{ii} & B_{ij} \\ B_{ji} & \frac{1}{d-1}A_{jj} \end{pmatrix}_{i,j\in [d]}, \quad C^{(ij)}=\begin{pmatrix}\frac{1}{d-1}A_{ii} & C_{ij} \\ C_{ji} & \frac{1}{d-1}A_{jj} \end{pmatrix}_{i,j\in [d]}.$$
\normalsize
Let $\Phi_{ij}:=\Phi^{(3)}_{A_{ij},B_{ij},C_{ij}}$ be the DOC map associated to the matrix triple $(A^{(ij)},B^{(ij)},C^{(ij)})$. Note that in the regime $s\gtrsim d^{2+t}$ and $d\to\infty$, \cref{prop:DOC-entry-limit} implies
    $$\min_{i,j}dA_{ij}\to 1, \quad  \max_{i<j}d^2|B_{ij}|\to 0, \quad \max_{i<j}d^2|C_{ij}|\to 0\;\text{ almost surely.}$$
Therefore, we obtain almost surely,
    $$\frac{1}{(d-1)^2}A_{ii}A_{jj}\geq |B_{ij}|^2, \quad \frac{1}{(d-1)^2}A_{ii}A_{jj}\geq |C_{ij}|^2, \quad A_{ij}A_{ji}\geq |B_{ij}|^2,\;\; \forall\,i<j.$$
In particular, this implies that the map $\Phi_{ij}$ is PPT for all $i<j$ according to \cref{prop:DOC-PPT}. However, since each triple $(A^{(ij)},B^{(ij)},C^{(ij)})$ is supported on a $2$-dimensional subspace, this implies that $\Phi_{ij}$ is EB for all $i<j$, and hence $\Phi_{\DOC}=\sum_{i<j}\Phi_{ij}$ is EB.
\end{proof}
Note that the proof does not work anymore for the regime $s\sim cd^2$. For example, when $c=1$, we can show that $d^2B_{ij}$ converges in moment to the exponential distribution $\mathsf{Exp}(1)$ using \cref{lem-ranDOCentry}, and thus
    $$\mathbb{P}\left(\frac{1}{(d-1)^2}A_{ii} A_{jj}\geq |B_{ij}|^2\right)\to \mathbb{P}(Z_{\mathsf{Exp}(1)}\leq 1)<1$$
as $d\to \infty$. 

Since channels having diagonal (conjugate) unitary covariance can be obtained from DOC channels by twirling, we have the following simple consequence of \cref{thm-DOCEB}.
\begin{corollary}
    Under the same hypothesis as above, the same conclusion holds for random DUC and CDUC channels. 
\end{corollary}

\section{\texorpdfstring{The PPT$^2$ conjecture holds for random DOC channels}{The PPT squared conjecture holds for random DOC channels}}\label{sec:PPT2}

Originally motivated by the theory of quantum repeaters, the \emph{PPT$^2$ conjecture} is now a central open question in quantum information theory, with recent ramifications into operator theory and operator algebras. We recall below its statement.

\begin{conjecture}[\cite{PPTsq, Christandl2018}]
The composition of two arbitrary \emph{PPT} quantum channels is entanglement breaking. 
\end{conjecture}

The conjecture holds trivially in dimension $d=2$, and it has been shown to hold also in dimension $d=3$ \cite{Christandl2018, Chen2019}. A plethora of papers establish the conjecture in several restricted scenarios, or with weaker conclusions (such as the need to compose a PPT map with itself several times) \cite{Lami2015entanglebreak,Kennedy2017,Rahaman2018,hanson2020eventually,girard2020convex,jin2020investigation}. Importantly for us, Collins, Yin, and Zhong showed \cite{collins2018ppt} that two \emph{independent} random quantum channels that are PPT satisfy the conjecture. 

In the case of quantum channels with diagonal unitary symmetry, is has been shown in \cite{singh2022ppt} that \emph{any} two (C)DUC channels satisfy the PPT$^2$ conjecture. The question whether the same holds for DOC channels has been left open \cite[Conjecture 5.1]{singh2022ppt}. In this section, we contribute to this conjecture in two perspectives: \textit{deterministic} and \textit{random}. First of all, we show that one of the conditions for two PPT (C)DUC channels in \cite{singh2022ppt} can be relaxed to DOC. Furthermore, we answer in the affirmative this question for \emph{random} DOC channels, under much weaker conditions. We refer to \cref{cor-PPTDOC,thm:PPTsquared-DOC} for the details. 

The results rely on the following condition for the composition of two matrix triples to be TCP. Recall from \cite[Lemmas 9.3 and 9.7]{singh2021diagonal} that the composition of two DOC channels is DOC: 
$$\Phi^{(3)}_{A,B,C} \circ \Phi^{(3)}_{D,E,F} = \Phi^{(3)}_{\mathcal A, \mathcal B, \mathcal C}$$
with $(\mathcal{A},\mathcal{B},\mathcal{C}):=(A,B,C)\circ (D,E,F)$, where
\begin{align*}
    \mathcal{A}&=AD,\\
    \mathcal{B}&=\diag(AD)+\ring{B}\odot \ring{E}+\ring{C}\odot \ring{F}^T,\\
    \mathcal{C}&=\diag(AD)+\ring{B}\odot \ring{F}+\ring{C}\odot \ring{E}^T.
\end{align*}

\begin{proposition} \label{prop:TCP-composition}
If $(A,B,C)$ and $(D,E,F)$ are triples of matrices such that the corresponding DOC linear maps $\Phi_{A,B,C}^{(3)}$, $\Phi_{D,E,F}^{(3)}$ are completely positive, and, moreover, for all $i\neq j \in [d]$,
\begin{align}
    \left(\frac{A_{ii}D_{ii}}{d-1}+A_{ij}D_{ji}\right)\left(\frac{A_{jj}D_{jj}}{d-1}+A_{ji}D_{ij}\right)  &\geq \max(|B_{ij}E_{ij}+C_{ij}\overline{F_{ij}}|^2,|B_{ij}F_{ij}+C_{ij}\overline{E_{ij}}|^2), \label{eq-PPTDOC1}\\
    (AD)_{ij}(AD)_{ji}&\geq \max(|B_{ij}E_{ij}+C_{ij}\overline{F_{ij}}|^2,|B_{ij}F_{ij}+C_{ij}\overline{E_{ij}}|^2). \label{eq-PPTDOC2}
\end{align}
Then, their composition $(A,B,C)\circ (D,E,F)$ is TCP.
\end{proposition}
\begin{proof}
The proof is a slight modification of the proof of \cite[Theorem 4.5]{singh2022ppt} in straightforward way. 
Using the notation in \cref{eq-22matrix}, we can decompose $(\mathcal{A},\mathcal{B},\mathcal{C})=\sum_{1\leq i<j\leq d} (\mathcal{A}^{(ij)},\mathcal{B}^{(ij)},\mathcal{C}^{(ij)})$ where
\begin{align*}
\mathcal{A}^{(ij)}&=\begin{pmatrix}\frac{1}{d-1}A_{ii}D_{ii}+A_{ij}D_{ji} & (AD)_{ij} \\ (AD)_{ji} & \frac{1}{d-1}A_{jj}D_{jj}+A_{ji}D_{ij}\end{pmatrix}_{i,j\in [d]},\\
\mathcal{B}^{(ij)}&=\begin{pmatrix}\frac{1}{d-1}A_{ii}D_{ii}+A_{ij}D_{ji} & B_{ij}E_{ij}+C_{ij}\overline{F_{ij}} \\ B_{ji}E_{ji}+C_{ji}\overline{F_{ji}} & \frac{1}{d-1}A_{jj}D_{jj}+A_{ji}D_{ij}\end{pmatrix}_{i,j\in [d]},\\
\mathcal{C}^{(ij)}&=\begin{pmatrix}\frac{1}{d-1}A_{ii}D_{ii}+A_{ij}D_{ji} & B_{ij}F_{ij}+C_{ij}\overline{E_{ij}} \\ B_{ji}F_{ji}+C_{ji}\overline{E_{ji}} & \frac{1}{d-1}A_{jj}D_{jj}+A_{ji}D_{ij}\end{pmatrix}_{i,j\in [d]},\\
\end{align*}
since $E^T=\overline{E}$ and $F^T=\overline{F}$. The two conditions \eqref{eq-PPTDOC1} and \eqref{eq-PPTDOC2}, combined with \cref{prop:DOC-PPT}, guarantee that each triple $(\mathcal{A}^{(ij)},\mathcal{B}^{(ij)},\mathcal{C}^{(ij)})$ is a ``PPT" triple in $2$-dimensional space, and hence TCP. This shows that $(\mathcal{A},\mathcal{B},\mathcal{C})$ is TCP.
\end{proof}

\begin{remark}
    We actually show something stronger: the composition has \emph{factor width} 2, see \cite[Definition 3.2]{singh2022ppt}:
    $$(A,B,C)\circ (D,E,F) \in \TCP_d^{(2)}.$$
\end{remark}

Note that the conditions \cref{eq-PPTDOC1,eq-PPTDOC2} are already weak enough to improve the main result of \cite[Theorem 4.5]{singh2022ppt}. Indeed, 
suppose that $\Phi^{(3)}_{A,B,C}$ and $\Phi^{(3)}_{D,E,F}$ are PPT, and that one of $B$, $C$, $E$, or $F$ is \textit{diagonal} (corresponding to either $\Phi^{(3)}_{A,B,C}$ or $\Phi^{(3)}_{D,E,F}$ being DUC / CDUC). Then by \cref{prop:DOC-PPT} we have the entrywise inequalities
    $$A_{ij}A_{ji}\geq \max(|B_{ij}|^2,|C_{ij}|^2),\quad D_{ij}D_{ji}\geq \max(|E_{ij}|^2,|F_{ij}|^2),\quad i\neq j\in [d],$$
which yields \cref{eq-PPTDOC1} since any of the diagonal conditions imply that
    $$\max(|B_{ij}E_{ij}+C_{ij}\overline{F_{ij}}|^2,|B_{ij}F_{ij}+C_{ij}\overline{E_{ij}}|^2)\leq \max(|B_{ij}E_{ij}|^2,|C_{ij}\overline{F_{ij}}|^2,|B_{ij}F_{ij}|^2,|C_{ij}\overline{E_{ij}}|^2).$$
Furthermore, \cref{eq-PPTDOC2} holds since $\Phi^{(3)}_{\mathcal A, \mathcal B, \mathcal C}=\Phi^{(3)}_{A,B,C} \circ \Phi^{(3)}_{D,E,F}$ is PPT. Consequently, $(\mathcal A, \mathcal B, \mathcal C)\in \TCP_d^{(2)}$, and hence $\Phi^{(3)}_{\mathcal A, \mathcal B, \mathcal C}$ is entanglement breaking. We summarize the argument in the following corollary.

\begin{corollary} \label{cor-PPTDOC}
If $\Phi_{A, B, C}^{(3)}$ and $\Phi_{D, E}^{(i)}$ ($i\in \{1,2\}$) are both PPT maps, then both compositions  $\Phi_{A, B, C}^{(3)}\circ \Phi_{D, E}^{(i)}$ and $\Phi_{D, E}^{(i)}\circ \Phi_{A, B, C}^{(3)}$ are entanglement breaking.
\end{corollary}

The conditions in \cref{prop:DOC-PPT} also allow us to obtain another main result for random DOC channels.

\begin{theorem} \label{thm:PPTsquared-DOC}
Consider two random DOC channels $\Phi_i \sim \mu^\DOC_{d,s_i}$, $i=1,2$, in the asymptotic regime where $d \to \infty$ and $s_i\gtrsim d^{t_i}$ for some constants $t_i>0$. Then almost surely as $d\to \infty$, the composition $\Phi_1\circ \Phi_2$ is entanglement breaking. 

In particular, in the asymptotic regime above, any two random DOC channels satisfy the PPT$^2$ conjecture. 
\end{theorem}
\begin{proof}
Let $\Phi_1 = \Phi_{A, B, C}^{(3)}$ and $\Phi_2 = \Phi_{D, E, F}^{(3)}$. The result follows from \cref{prop:TCP-composition} and \cref{prop:DOC-entry-limit}: almost surely, we have
    $$\lim_{d\to \infty} d^4\min_{i\neq j} \left(\frac{A_{ii}D_{ii}}{d-1}+A_{ij}D_{ji}\right)\left(\frac{A_{jj}D_{jj}}{d-1}+A_{ji}D_{ij}\right)= 1,$$
    $$\liminf_{d\to \infty}d^4\min_{i\neq j} (AD)_{ij} (AD)_{ji}\geq \liminf_{d\to\infty} \min_{i\neq j}d^4(A_{ii}D_{ij}A_{jj}D_{ji}+A_{ij}D_{jj}A_{ji}D_{ii})=2,$$
while $d^4\max_{i\neq j}|B_{ij}E_{ij}+C_{ij}\overline{F_{ij}}|^2\to 0$ and $d^4\max_{i\neq j}|B_{ij}F_{ij}+C_{ij}\overline{E_{ij}}|^2\to 0$. Therefore, the conditions \eqref{eq-PPTDOC1} and \eqref{eq-PPTDOC2} hold asymptotically.
\end{proof}

\begin{remark}
    We would like to emphasize that the hypothesis of \cref{thm:PPTsquared-DOC} are very weak:
    \begin{itemize}
        \item The two channels need not be PPT; actually, from \cref{thm:DOC-PPT}, if $s_i\sim c_id^{t_i}$ with $0<t_i<1$ and $c_i>0$, or $t_i=1$ and $c_i<4$, $\Phi_i$ is asymptotically not PPT.
        \item There is no assumption whatsoever about the correlations of the two random variables $\Phi_1$ and $\Phi_2$: the two DOC channels can be independent, equal, complex conjugate to one another, or have any possible correlation.
    \end{itemize}
    The fact that we can show that the PPT$^2$ conjecture holds for random DOC quantum channels under such weak conditions (compared to \cite{collins2018ppt} where the channels are assumed to be PPT and independent) is due to the presence of the $\mathcal{DO}_d$ covariance. 
\end{remark}

\bigskip

\noindent\textbf{Conflict of interest.} The authors declare no competing interests.

\bigskip

\noindent\textbf{Data availability statement.} Data sharing is not applicable to this article as no new data were created or analyzed in this study. 

\bigskip

\noindent\textbf{Acknowledgments.} Both authors were supported by the ANR project \href{https://esquisses.math.cnrs.fr/}{ESQuisses}, grant number ANR-20-CE47-0014-01, and by the PHC program \emph{Star} (Applications of random matrix theory and abstract harmonic analysis to quantum information theory). I.N.~also received support from the ANR project \href{https://www.math.univ-toulouse.fr/~gcebron/STARS.php}{STARS}, grant number ANR-20-CE40-0008.

\appendix

\section{Technical lemmas and proofs}

We gather in this appendix some technical background material, as well as the proofs of some of the results in the main text. 

\subsection{General results about the combinatorics of permutations and Weingarten functions}

Our main approach to understanding the entanglement properties of the different models of random covariant channels is the \emph{moment method} in random matrix theory. Specifically, we shall apply graphical Weingarten calculus \cite{collins2003moments,collins2006integration,collins2010randoma} to compute many kinds of integrals involving Haar random isometries for the precise analysis. In this section, we introduce basic notions on permutations, partitions, and the Weingarten formula. Moreover, we present several combinatorial lemmas which will be useful throughout the appendix for the actual computation later on.

The main combinatorial objects that appear in the resolution of integrals with respect to the Haar measure on the unitary group are \emph{permutations} and, in the limit of large dimensions, \emph{non-crossing partitions}. We refer the reader to the excellent collection of lecture notes \cite{Nica_Speicher_2006} for an introduction to the relation between these concepts, free probability theory, and random matrix theory. 

We shall denote by $S_p$ the symmetric group of order $p$, that is the permutations of the set 
$$[p]:=\{1,2,\ldots, p\}.$$
For a permutation $\alpha \in S_p$, we denote by $\# \alpha$ the \emph{number of cycles} in its cycle decomposition, including singletons. We have $\# \id = p$, whereas for the \emph{full cycle permutation} $\# \gamma = 1$, were
$$\gamma:=(p\;\;p-1\,\cdots\,3\; 2\; 1)\in S_p.$$
The \emph{length} $|\alpha|$ of a permutation $\alpha \in S_p$ is the minimal number of transpositions that multiply to $\alpha$. We have the important relation \cite[Proposition 23.11]{Nica_Speicher_2006}
$$\forall \alpha \in S_p, \qquad \# \alpha + |\alpha| = p.$$

Let us discuss now the \emph{Weingarten formula} \cite{collins2003moments,collins2006integration}, the main technical tool that will allow us to compute various average quantities related to the Haar measure on the set of isometries. Below, the integral is taken either to the Haar measure on the set of isometries $V : \C{k} \to \C{d}$ for $k \leq d$. The isometry $V$ can be thought of as the truncation (keeping the first $k$ columns) of a Haar-distributed random unitary matrix $U \in \mathcal U_d$.

\begin{theorem}
 Let $k \leq d$ be positive integers and
$\mathbf{i}=(i_1,\ldots ,i_p)$, $\mathbf{i'}=(i'_1,\ldots ,i'_p)$,
$\mathbf{j}=(j_1,\ldots ,j_p)$, $\mathbf{j'}=(j'_1,\ldots ,j'_p)$
be $p$-tuples of positive integers with $\mathbf{i}, \mathbf{i'} \in [d]^p$ and $\mathbf{j}, \mathbf{j'} \in [k]^p$. Then
\begin{equation}\label{eq:Wg}
\int V_{i_1j_1} \cdots V_{i_pj_p}
\overline{V}_{i'_1j'_1} \cdots
\overline{V}_{i'_pj'_p}\, \mathrm{d}V=
\sum_{\sigma, \tau\in S_{p}}\delta_{i_1i'_{\sigma (1)}}\ldots
\delta_{i_p i'_{\sigma (p)}}\delta_{j_1j'_{\tau (1)}}\ldots
\delta_{j_p j'_{\tau (p)}} \Wg_d(\tau\sigma^{-1}).
\end{equation}
\end{theorem}
Above, $\Wg_d(\sigma)$ is a rational function of $d$ that depends only on the cycle type of the permutation $\sigma \in S_p$. The large $d$ behavior of this function will be important for us: 
$$\Wg_d(\sigma) = \big(1+O(d^{-2})\big) \Mob(\sigma) d^{-p-|\sigma|},$$
where $\Mob$ is the \emph{M\"obius function}
$$\Mob(\sigma) := \prod_{c \in \sigma} (-1)^{|c|-1} \operatorname{Cat}_{|c|-1}$$
with the Catalan numbers 
$$\operatorname{Cat}_n := \frac{1}{n+1} \binom{2n}{n}.$$
We shall mainly use the graphical formulation of this formula, introduced in \cite{collins2010randoma}; we refer the reader to the review article \cite{CN16} for a pedagogical presentation.

\begin{lemma}[{\cite[Proposition 1.3.7]{Stanley11}},\cite{collins2010randoma}] \label{lem-weingarten}
We have the following two identities for all $p\geq 1$:
\begin{enumerate}
    \item $\sum_{\alpha\in S_p}n^{\# \alpha}=n(n+1)\cdots (n+p-1)$,

    \item $\sum_{\alpha\in S_p}{\rm Wg}_{n}(\alpha)=\left(n(n+1)\cdots (n+p-1)\right)^{-1}$.
\end{enumerate}
\end{lemma}

The following combinatorial formula will be crucially used throughout the appendix.
\begin{lemma} [{\cite[Eq.(2.61)]{collins2003moments}}] \label{lem-binomformula}
For any integer $n\geq 1$ and any polynomial $F$ of degree $<n$, we have
    $$\sum_{p=0}^n(-1)^{n-p}F(p)\binom{n}{p}=0.$$
\end{lemma}

\begin{lemma} \label{lem-polycoeff}
\begin{enumerate}
    \item Let us write the power series expansion
        $$s^{1-p}(s+1)\cdots (s+p-1)=\sum_{k=0}^{\infty}a_{k,p}s^{-k}$$
    in $s^{-1}$. Then for each $k$, the coefficient $a_{k,p}$ is a polynomial in $p$ of degree $2k$. Moreover, $a_{k,p}=0$ for all $k\geq p$.

    \item If we write the power series expansion
        $$\sum_{\alpha\in \mathcal{P}_{1,2}(p)}s^{-|\alpha|}=\sum_{k=0}^{\infty}b_{k,p}s^{-k},$$
    then $b_{0,p}\equiv 1$ and $b_{k,p}=\frac{p(p-1)\cdots (p-2k+1)}{2^k k!}$ for $k\geq 1$.

    \item If we write the power series expansion
        $$\sum_{\alpha\in NC_{1,2}(p)}s^{-|\alpha|}=\sum_{k=0}^{\infty}c_{k,p}s^{-k},$$
    then $c_{0,p}\equiv 1$ and $c_{k,p}=\binom{p}{2k}{\rm Cat}_k=\frac{p(p-1)\cdots(p-2k+1)}{(k+1)!k!}$ for $k\geq 1$.
\end{enumerate}
\end{lemma}
\begin{proof}
(1) The second assertion is obvious since $s(s+1)\cdots (s+p-1)$ is a polynomial in $s$ of degree $p$ with zero constant. On the other hand, we can easily see that $a_{0,p}\equiv 1$, and $a_{k,p}$ satisfy the recurrence
    $$a_{k,p}=(p-1)a_{k-1,p-1}+a_{k,p-1},\;\; p\geq 1, k\geq 1,$$
from the relation
    $$s(s+1)\cdots(s+p-1)=s^{p}\sum_{k=0}^{\infty}a_{k,p}s^{-k}=(s+p-1)s^{p-1}\sum_{k=0}^{\infty}a_{k,p-1}s^{-k}.$$
Therefore, we have $a_{k,p}=\sum_{j=0}^{p-1}ja_{k-1,j}+a_{k,0}=\sum_{j=0}^{p-1}ja_{k-1,j}$. By induction hypothesis, $pa_{k-1,p}$ is a polynomial in $p$ of degree $2k-1$, and hence the first assertion follows.

(2) We leave its proof to the reader, just by noting that $b_{k,p}$ is equal to the number of permutations consisting of exactly $k$ transpositions (and $p-2k$ singletons) in their cycles.

(3) Similarly with (2), $c_{k,p}$ is the number of permutations consisting of $k$ transpositions which are non-crossing. We can find this by first choosing $2k$ elements in $[p]$ and then count the number of non-crossing pair partitions ($={\rm Cat}_k$) of these elements.
\end{proof}

\begin{lemma} \label{lem:permutation-triangle}
Given two permutations $\sigma_1,\sigma_2\in S_p$, the map
$$S_p \ni \alpha\mapsto (|\sigma_1^{-1}\alpha|+|\alpha^{-1}\sigma_2|)\; {\rm mod\,} 2\in \{0,1 \}$$
is constant. In particular, for all $\alpha\in S_p$, there exists a nonnegative integer $g$ such that
    $$|\sigma_1^{-1}\alpha|+|\alpha^{-1}\sigma_2|=|\sigma_1^{-1}\sigma_2|+2g.$$
\end{lemma}
\begin{proof}
The first assertion follows from the observations that every permutation $\alpha$ can be written as a product of transpositions and $|\sigma \tau|=|\sigma|\pm 1$ for every $\sigma\in S_p$ and a transposition $\tau$ \cite[Lemma 23.10]{Nica_Speicher_2006}. The second point then follows from the first point and the triangle inequality $|\sigma_1^{-1}\alpha|+|\alpha^{-1}\sigma_2|\geq |\sigma_1^{-1}\sigma_2|$.
\end{proof}

We shall denote by $\mathcal P(p)$ the set of partitions of $[p]=\{1,2,\ldots, p\}$. Importantly, we denote by $NC(p)$ the set of \emph{non-crossing partitions} of $[p]$, that is partitions $\pi$ for which there do not exist $i < j < k < l \in [p]$ such that $i \stackrel{\pi}{\sim}k$ and $j \stackrel{\pi}{\sim} l$ (in which case we say that the blocks of $\pi$ containing $i$ and $k$ and, respectively, $j$ and $l$ are crossing). Let us denote by $S_{NC}(\sigma):=\{\alpha\in S_p: |\alpha|+|\alpha^{-1}\sigma|=|\sigma|\}$ the set of permutations which lie on a \emph{geodesic} from $\id$ to $\sigma$. In what follows, we will provide a complete description of the elements of $S_{NC}(\sigma)$ for arbitrary permutation $\sigma\in S_p$. First, we introduce the projection map $\pi:S_p\to \mathcal{P}(p)$
which naturally identifies each of the cycles of a permutation $\alpha$ as blocks of $\pi(\alpha)$. Then the following proposition states that we can interpret the elements of $S_{NC}(\gamma)$ as non-crossing partitions when $\gamma$ is a full cycle.

\begin{proposition} [{\cite[Proposition 23.23]{Nica_Speicher_2006}}] \label{prop:noncrossing-permutation}
Let $\gamma=(p\;\;p-1\,\cdots\, 1)\in S_p$. Then the restriction map $\pi\bigm|_{S_{NC}(\gamma)}$ of the projection map $\pi:S_p\to \mathcal{P}(p)$ defines a natural bijection between $S_{NC}(\gamma)$ and $NC(p)$, which preserves the partial order structure. More explicitely, 
\begin{enumerate}
    \item $\alpha\in S_{NC}(\gamma)$ if and only if $\pi(\alpha)$ forms a non-crossing partition and every cycle of $\alpha$ can be written in the form
        $$(j_1\, j_2\,\cdots \, j_r),\quad j_1>j_2>\cdots>j_r,$$

    \item if $\alpha \in S_{NC}(\gamma)$ and $\beta\in S_{NC}(\alpha)$, then $\beta\in S_{NC}(\gamma)$ and $\pi(\beta)\leq \pi(\alpha)$ as partitions \cite[Definition 9.14]{Nica_Speicher_2006}.
\end{enumerate}
\end{proposition}

\begin{proposition}\label{prop:geodesic-permutation}
Let $\sigma=c_1\cdots c_k$ be the cycle decomposition of $\sigma\in S_p$. If $\alpha\in S_{NC}(\sigma)$, then $\alpha$ can be decomposed into disjoint permutations $\alpha_1\cdots \alpha_k$ such that, for each $i=1,\ldots k$, $\alpha_i\leq c_i$ as partitions and $\alpha_i\in S_{NC}(c_i)$. Furthermore, if $\alpha'\in S_{NC}(\sigma)$ and $\alpha \equiv \alpha'$ as partitions, then $\alpha=\alpha'$ as permutations.
\end{proposition}
\begin{proof}
By re-indexing (i.e. conjugation $\delta\sigma \delta^{-1}$ for some $\delta\in S_p$), we may assume that
    $$c_i=(p_i\,\,(p_i-1)\,\cdots\, (p_{i-1}+1)),\quad i=1,\ldots,k,$$
for some $1=p_0<p_1<\cdots<p_k=p.$ Then \cref{prop:noncrossing-permutation} implies that $\sigma,\alpha\in S_{NC}(\gamma)$ and $\alpha\leq \sigma$ as partitions. In particular, each block $V_i=\{p_{i-1}+1,\ldots, p_i-1, p_i\}$ is invariant under $\alpha$, so the restriction $\alpha_i:=\alpha \bigm|_{V_i}\in S(V_i)$ is well-defined. Now we can write $\alpha=\alpha_1\cdots \alpha_k$ and $|\alpha|=\sum_{i=1}^k|\alpha_i|$. Moreover, the condition $\alpha\in S_{NC}(\sigma)$ implies that
    $$|\sigma|=\sum_{i=1}^k|\sigma_i|\leq \sum_{i=1}^k(|\alpha_i|+|\alpha_i^{-1}\sigma_i|)=\sum_i |\alpha_i|+\sum_i |\sigma_i^{-1}\alpha_i|=|\alpha|+|\alpha^{-1}\sigma|=|\sigma|,$$
and therefore, the equality condition forces that $\alpha_i\in S_{NC}(\sigma_i)$ for all $i$. Finally, the last statement follows from \cref{prop:noncrossing-permutation}(1): $\alpha,\alpha'\in S_{NC}(\gamma)$ and $\pi(\alpha)= \pi(\alpha')$ imply that $\alpha=\alpha'$.
\end{proof}

\cref{prop:geodesic-permutation} implies that we have the natural bijection
\begin{equation} \label{eq-NC-bij}
    S_{NC}(\sigma)\cong S_{NC}(c_1)\times\cdots \times S_{NC}(c_k).
\end{equation}
Now let us denote by $S_{NC_{1,2}}(\sigma)$ the set of permutations in $S_{NC}(\sigma)$ having only cycles of size 1 or 2. The set can be defined equivalently by
\begin{equation} \label{eq-NC12}
    \alpha\in S_{NC_{1,2}}(\sigma) \iff \alpha\in S_{NC}(\sigma) \text{ and } \alpha=\alpha^{-1} \iff \alpha, \alpha^{-1}\in S_{NC}(\sigma).
\end{equation}
since any cycle having size larger than 2 is distinct from its inverse, and since $\alpha,\alpha^{-1}\in S_{NC}(\sigma)$ implies $\alpha=\alpha^{-1}$ by \cref{prop:geodesic-permutation}. Furthermore, note that the set $S_{NC_{1,2}}(\gamma)$ can be identified with $NC_{1,2}(p)$, the set of non-crossing partitions having only blocks of size 1 or 2, in a trivial way. Finally, if $\sigma=c_1\cdots c_k$ is the cycle decomposition of $\sigma$, then the bijection \cref{eq-NC-bij} reduces to
\begin{equation} \label{eq-NC12-bij}
    S_{NC_{1,2}}(\sigma)\cong S_{NC_{1,2}}(c_1)\times\cdots \times S_{NC_{1,2}}(c_k).
\end{equation}

For two partitions $\pi_1,\pi_2\in \mathcal{P}(p)$, the \textit{join} $\pi_1\vee \pi_2$ is defined as the supremum of $\pi_1$ and $\pi_2$ in $\mathcal{P}(p)$ \cite[Definition 9.15]{Nica_Speicher_2006}. Since the length function $|\alpha|=p-\#\, \alpha$ only depends on the corresponding partition $\pi(\alpha)$, we can define  $|\alpha\vee\beta|:=p-\#(\pi(\alpha)\vee \pi(\beta))$ for arbitrary permutations $\alpha,\beta\in S_p$. Note that $|\alpha\vee \beta|\geq \max(|\alpha|,|\beta|)$ in general since the join operation $\vee$ does not increase the number of blocks as partitions. Moreover, if $|\alpha\vee \beta|=|\alpha|=|\beta|$, or equivalently $\#(\pi(\alpha)\vee \pi(\beta))=\#(\pi(\alpha))=\#(\pi(\beta))$, then we have $\pi(\alpha)=\pi(\beta)$.

We now present several permutation inequalities that will be applied later to prove \cref{prop-ranDOCmatrix}.

\begin{lemma} \label{lem:permutation-optimization}
Let $\sigma\in S_p$ be a permutation.
\begin{enumerate}
    \item For all $\alpha,\beta\in S_p$, we have
        $$|\alpha|+2|\alpha\beta^{-1}|+|(\sigma^{-1}\alpha)\vee (\sigma^{-1}\beta^{-1})|\geq |\sigma|,$$
    and the equality holds if and only if $\alpha=\beta\in S_{NC_{1,2}}(\sigma)$. Furthermore, whenever the equality does not hold, we have 
    $$|\alpha|+2|\alpha\beta^{-1}|+|(\sigma^{-1}\alpha)\vee (\sigma^{-1}\beta^{-1})|\geq |\sigma|+2.$$

    \item For all $\alpha,\beta\in S_p$, we have
        $$|\alpha\beta^{-1}|+|(\sigma^{-1}\alpha)\vee (\sigma^{-1}\beta^{-1})|\geq \frac{1}{2}|\sigma^{2}|,$$
    and the equality holds if and only if $\alpha=\beta\in S_{NC_{1,2}}(\sigma)$ and $\sigma^{-1}\alpha\in S_{NC}(\sigma^{-2})$. Furthermore, when $\sigma=\gamma=(p\;\;p-1\,\cdots\, 1)$, $\alpha\in NC_{1,2}(p)$ is `maximally paired', i.e.,
    \begin{align}
        p=\text{even} &\implies \text{$\alpha$ consists of $\frac{p}{2}$ transpositions i.e., $\alpha\in NC_{2}(p)$}, \label{eq-perm-opt1}\\
        p=\text{odd} &\implies \text{$\alpha$ consists of one singleton and $\frac{p-1}{2}$ transpositions} \label{eq-perm-opt2}.
    \end{align}
\end{enumerate}
\begin{proof}
(1) Since $|(\sigma^{-1}\alpha)\vee (\sigma^{-1}\beta^{-1})|\geq \max(|\sigma^{-1}\alpha|,|\sigma^{-1}\beta^{-1}|)$, we have
\begin{align*}
    {\rm LHS}&\geq |\alpha|+|\sigma^{-1}\alpha|\geq |\sigma|,\\
    {\rm LHS}&\geq |\alpha|+|\alpha\beta^{-1}|+|\sigma^{-1}\beta^{-1}|\geq |\alpha^{-1}|+|\sigma^{-1}\alpha^{-1}|\geq |\sigma|.
\end{align*}
In particular, the equality conditions for both two inequalities imply that $\alpha=\beta$ and $\alpha,\alpha^{-1}\in S_{NC}(\sigma)$, and therefore $\alpha\in S_{NC_{1,2}}$ by \cref{eq-NC12}. On the other hand, if the equality does not hold, then either $\alpha\neq \beta$ or $\alpha=\beta\not\in S_{NC_{1,2}}$. For the former case, we have
    $${\rm LHS}\geq (|\alpha|+|\sigma^{-1}\alpha|)+2|\alpha\beta^{-1}|\geq |\sigma|+2,$$
and for the latter case, we again have
    $${\rm LHS}\geq \max(|\alpha|+|\sigma^{-1}\alpha|,|\alpha^{-1}|+|\sigma^{-1}\alpha^{-1}|)\geq |\sigma|+2$$
by \cref{lem:permutation-triangle}.

\medskip

(2) Similarly as above, we have
    $${\rm LHS}\geq |\alpha\beta^{-1}|+\max(|\sigma^{-1}\alpha|,|\sigma^{-1}\beta^{-1}|)\geq \frac{1}{2}\Big(|\alpha\beta^{-1}|+(|\sigma^{-1} \alpha|+|\alpha^{-1}\beta|+|\beta^{-1}\sigma^{-1}|)\Big)\geq \frac{1}{2}|\sigma^2|.$$
All the equalities above hold if and only if $\alpha=\beta$, $\sigma^{-1}\alpha\equiv \sigma^{-1}\alpha^{-1}$ as partitions, and $|\sigma^{-1}\alpha|+|\sigma^{-1}\alpha^{-1}|=|\sigma^{-2}|$. Note that the last condition implies that $\sigma^{-1}\alpha,\sigma^{-1}\alpha^{-1}\in S_{NC}(\sigma^{-2})$, and therefore, $\alpha=\alpha^{-1}$ by \cref{prop:geodesic-permutation}. In particular, every cycle of $\alpha$ has size 1 or 2.

We claim that $\alpha\in S_{NC_{1,2}}(\sigma)$ in this case. Indeed, by considering the cycle decomposition of $\sigma$ and applying \cref{prop:geodesic-permutation} for $\sigma^{-1}\alpha\in S_{NC}(\sigma^{-2})$, we may assume that $\sigma=\gamma=(p\;\;p-1\,\cdots\, 1)$. Then we have
    $$|\gamma^{-1}\alpha|=|\gamma^{-1}\alpha^{-1}|=\frac{1}{2}|\gamma^{-2}|=\begin{cases} p/2-1 & \text{if $p$ is even,} \\ (p-1)/2 & \text{if $p$ is odd.}  \end{cases}$$
Therefore, it suffices to show $|\alpha|=\begin{cases} p/2 & \text{if $p$ is even}, \\ (p-1)/2 & \text{if $p$ is odd} \end{cases}$ to conclude $\alpha\in S_{NC_{1,2}}(\gamma)\cong NC_{1,2}(p)$, which is equivalent to \cref{eq-perm-opt1,eq-perm-opt2}. First, when $p$ is even, we have
    $$\gamma^{-2}=(1\,3\,\cdots p-1)(2\,4\,\cdots p),$$
and in particular, the condition $\gamma^{-1}\alpha\in S_{NC}(\gamma^{-2})$ implies that every cycle $\gamma^{-1}\alpha$ consists of the elements having the same parity. This is impossible when there is any $i\in [p]$ such that $\alpha(i)=i$. Therefore, we have \cref{eq-perm-opt1}. On the other hand, when $p$ is odd, we can write
    $$\gamma^{-2}=(1\,\,3\,\,\cdots\,\,p\,\,2\,\,4\cdots\,\,p-1).$$
If \cref{eq-perm-opt2} is not true, then we have at least three singletons inside $\alpha$, and in particular, we can find two distinct integers $i<i'$ with the same parity such that $\alpha(i)=i$ and $\alpha(i')=i'$. Then $\gamma^{-1}\alpha(i)=i+1$, $\gamma^{-1}\alpha(i')=i'+1$, and we can easily see that the pairs $(i,i+1)$ and $(i',i'+1)$ are crossing inside the full cycle $\gamma^{-2}$. Moreover, $i,i+1,i',i'+1$ cannot appear, in this order, in the same cycle of $\gamma^{-1}\alpha$ since any cycle inside $\gamma^{-1}\alpha\in S_{NC}(\gamma^{-2})$ must be of the form
    $$(i_1\,\cdots i_r\,j_1\,\cdots j_{r'})$$
where $i_1<\cdots<i_r$ are odd integers and $j_1<\cdots <j_{r'}$ are even integers, according to \cref{prop:noncrossing-permutation} (1). Consequently, we have a contradiction, and hence \cref{eq-perm-opt2} holds.

\end{proof}
\end{lemma}
 
\begin{lemma} \label{lem:permutation-optimization2}
Let $\gamma:=(p\,\,p-1\,\cdots\, 1)\in S_p$ and $\tilde{\gamma}:=(p\,\,p-1\,\cdots\, 1)(2p\;2p-1\,\cdots \,p+1)\in S_{2p}$. For $p=2,3$ we have
\begin{align}
    \alpha,\beta\in S_p,\;\; |\alpha|+|\alpha\beta^{-1}|> \frac{p}{2} &\implies |\gamma^{-1}\alpha \vee \gamma^{-1}\beta^{-1}|+|\alpha\beta^{-1}|\geq \frac{p}{2}+\frac{1}{2}, \label{eq-C-mean1}\\
    \alpha,\beta\in S_{2p},\;\; |\alpha|+|\alpha\beta^{-1}|> p &\implies |\tilde{\gamma}^{-1}\alpha \vee \tilde{\gamma}^{-1}\beta^{-1}|+|\alpha\beta^{-1}|\geq p. \label{eq-C-mean2}
\end{align}
\end{lemma}
\begin{proof}
Let us first consider the case $p=2$. The proof of \cref{eq-C-mean1} is simple: $|\alpha|+|\alpha\beta^{-1}|>1$ holds only if $\alpha=(1\, 2)$ and $\beta=\id$, and $|\gamma^{-1}\alpha\vee \gamma^{-1}\beta^{-1}|+|\alpha\beta^{-1}|=2$ in this case. Regarding \cref{eq-C-mean2}, note that we always have $|\tilde{\gamma}^{-1}\alpha\vee \tilde{\gamma}^{-1}\beta^{-1}|\geq 1$: otherwise, $\alpha=\beta=\tilde{\gamma}=(1\,2)(3\,4)$, but this contradicts the condition $|\alpha|+|\alpha\beta^{-1}|>2$. Now if we assume $\alpha=\beta$ and $|\tilde{\gamma}^{-1}\alpha\vee \tilde{\gamma}^{-1}\alpha^{-1}|=1$, then $\max(|\tilde{\gamma}^{-1}\alpha|,|\tilde{\gamma}^{-1}\alpha^{-1}|)=1$. Both the cases $|\tilde{\gamma}^{-1}\alpha|=0$ and $|\tilde{\gamma}^{-1}\alpha|=0$ imply that $\alpha=\beta=\tilde{\gamma}$ for which we already had a contradiction. If $|\tilde{\gamma}^{-1}\alpha|=|\tilde{\gamma}^{-1}\alpha^{-1}|=1=|\tilde{\gamma}^{-1}\alpha\vee \tilde{\gamma}^{-1}\alpha^{-1}|$, then we actually have $\tilde{\gamma}^{-1}\alpha=\tilde{\gamma}^{-1}\alpha^{-1}$ and $\alpha=\alpha^{-1}$. This again induces a contradiction since $|\alpha|+|\alpha\beta^{-1}|=|\alpha|\leq 2$. 

Now consider the case $p=3$. By \cref{lem:permutation-optimization} (2), we already know
\begin{align*}
    |\gamma^{-1}\alpha \vee \gamma^{-1}\beta^{-1}|+|\alpha\beta^{-1}| &\geq \frac{1}{2}|\gamma^2|=1
\end{align*}
for all $\alpha,\beta\in S_p$ and the equality holds only if $\alpha=\beta$ and $|\alpha|=1$ (\cref{eq-perm-opt2}) which does not satisfy $|\alpha|+|\alpha\beta^{-1}|>3/2$. This implies \cref{eq-C-mean1}. Similarly, we always have
\begin{align*}
    |\tilde{\gamma}^{-1}\alpha \vee \tilde{\gamma}^{-1}\beta^{-1}|+|\alpha\beta^{-1}| &\geq \frac{1}{2}|\tilde{\gamma}^2|=2,
\end{align*}
and the equality condition implies that $\alpha=\beta\in S_{NC_{1,2}}(\tilde{\gamma})$ and
    $$|\alpha|=|\tilde{\gamma}|-|\tilde{\gamma}^{-1}\alpha|=|\tilde{\gamma}|-\frac{1}{2}|\tilde{\gamma}^2|=2,$$
which contradicts the condition $|\alpha|+|\alpha\beta^{-1}|>3$. This implies \cref{eq-C-mean2}.
\end{proof}
\begin{remark}
By repeating the same arguments above, we can actually show that \cref{eq-C-mean1,eq-C-mean2} hold for all odd integers $p$. However, they are not true for even $p\geq 4$ in general: when $\alpha=\beta=(1\, 2\, 3\, 4)\in S_4$,
    $$|\gamma^{-1}\alpha \vee \gamma^{-1}\beta^{-1}|+|\alpha\beta^{-1}|=|(1\,3)(2\,4)|=2<\frac{5}{2},$$
and when $\alpha=\beta=(1\,2\,3\,4)(5\,6)(7\,8)\in S_8$,
    $$|\tilde{\gamma}^{-1}\alpha \vee \tilde{\gamma}^{-1}\beta^{-1}|+|\alpha\beta^{-1}|=|(1\,3)(2\,4)(5\,7)|=3<4.$$
\end{remark}

\subsection{Scalar parameters}\label{app:scalar}

This section contains results about the real parameters $\lambda_{1,2,3}^{(d,s)}$ which appear in the study of random (conjugate) unitary, orthogonal, and hyperoctahedral quantum channels from \cref{sec:conjugate-unitary-orthogonal-covariant,sec:hyperoctahedral-covariant}.

The following moment formulae characterize the distributions of these random variables.

\begin{lemma} \label{lem:moments-lambda-123}
The moments of the random variables $\lambda_{1,2,3}^{(d,s)}$ are given by
\begin{align*}
    \mathbb{E}\Big[\big(\lambda_1^{(d,s)}\big)^p\Big] &= s(s+1)\cdots (s+p-1)\sum_{\beta\in S_p} d^{\# \beta}{\Wg}_{ds}(\beta)\\
    \mathbb{E}\Big[\big(\lambda_2^{(d,s)}\big)^p\Big] &= \sum_{\alpha, \beta\in S_p} s^{\# \alpha} d^{\# (\alpha\beta)}{\Wg}_{ds}(\alpha\beta^{-1})\\
    \mathbb{E}\Big[\big(\lambda_3^{(d,s)}\big)^p\Big] &= \sum_{\alpha, \beta\in S_p} s^{\# \alpha} d^{\# (\alpha\vee\beta)}\Wg_{ds}(\alpha\beta^{-1}).
\end{align*}
\end{lemma}
\begin{proof}
    To obtain the three formulas in the statement, we apply the graphical Weingarten formula for the diagram consisting of $p$ (disconnected) copies of the diagram for $\lambda_{1,2,3}^{(d,s)}$ (see \cref{fig:lambda-1-2,fig:lambda-3}). In the case of $\lambda_{1}^{(d,s)}$, we obtain
    $$\mathbb{E}\Big[\big(\lambda_1^{(d,s)}\big)^p\Big] = \sum_{\alpha,\beta \in S_p} d^{\#(\alpha^{-1}\beta)} s^{\#\alpha} \Wg_{ds}( \alpha^{-1}\beta)= \sum_{\sigma,\pi \in S_p} d^{\#\pi} s^{\#\sigma} \Wg_{ds}(\pi),$$
    since there are precisely $\#(\alpha^{-1}\beta)$ loops of dimension $d$ and $\#\alpha$ loops of dimension $s$. We use \cref{lem-weingarten} to compute explicitly the sum over $\sigma$.

    Counting the loops in a similar fashion gives the formulas for $\lambda_{2}^{(d,s)}$. Regarding the integration of $(\lambda_{3}^{(d,s)})^p$, we can understand the diagram of $\lambda_{3}^{(d,s)}$ by joining the two separated lines representing the tensor component of dimension $d$ as in \cref{fig-weingarten-lambda_3}. Note that after the graphical calculus in the diagram for $\alpha,\beta\in S_p$, the upper (external) and lower (internal) lines contributes $\#\, \alpha$ and $\#\,\beta$ loops of dimension $d$, respectively. However, we have to ``join'' them because we identify the the two kinds of loops (thick lines in the figure). Consequently, all the blocks in $\alpha\vee \beta$ contribute the same loop in the final diagram, and hence there are exactly $\#(\alpha\vee \beta)$ loops of dimension $d$. The term $s^{\# \alpha}$ appears similarly with the case $\lambda_{1,2}^{(d,s)}$, and this concludes the proof.

    \begin{figure}[htb!] 
    \centering
    \includegraphics{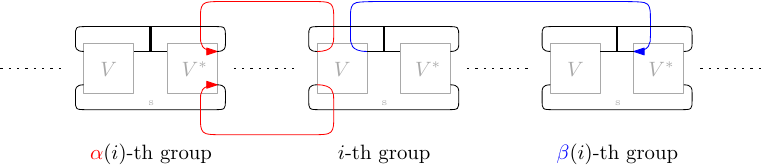}
    \caption{The way the permutations \textcolor{red}{$\alpha$} and \textcolor{blue}{$\beta$} connect the $V$ and the $V^*$ boxes in the diagram for $\mathbb{E}\Big[\big(\lambda_3^{(d,s)}\big)^p\Big]$. Note the thick vertical lines appearing in each diagram of $\lambda_3^{(d,s)}$ that are responsible for joining the otherwise disjoint \textcolor{red}{$\alpha$} and \textcolor{blue}{$\beta$} loops.}
    \label{fig-weingarten-lambda_3}
\end{figure}
\end{proof}

In the case $s=1$, the random variable $\lambda_1^{(d,1)}$ has an interesting interpretation. Indeed, in that case the random isometry $V$ is a Haar-distributed random unitary matrix $U \in \mathcal U(d)$, and thus 
     $$\lambda_1^{(d,s=1)} = |\Tr U|^2.$$
     Then, for any $p\leq d/2$, we have \cite[Theorem 2]{diaconis1994eigenvalues}
     $$\mathbb E \Big[ \big | \Tr U \big|^{2p} \Big] = p!.$$
In the general case. if we apply graphical Weingarten formula and use the fact that $\sigma\mapsto \Wg_{d}(\sigma)$ is the pseudo-inverse of $\sigma\mapsto d^{\#\sigma}$ in $\mathbb{C}[S_p]$, we can extend the same formula to all $d$, and in general we have
        $$\mathbb{E}\left[|\Tr U|^{2p}\right]=p!\sum_{\sigma\in S_p}d^{\# \sigma}\Wg_d(\sigma)=\sum_{\lambda\vdash p,\; l(\lambda)\leq d} d_{\lambda}^2,$$
where $d_{\lambda}$ is the dimension of the irreducible representation of $S_p$ associated to the partition $\lambda$, with the condition that the length of the diagram $\lambda$ is smaller than the dimension $d$.

\bigskip

We now prove the claims about the convergence of the random variables $\lambda_{1,2}^{(d,s)}$ from \cref{sec:conjugate-unitary-orthogonal-covariant}. 

\begin{proof}[Proof of \cref{prop-UUparaconv}]
    We start with the case where $s$ is fixed and we treat the $\lambda_1^{(d,s)}$ random variable. We have
    \begin{align*}
        \mathbb{E}\Big[\big(\lambda_1^{(d,s)}\big)^p\Big] &= s(s+1)\cdots (s+p-1)\sum_{\beta\in S_p} d^{\# \beta}{\Wg}_{ds}(\beta)\\
        &= s(s+1)\cdots (s+p-1)\sum_{\beta\in S_p} d^{p-|\beta|}(ds)^{-p-|\beta|}\Mob(\beta)(1+O((ds)^{-2}))\\
        &= s^{-p} \cdot s(s+1)\cdots (s+p-1) + O(d^{-1}).
    \end{align*}
    The assertion follows now from the observation that 
    $$s(s+1)\cdots (s+p-1)=\int_0^{\infty}\frac{x^{s+p-1}}{(s-1)!}e^{-x}\mathrm{d}x$$
    is the $p$-th moment of the Gamma distribution $\Gamma(s,1)$. 
    
    For the assertion regarding $\lambda_2^{(d,s)}$, compute
    \begin{align*}
        \mathbb{E}\left[(\lambda_2^{(d,s)})^p\right]&= \sum_{\alpha, \beta\in S_p} s^{\# \alpha} d^{\# (\alpha\beta)}{\Wg}_{ds}(\alpha\beta^{-1})\\
        &=\sum_{\alpha,\beta\in S_p}s^{p-|\alpha|}d^{p-|\alpha\beta|}(ds)^{-p-|\alpha\beta^{-1}|}(\Mob(\alpha\beta^{-1})+O((ds)^{-2}))\\
        &=\sum_{\substack{\alpha,\beta\in S_p\\ |\alpha\beta|=|\alpha\beta^{-1}|=0}}s^{-|\alpha|}+\sum_{\substack{\alpha,\beta\in S_p\\ |\alpha\beta|+|\alpha\beta^{-1}|=1}}s^{-|\alpha|-|\alpha\beta^{-1}|}d^{-1}\Mob(\alpha\beta^{-1})+O(d^{-2})\\
        &=\sum_{\substack{\alpha\in S_p\\\alpha^2 = \id}}s^{-|\alpha|}+\sum_{\substack{\alpha,\beta\in S_p\\ |\alpha\beta|=1,\,|\alpha\beta^{-1}|=0}}s^{-|\alpha|}d^{-1}- \sum_{\substack{\alpha,\beta\in S_p\\ |\alpha\beta|=0,\,|\alpha\beta^{-1}|=1}}s^{-|\alpha|-1}d^{-1}+O(d^{-2})\\
        &=\sum_{\substack{\alpha\in S_p\\\alpha^2 = \id}}s^{-|\alpha|}+O(d^{-2}).
    \end{align*}
    Here we used the fact that there is no pair $(\alpha,\beta)$ such that $|\alpha\beta|=1,|\alpha\beta^{-1}|=0$ or $|\alpha\beta|=0,|\alpha\beta^{-1}|=1$. Indeed, for both two cases we have $|\alpha^2|=1$ or equivalently, $\alpha^2$ is a transposition. However, this is impossible since ${\rm sgn}(\alpha^2)={\rm sgn}(\alpha)^2=1$.
    We have thus the moment limit
        $$\lim_{d\to \infty}\mathbb{E}[(s\lambda_2^{(d,s)})^p]=\sum_{\pi\in \mathcal{P}_{1,2}(p)}s^{\# \pi},$$
    where we have identified idempotent permutations $\alpha \in S_p$ with partitions $\pi=\pi(\alpha) \in \mathcal P(p)$ having only blocks of size $1$ or $2$. Recall now that the normal distribution $\mathcal{N}(s,s)$ has classical cumulants 
    $$c_n=\begin{cases}s \quad \text{ if $n=1,2$,}\\ 0 \quad \text{ otherwise.} \end{cases}$$
    Hence the moment limit above corresponds to the $p$-th moment of the normal distribution $\mathcal{N}(s,s)$ by the classical moment-cumulant formula, finishing the proof of the case $s$ fixed.

    \medskip

    Let us now prove the statement (2) of the proposition, regarding the convergence $\lambda_{1,2}^{(d,s)} \to 1$ as $d \to \infty$ and $s(d)\to \infty$. We start with the case of $\lambda_1^{(d,s)}$. By Lemma \ref{lem:moments-lambda-123} and \ref{lem-polycoeff}, we have
    \begin{align*}
        \mathbb{E}\left[(\lambda_1^{(d,s)})^p\right]&=s^p\sum_{k=0}^{\infty}a_{k,p}s^{-k}\sum_{\beta}d^{p-|\beta|}(ds)^{-p-|\beta|}(\Mob(\beta)+O((ds)^{-2}))\\
        &=\sum_{k=0}^{\infty}a_{k,p}s^{-k}\left(1-\frac{p(p-1)}{2}d^{-2}s^{-1}+O((ds)^{-2})\right)\\
        &=\sum_{k=0}^{\infty}a_{k,p}s^{-k}+O(d^{-2})\\
        &=\sum_{k=0}^{N-1}a_{k,p}s^{-k}+O(s^{-N})+O(d^{-2})
    \end{align*}
    for any integers $p,N\geq 0$. 
    Then by Lemma \ref{lem-binomformula}, we can further estimate
    \begin{align*}
        \mathbb{E}\left[(\lambda_1^{(d,s)}-1)^{2N}\right]&=\sum_{p=0}^{2N}(-1)^{2N-p}\mathbb{E}[(\lambda_1^{(d,s)})^p]\binom{2N}{p}\\
        &=\sum_{k=0}^{N-1}\sum_{p=0}^{2N}(-1)^{2N-p}a_{k,p}\binom{2N}{p}s^{-k}+O(s^{-N})+O(d^{-2})\\
        &=O(s^{-N})+O(d^{-2}).
    \end{align*}
    In particular, this implies $\|\lambda_1^{(d,s)}-1\|_{L^{2N}}\to 0$ for all integers $N\geq 1$, and hence $\|\lambda_1^{(d,s)}-1\|_{L^p}\to 0$ for all $p\in (0,\infty)$, whenever $s,d\to \infty$. Now assuming $s\gtrsim d^t$ for some $t$, we can choose $N$ such that $Nt\geq 2$ so that above estimation again implies
        $$\mathbb{E}\left[(\lambda_1^{(d,s)}-1)^{2N}\right]=O(d^{-2}).$$
    Therefore, we have for any $\epsilon>0$,
        $$\sum_{d=1}^{\infty}\mathbb P\left(|\lambda_1^{(d,s)}-1|\geq \epsilon\right)\leq \sum_{d=1}^{\infty}\frac{\mathbb{E}[(\lambda_1^{(d,s)}-1)^{2N}]}{\epsilon^{2N}}<\infty$$
    by the Markov inequality. Then the Borel-Cantelli lemma implies that $\lambda_1^{(d,s)}\to 1$ almost surely.

    The proof for $\lambda_2^{(d,s)}$ is identical to the one above, up to replacing the constants $a_{k,p}$ by the explicit constants $b_{k,p}$ from \cref{lem-polycoeff}.
    
    \medskip
    
    Finally, let us prove the statement (3) of \cref{prop-UUparaconv}. We show only the assertions for $\lambda_1^{(d,s)}$ since the one for $\lambda_2^{(d,s)}$ can be analogously shown. Lemma \ref{lem:moments-lambda-123} implies that for every $\epsilon>0$,
    $$\mathbb{P}(|d^{-t}\lambda_1^{(d,s)}|>\epsilon)\leq \frac{1}{d^{Nt}\epsilon^N}\mathbb{E}[(\lambda_1^{(d,s)})^{N}]=s^{1-p}(s+1)\cdots (s+N-1)O(d^{-Nt})=O(d^{-2}),$$
for any choice of large $N$ such that $Nt\geq 2$; this proves the claim in the regime $d \to \infty$. On the other hand, for the regime $s \to \infty$, Lemma \ref{lem:moments-lambda-123} again implies 
    $$\mathbb{E}[(\lambda_1^{(d,s)})^p]=1+\frac{p(p-1)}{2}(1-d^{-2})s^{-1}+O(s^{-2}).$$
Thus, we have
\begin{align*}
    \mathbb{P}(|d^{-t}(\lambda_1^{(d,s)}-1)|>\epsilon)&\leq \frac{1}{d^{4t}\epsilon^{4}}\mathbb{E}[(\lambda_1^{(d,s)}-1)^4]\\
    &=\sum_{p=0}^4 (-1)^{4-p}\left(1+\frac{p(p-1)}{2}(1-d^{-2}s^{-1})\right)+O(s^{-2})=O(s^{-2})
\end{align*}
by Lemma \ref{lem-binomformula}. Now the usual Borel-Cantelli arguments again show the almost sure convergence $d^{-t}(\lambda_1^{(d,s)}-1)\to 0$ both for the cases $d\to \infty$ and $s\to \infty$, and the proof is complete.
\end{proof}

\bigskip

We now show the convergence result for the random variable $\lambda_3^{(d,s)}$.

\begin{proof}[Proof of \cref{prop:convergence-lambda3}]
The proof proceeds by a moment computation using the formula for $\lambda_3^{(d,s)}$ from \cref{lem:moments-lambda-123}:

\begin{align*}
    \mathbb{E}\left[(\lambda_3^{(d,s)})^p\right]&=\sum_{\alpha, \beta\in S_p} s^{\# \alpha} d^{\# (\alpha\vee\beta)}\Wg_{ds}(\alpha\beta^{-1})= 1 + \sum_{\substack{\alpha, \beta\in S_p\\(\alpha,\beta) \neq (\id,\id)}} s^{\# \alpha} d^{\# (\alpha\vee\beta)}\Wg_{ds}(\alpha\beta^{-1})\\
    &=1 + \sum_{\substack{\alpha, \beta\in S_p\\(\alpha,\beta) \neq (\id,\id)}}s^{p-|\alpha|}d^{p-|\alpha\vee \beta|}(ds)^{-p-|\alpha\beta^{-1}|}(\Mob(\alpha\beta^{-1})+O((ds)^{-2}))\\
    &=1+\sum_{\substack{\alpha,\beta\in S_p\\(\alpha,\beta) \neq (\id,\id)\\ |\alpha\vee\beta|+|\alpha\beta^{-1}|=1}}s^{-|\alpha|-|\alpha\beta^{-1}|}d^{-1}\Mob(\alpha\beta^{-1})+O(d^{-2})\\
    &=1+\sum_{\substack{\alpha,\beta\in S_p\\(\alpha,\beta) \neq (\id,\id)\\ |\alpha\vee\beta|=1,\,|\alpha\beta^{-1}|=0}}s^{-|\alpha|}d^{-1}+O(d^{-2})\\
    &=1+\frac{p(p-1)}{2}(ds)^{-1}+O(d^{-2}).
\end{align*}
Using \cref{lem-binomformula}, we can show
    $$\mathbb{E}\left[(\lambda_3^{(d,s)}-1)^{4}\right]=O(d^{-2}).$$
The almost sure convergence $\lambda_3^{(d,s)}\to 1$ follows by the Markov inequality and the Borel-Cantelli lemma.
\end{proof}

\subsection{Matrix parameters}\label{app:matrix}

This section contains results about the matrix parameters $A,B,C$ which appear in the study of random CDUC, DUC, and DOC quantum channels from \cref{sec:DOC}.

We begin with the description of average moments. Note that for our analysis, only the asymptotic eigenvalue distribution and moments of the matrix $C$ is needed. We give, for the sake of completeness, the formulas for the three matrices. 

\begin{lemma} \label{lem-ranDOCmatrixmoment}
The average $p$th moment of the random matrices $A,B,C$ are given by
\begin{align*}
    \mathbb{E}[\Tr(A^p)]&=\sum_{\alpha,\beta\in S_p} s^{\# \alpha}d^{\#((\gamma^{-1}\alpha\gamma) \vee \beta)}{\rm Wg}_{ds}(\alpha\beta^{-1}),\\
    \mathbb{E}[\Tr(B^p)]&=\sum_{\alpha,\beta\in S_p} s^{\# \alpha}d^{\#((\gamma^{-1}\alpha) \vee (\gamma^{-1}\beta))}{\rm Wg}_{ds}(\alpha\beta^{-1}),\\
    \mathbb{E}[\Tr(C^p)]&=\sum_{\alpha,\beta\in S_p} s^{\# \alpha}d^{\#((\gamma^{-1}\alpha) \vee (\gamma^{-1}\beta^{-1}))}{\rm Wg}_{ds}(\alpha\beta^{-1}),
\end{align*}
where $\gamma=(p\;p-1\; \cdots \; 2 \; 1)\in S_p$ is the full cycle. 
\end{lemma}
\begin{proof}
We shall prove the result only for the matrix $C$, since it is the only result we shall use later. The proofs for $A,B$ are similar and we leave them to the reader. 

To show the third formula above, we need to compute the expectation of the diagram in \cref{fig:E-Tr-C-p}, which is obtained by tracing the product of $p$ copies of the diagram for $C$ from \cref{fig-DOCpara}.

\begin{figure}[htb!] 
    \centering
    \includegraphics{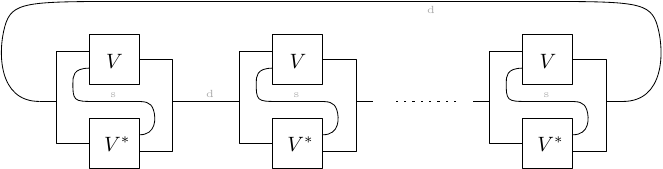}
    \caption{The diagram that corresponds to the quantity $\Tr(C^p)$.}
    \label{fig:E-Tr-C-p}
\end{figure}

To compute the expectation value of the diagram, we use the graphical Weingarten formula from \cite{collins2010randoma,CN16}. We obtain a sum indexed by a pair of permutations $\alpha, \beta \in S_p$ of diagrams $\mathcal D_{\alpha, \beta}$, weighted by the Weingarten function $\Wg_{ds}(\alpha\beta^{-1})$. Each diagram $\mathcal D_{\alpha, \beta}$ in the sum will be a collection of loops corresponding to the $\C{d}$ and $\C{s}$ vector spaces. Our goal is to count how many loops of each type appear in the diagram $\mathcal D_{\alpha, \beta}$. We refer to \cref{fig:D-alpha-beta-C-p} for the construction of the diagram $\mathcal D_{\alpha, \beta}$, obtained by erasing the $V$ and the $V^*$ boxes, and connecting the loose wires using $\alpha$ for the left side of the $V$ boxes and $\beta$ for the right side of the $V$ boxes. Note that in the Weingarten formula \eqref{eq:Wg}, the indices of $V$ and $\overline{V}$ are matched; in our setting, we connect to the opposite sides of $V^* = \overline{V}^\top$.

\begin{figure}[htb!] 
    \centering
    \includegraphics{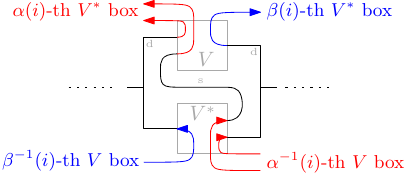}
    \caption{The way the $i$-th $V$ and $V^*$ boxes are connected, using the permutations $\textcolor{red}{\alpha},\textcolor{blue}{\beta} \in S_p$, to the rest of the diagram.}
    \label{fig:D-alpha-beta-C-p}
\end{figure}

The number of loops of type $s$ in $\mathcal D_{\alpha, \beta}$ is easily computed: only the permutation $\alpha$ is involved in connecting the $s$-type wires, and it is easy to see that each cycle of $\alpha$ yields a connected component (loop) in the end. We have thus a contribution of $s^{\#\alpha}$. The computation of the loops of size $d$ is more complicated, because of the presence of the ``vertical'' connection between the $i$-th $V$ and $V^*$ boxes. In the absence of this vertical connection, each $V$ box would have been connected to the previous $V^*$ box, resulting in a total number of $\#(\gamma^{-1}\alpha) +\#(\gamma^{-1}\beta^{-1})$ $d$-loops. The vertical wire connects some of the resulting loops, yielding the announced number of connected components (loops).  
\end{proof}

We now prove the formulas for the limiting eigenvalue distributions of the random matrix $C$.

\begin{proof}[Proof of \cref{prop-ranDOCmatrix}]
First of all, suppose $s/d\to 0$ as $d\to \infty$. Our method is to show that every mean $p$-th moment of $(ds)^{1/2}C$ converges to that of the semicircular distribution $SC_{0,1}$, and its variance converges to 0. First, by applying Lemma \ref{lem-ranDOCmatrixmoment},
{\small
\begin{align*}
    \frac{1}{d}\mathbb{E}\left[\Tr((ds)^{p/2}C^p)\right]&=\sum_{\alpha,\beta\in S_p}s^{p/2-|\alpha|-|\alpha\beta^{-1}|}d^{p/2-1-|\gamma^{-1}\alpha \vee \gamma^{-1}\beta^{-1}|-|\alpha\beta^{-1}|}\left(\Mob(\alpha\beta^{-1})+O((ds)^{-2})\right)\\
    &=\Bigg[\sum_{|\alpha|+|\alpha\beta^{-1}|\leq p/2} \left(\frac{s}{d}\right)^{p/2-|\alpha|-|\alpha\beta^{-1}|}d^{p-1-|\alpha|-|\gamma^{-1}\alpha \vee \gamma^{-1}\beta^{-1}|-2|\alpha\beta^{-1}|}\\
    &\!\!\!\!\!\!\!\!\!\!+ \sum_{|\alpha|+|\alpha\beta^{-1}|> p/2}s^{p/2-|\alpha|-|\alpha\beta^{-1}|}d^{p/2-1-|\gamma^{-1}\alpha \vee \gamma^{-1}\beta^{-1}|-|\alpha\beta^{-1}|}\Bigg]\left(\Mob(\alpha\beta^{-1})+O((ds)^{-2})\right)\\
\end{align*}
}
Note that by \cref{lem:permutation-optimization},
    $$\frac{p}{2}-1-|\gamma^{-1}\alpha \vee \gamma^{-1}\beta^{-1}|-|\alpha\beta^{-1}|\leq \frac{p}{2}-1-\frac{1}{2}|\gamma^{-2}|=\begin{cases}0 & \text{if $p$ is even,}\\ -1/2 & \text {if $p$ is odd,} \end{cases}$$
and when $p$ is even, the equality condition further implies $\alpha=\beta\in NC_{2}(p)$ (\cref{eq-perm-opt1}) which contradicts the condition $|\alpha|>p/2$. Therefore, the second sum becomes $O(s^{-1/2}d^{-1/2})$ regardless of $p$. Moreover, regarding the first sum,
    $$p-1-|\alpha|-|\gamma^{-1}\alpha \vee \gamma^{-1}\beta^{-1}|-2|\alpha\beta^{-1}|\leq p-1-|\gamma|=0,$$
and the equality holds iff $\alpha=\beta\in S_{NC}(\gamma)\cong NC_{1,2}(p)$. Therefore, we can further estimate
\begin{align*}
    \frac{1}{d}\mathbb{E}\left[\Tr((ds)^{p/2}C^p)\right]&=\sum_{\alpha\in NC_{1,2}(p)}\left(\frac{s}{d}\right)^{p/2-|\alpha|}+O(d^{-2})+O((ds)^{-1/2})\\
    &\to \begin{cases} \sum_{\alpha\in NC_2(p)}1={\rm Cat}_{p/2} & \text{if $p$ is even,}\\ 0  & \text{if $p$ is odd,}     
    \end{cases}
\end{align*}
as $d\to \infty$. On the other hand, if we define $\tilde{\gamma}:=(2p\,2p-1\;\ldots\; p+1)(p\;p-1\;\ldots\; 2\;1)\in S_{2p}$ and apply the Weingarten calculus, we have
    $$\mathbb{E}\left[\frac{1}{d^2}\left(\Tr((ds)^{p/2}C^p)\right)^2\right]=\sum_{\alpha,\beta\in S_{2p}}s^{p-|\alpha|-|\alpha\beta^{-1}|}d^{p-2-|\tilde{\gamma}^{-1}\alpha \vee \tilde{\gamma}^{-1}\beta^{-1}|-|\alpha\beta^{-1}|}\left(\Mob(\alpha\beta^{-1})+O((ds)^{-2})\right).$$
By repeating the prior argument, we can show that
\begin{align*}
    \mathbb{E}\left[\frac{1}{d^2}\left(\Tr((ds)^{p/2}C^p)\right)^2\right]&=\sum_{\alpha\in S_{NC_{1,2}}(\tilde{\gamma})}\left(\frac{s}{d}\right)^{p-|\alpha|}+O(d^{-2})+O((ds)^{-1})\\
    &=\left(\sum_{\pi\in NC_{1,2}(p)}\left(\frac{s}{d}\right)^{p/2-|\pi|}\right)^2+O(d^{-1}),
\end{align*}
where the last equality follows from $S_{NC_{1,2}}(\tilde{\gamma})\cong N_{1,2}(p)\times N_{1,2}(p)$ (\cref{eq-NC12-bij}). Therefore, we obtain
    $$\mathbb{E}\left[\left(\frac{1}{d}\Tr((ds)^{p/2}C^p)\right)^2\right]-\left(\frac{1}{d}\mathbb{E}\left[\Tr((ds)^{p/2}C^p)\right]\right)^2=O(d^{-1/2})\to 0$$
for all $p$, hence proving that for every $p\geq 1$,
    $$\frac{1}{d}\Tr((ds)^{p/2}C^p) \to \text{$p$-th moment of $SC_{0,1}$}$$
in probability.
    
Now let us show that the convergence above holds almost surely for $p=1,2,3$. First, note that \cref{eq-C-lambda3} and \cref{prop:convergence-lambda3} implies that
    $$\frac{1}{d}\Tr((ds)^{1/2}C)=\left(\frac{s}{d}\right)^{1/2}\lambda_3^{(d,s)}\to 0$$
almost surely, regardless of the scale of $s$. Furthermore, by combining the previous arguments with \cref{lem:permutation-optimization2}, we have that
\begin{align*}
    \frac{1}{d}\mathbb{E}\left[\Tr((ds)^{p/2}C^p)\right]&=\sum_{\alpha\in NC_{1,2}(p)}\left(\frac{s}{d}\right)^{p/2-|\alpha|}+O(d^{-3/2}),\\
    \mathbb{E}\left[\frac{1}{d^2}\left(\Tr((ds)^{p/2}C^p)\right)^2\right]&=\left(\sum_{\alpha\in NC_{1,2}(p)}\left(\frac{s}{d}\right)^{p/2-|\alpha|}\right)^2+O(d^{-2})
\end{align*}
hold for $p=2,3$. Therefore, the usual Borel-Cantelli technique shows the wanted almost sure convergences. 

\medskip

The remaining convergences can be obtained in analogous way. For the regime $s\sim cd$ for $c>0$, \cref{lem-ranDOCmatrixmoment} and \ref{lem:permutation-optimization} again imply
\begin{align*}
    \frac{1}{d}\mathbb{E}\left[\Tr(s^pC^p)\right]& = \sum_{\alpha,\beta\in S_p}\left(\frac{s}{d}\right)^{\# \alpha-|\alpha\beta^{-1}|}d^{p-1-|\alpha|-|\gamma^{-1}\alpha \vee \gamma^{-1}\beta^{-1}|-2|\alpha\beta^{-1}|}\left(\Mob(\alpha\beta^{-1})+O(d^{-4})\right)\\
    &=\sum_{\pi\in NC_{1,2}(p)}\left(\frac{s}{d}\right)^{\# \pi}+O(d^{-2})\\
    &\to \sum_{\pi\in NC_{1,2}(p)}c^{\# \pi} \text{ as $d\to \infty$},
\end{align*}
\begin{align*}
    \mathbb{E}\left[\left(\frac{1}{d}\Tr(s^pC^p)\right)^2\right]&= \sum_{\alpha,\beta\in S_{2p}}\left(\frac{s}{d}\right)^{\# \alpha-|\alpha\beta^{-1}|}d^{2p-2-|\alpha|-|\tilde{\gamma}^{-1}\alpha \vee \tilde{\gamma}^{-1}\beta^{-1}|-2|\alpha\beta^{-1}|}\left(\Mob(\alpha\beta^{-1})+O(d^{-4})\right)\\
    &=\left(\sum_{\pi\in NC_{1,2}(p)}\left(\frac{s}{d}\right)^{\# \pi}\right)^2+O(d^{-2}),
\end{align*}
and therefore,
    $$\mathbb{E}\left[\left(\frac{1}{d}\Tr((ds)^{p/2}C^p)\right)^2\right]-\left(\frac{1}{d}\mathbb{E}\left[\Tr((ds)^{p/2}C^p)\right]\right)^2=O(d^{-2}).$$
This shows that $\frac{1}{d}\Tr(s^pC^p)\to \sum_{\pi\in NC_{1,2}(p)}c^{\# \pi}$ almost surely for every $p\geq 0$. Note that this limit value is the $p$-th moment of the semicircular distribution ${\rm SC}_{c,c}$.

\medskip

Fianlly, in the regime $d\to \infty$ and $s/d\to \infty$, we similarly have
\begin{align}
    \frac{1}{d}\mathbb{E}\left[\Tr(d^pC^p)\right]&= \sum_{\alpha,\beta\in S_p}\left(\frac{s}{d}\right)^{-|\alpha|-|\alpha\beta^{-1}|}d^{p-1-|\alpha|-|\gamma^{-1}\alpha \vee \gamma^{-1}\beta^{-1}|-2|\alpha\beta^{-1}|}\left(\Mob(\alpha\beta^{-1})+O((ds)^{-2})\right) \nonumber \\
    &=\sum_{\alpha\in NC_{1,2}(p)}\left(\frac{s}{d}\right)^{-|\alpha|}+O((ds)^{-1}) \label{eq-C-eigenvalue}\\
    &\to 1 \text{ as $d\to \infty$} \nonumber,
\end{align}
\begin{align*}
    \mathbb{E}\left[\left(\frac{1}{d}\Tr(d^pC^p)\right)^2\right]&= \sum_{\alpha,\beta\in S_{2p}}\left(\frac{s}{d}\right)^{-|\alpha|-|\alpha\beta^{-1}|}d^{2p-2-|\alpha|-|\tilde{\gamma}^{-1}\alpha \vee \tilde{\gamma}^{-1}\beta^{-1}|-2|\alpha\beta^{-1}|}\left(\Mob(\alpha\beta^{-1})+O((ds)^{-2})\right)\\
    &=\left(\sum_{\pi\in NC_{1,2}(p)}\left(\frac{s}{d}\right)^{-|\pi|}\right)^2+O(d^{-2}),
\end{align*}
and therefore,
    $$\mathbb{E}\left[\left(\frac{1}{d}\Tr(d^pC^p)\right)^2\right]-\left(\frac{1}{d}\mathbb{E}\left[\Tr(d^pC^p)\right]\right)^2=O(d^{-2}).$$
\end{proof}

\begin{proof}[Proof of \cref{prop-ranDOCmatrix2}]
By the Borel-Cantelli argument, it suffices to show that for any $\epsilon>0$,
    $$\mathbb{P}\left(\|dC-I_d\|_{\infty}>\epsilon\right)=O(d^{-(1+t)}).$$
Let us begin with the mean moment formula \cref{eq-C-eigenvalue} of the random matrix $dC$. Applying \cref{lem-polycoeff} (3), we further obtain
\begin{align*}
    \frac{1}{d}\mathbb{E}\left[\Tr(d^pC^p)\right]    &=\sum_{k=0}^{\infty}c_{k,p}\left(\frac{s}{d}\right)^{-k}+O((ds)^{-1})\\
    &=\sum_{k=0}^{N-1}c_{k,p}\left(\frac{s}{d}\right)^{-k}+O\left((s/d)^{-N}\right)+O((ds)^{-1})
\end{align*}
For arbitrary $N\geq 1$, where $c_{k,p}={\rm Cat}_k\binom{p}{2k}$. If we choose $N$ such that $Nt\geq 2+t$, then
    $$\mathbb{E}\left[\Tr(d^pC^p)\right]=\sum_{k=0}^{N-1}c_{k,p}d\left(\frac{s}{d}\right)^{-k}+O(d^{-(1+t)}),$$
Therefore, \cref{lem-binomformula} implies that
\begin{align*}
    \mathbb{E}\left[\Tr ((dC-I_d)^{2N})\right]&=\sum_{p=0}^{2N}(-1)^{2N-p}\mathbb{E}\left[\Tr (dC)^p\right]\binom{2N}{p}\\
    &=\sum_{k=0}^{N-1}d\left(\frac{s}{d}\right)^{-k}\sum_{p=0}^{2N}(-1)^{2N-p}c_{k,p}\binom{2N}{p}+O(d^{-(1+t)})\\
    &=O(d^{-(1+t)}).
\end{align*}
Consequently, we obtain by Markov's inequality and the matrix norms $\|\cdot\|_{\infty}\leq \|\cdot\|_{2N}$,
\begin{align*}
    \mathbb{P}(\|dC-I_d\|_{\infty}>\epsilon)\leq \mathbb{P}\left(\Tr((dC-I_d)^{2N})>\epsilon^{2N}\right)\leq \frac{\mathbb{E}\left[\Tr((dC-I_d)^{2N})\right]}{\epsilon^{2N}}=O(d^{-(1+t)}).
\end{align*}
\end{proof}

Let us now consider the matrix elements of the three parameters $A,B,C$. We first compute the moments. 

\begin{lemma} \label{lem-ranDOCentry}
For $i,j$, the $p$th moment of the entry $A_{ij}$ is
    $$\mathbb{E}\left[(A_{ij})^p\right]=\frac{s(s+1)\cdots (s+p-1)}{ds(ds+1)\cdots (ds+p-1)}.$$
Moreover, for $i\neq j$, the $p$th moment of the values $|B_{ij}|^2$ and $|C_{ij}|^2$ are
    $$\mathbb{E}\left[|B_{ij}|^{2p}\right]= \mathbb{E}\left[|C_{ij}|^{2p}\right]
    =p!s(s+1)\cdots (s+p-1)\sum_{\beta\in \mathcal{S}} {\rm Wg}_{ds}(\beta),$$
where $\mathcal{S}=\set{\alpha\in S_{2p}: \alpha(i)\equiv i\;{\rm mod}\, 2 \text{ for all } i}$.
In particular, $|B_{ij}|$ and $|C_{ij}|$ have the same distribution.
\end{lemma}
\begin{proof}
\begin{figure}[htb!] 
    \centering
    \includegraphics{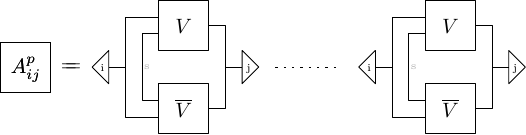}\\\bigskip
    \includegraphics{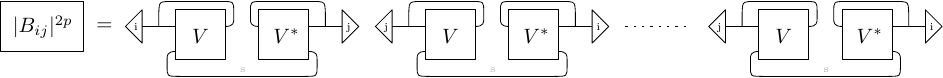}\\\bigskip
    \includegraphics{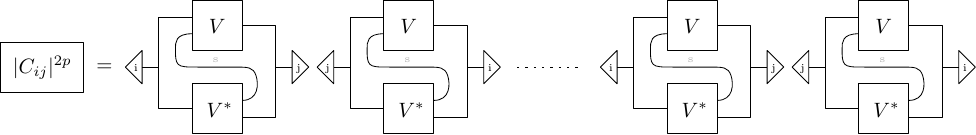}
    \caption{Diagrams for the for the moments of the entries of the random matrices $A,B,C$. The first diagram contains $p$ connected sub-diagrams, while the ones for $B$ and $C$ contain $2p$ connected sub-diagrams.}
    \label{fig-DOCentry}
\end{figure}

(Figure \ref{fig-DOCentry}) By Weingarten calculus and Lemma \ref{lem-weingarten}, we have
    $$\mathbb{E}\left[(A_{ij})^p\right]= \sum_{\alpha,\beta\in S_p}s^{\# \alpha} {\rm Wg}_{ds}(\alpha\beta^{-1})=\sum_{\alpha\in S_p}s^{\# \alpha}\sum_{\beta\in S_p} {\rm Wg}_{ds}(\beta)=\frac{s(s+1)\cdots (s+p-1)}{ds(ds+1)\cdots (ds+p-1)}.$$
Next, since $B_{ij}=\overline{B_{ji}}$, we have
    $$\mathbb{E}\left[|B_{ij}|^{2p}\right]= \mathbb{E}\left[(B_{ij}B_{ji})^p\right]=\sum_{\alpha,\beta\in \mathcal{T}}s^{\# \alpha} {\rm Wg}_{ds}(\alpha\beta^{-1})=\sum_{\alpha\in \mathcal{T}}s^{\# \alpha} \sum_{\beta\in \mathcal{S}} {\rm Wg}_{ds}(\beta),$$
where $\mathcal{T}=\set{\alpha\in S_{2p}: \alpha(i)\equiv i+1\;{\rm mod}\, 2 \text{ for all } i}$. Note the the last equality follows from the observation that for any $\alpha\in \mathcal{T}$, the map $\beta\in \mathcal{T}\mapsto \alpha\beta^{-1}\in \mathcal{S}$ is a bijection. We claim that the map 
    $$\varphi: \alpha\in \mathcal{T}\mapsto \alpha^2\big|_{\{1,3,\ldots, 2p-1\}}\in S(\{1,3,\ldots, 2p-1\})$$
is well-defined $p!$-to-$1$ surjection and $\#(\varphi(\alpha))=\# \alpha$. Indeed, the fact that the map $\varphi$ is well-defined is clear from the condition $\alpha(i)\equiv i+1\; {\rm mod\,} 2$. Moreover, for arbitrary $\sigma\in S(\{1,3,\ldots, 2p-1\})$ and a pair partition $\pi\in \mathcal{P}_2(2p)$ with $\pi(i)\equiv i+1 \; {\rm mod\,}2$ for all $i$, we can construct a unique $\alpha_{\sigma,\pi} \in \mathcal{T}$ satisfying
    $$\alpha_{\sigma,\pi}(2i-1)=\pi(2i-1), \quad \alpha_{\sigma,\pi}(\pi(2i-1))=\sigma(2i-1)$$
so that $\varphi(\alpha_{\sigma,\pi})=\sigma$. One can easily check
    $$\varphi^{-1}(\sigma)=\{\alpha_{\sigma,\pi}: \text{$\pi\in \mathcal{P}_2(2p)$ and $\pi(i)\equiv i+1 \; {\rm mod\,}2$ for all $i$} \}$$
and $\# (\alpha_{\sigma,\pi})=\# \sigma$ regardless of the choice of $\pi$, and therefore the claim is shown.
Now from the claim and \cref{lem-weingarten}, we obtain that
\begin{align*}
    \sum_{\alpha\in \mathcal{T}}s^{\# \alpha}=\sum_{\alpha\in \mathcal{T}}s^{\# \varphi(\alpha)}=p!\sum_{\sigma\in S_p} s^{\# \sigma}=p!s(s+1)\cdots (s+p-1),
\end{align*}
hence proving the moment formula $\mathbb{E}[|B_{ij}|^{2p}]$.

The moment $\mathbb{E}[|C_{ij}|^{2p}]$ is calculated by the same method. These moments satisfy Carleman's condition \cite[Chapter 2, Addenda and Problems, 11]{akhiezer2020classical}, so the distributions of $|B_{ij}|$ and $|C_{ij}|$ are the same by uniqueness of the moments.
\end{proof}

This allows us to prove the limit result for the entries of $A,B,C$.

\begin{proof}[Proof of \cref{prop:DOC-entry-limit-fixed}]
(1) and (2) is clear from
\begin{align*}
    \lim_{d\to \infty} \mathbb{E}[(dsA_{ij})^p]&=\lim_{d\to\infty}\frac{(ds)^p s(s+1)\cdots (s+p-1)}{ds(ds+1)\cdots (ds+p-1)}= s(s+1)\cdots (s+p-1),\\
    \lim_{d\to \infty} \mathbb{E}[|dsB_{ij}|^{2p}]=\lim_{d\to \infty} \mathbb{E}[|dsC_{ij}|^{2p}]
    &=\lim_{d\to \infty} p!(ds)^p s(s+1)\cdots (s+p-1)(ds)^{-p}(1+O((ds)^{-1}))\\
    &=p!s(s+1)\cdots (s+p-1),
\end{align*}
by \cref{lem-ranDOCentry}, and the fact that $s(s+1)\cdots (s+p-1)$ and $p!$ are $p$-th moment of the Gamma distribution $\Gamma(s,1)$ and the exponential distribution $\mathsf{Exp}(1)$, as noted in the proof of \cref{prop-UUparaconv}.

For the proof of (3), let us recall from \cref{eq-DOCentry} that the random variables
    $$A_{i_nj_n}=\sum_{k=1}^s |V_{i_nj_n}^{(k)}|^2, \quad A_{j_ni_n}=\sum_{k=1}^s |V_{i_nj_n}^{(k)}|^2, \quad B_{i_nj_n}=\sum_{k=1}^s V_{i_ni_n}^{(k)}\overline{V_{j_nj_n}^{(k)}},$$
for $n=1,\ldots, N$, are functions of the $2s\times 2N$-submatrix $V^{(2s\times 2N)}$ of $V$, where $V^{(2s\times 2N)}$ can be written as a block matrix $\left(V_{i_1 j_1   }^{(2s\times 2)}\; V_{i_2 j_2   }^{(2s\times 2)}\; \cdots \; V_{i_N j_N}^{(2s\times 2)}\right)$ for
    $$V_{ij}^{(2s\times 2)}:=\begin{pmatrix}
    V_{ii}^{(1)} & V_{ij}^{(1)} \\ \vdots & \vdots \\ V_{ii}^{(s)} & V_{ij}^{(s)} \\ V_{ji}^{(1)} & V_{jj}^{(1)} \\ \vdots & \vdots \\ V_{ji}^{(s)} & V_{jj}^{(s)} \end{pmatrix}.$$
However, it is known \cite{PR04, MT07} that if $s$ and $N$ are fixed, the random matrix $\sqrt{ds}V^{(2s\times 2N)}$ converges in distribution to the $(2s\times 2N)$ Ginibre ensemble whose entries are standard normal random variable as $d\to\infty$. In particular, the entries of $\sqrt{ds}V^{(2s\times 2N)}$ are asymptotically independent. Now the assertion follows from the observation that all the summands defining $\{A_{i_nj_n},A_{j_ni_n},B_{i_nj_n}\}_{n=1}^N$ are distinct.
\end{proof}

\begin{proof}[Proof of \cref{prop:DOC-entry-limit}]
(1) The proof is similar with that of \cref{prop-UUparaconv} (2). Applying Lemma \ref{lem-ranDOCentry} and Lemma \ref{lem-polycoeff},
\begin{align*}
    \mathbb{E}\left[(dA_{ij})^p\right]
    &=\frac{d^ps(s+1)\cdots (s+p-1)}{ds(ds+1)\cdots (ds+p-1)}\\
    &=\left(\sum_{k=0}^{\infty}a_{k,p}s^{-k}\right)\left(1+\sum_{k=1}^{\infty}a_{k,p}(ds)^{-k}\right)^{-1}\\
    &=\left(\sum_{k=0}^{\infty}a_{k,p}s^{-k}\right) \left(1-a_{1,p} (ds)^{-1}+(a_{1,p}^2-a_{2,p})(ds)^{-2}+O((ds)^{-3})\right)
\end{align*}
for $p\geq 0$. If we choose large $N$ such that $Nt\geq 3+t$, then we can further proceed
    $$\mathbb{E}\left[(dA_{ij})^p\right]=\sum_{k=0}^{N-1}s^{-k}a_{k,p}\left(1-a_{1,p}(ds)^{-1}+(a_{1,p}^2-a_{2,p})(ds)^{-2}\right)+O(d^{-(3+t)}).$$
Therefore, Lemma \ref{lem-binomformula} implies that
\begin{align*}
    \mathbb{E}\left[(dA_{ij}-1)^{2N+4}\right]&=\sum_{p=0}^{2N+4}(-1)^{2N+4-p}\mathbb{E}[(dA_{ij})^p]\binom{2N+4}{p}\\
    &=\sum_{k=0}^{N-1}s^{-k}\sum_{p=0}^{2N+4}(-1)^{2N+4-p}a_{k,p}F(p)\binom{2N+4}{p}+O(d^{-(3+t)})\\
    &=O(d^{-(3+t)}),
\end{align*}
for all $i,j$, where $F(p)=1-a_{1,p}(ds)^{-1}+(a_{1,p}^2-a_{2,p})(ds)^{-2}$ is a polynomial in $p$ of degree $4$. Now by Markov's inequality,
\begin{align*}
    \mathbb{P}\left(\max_{i,j}|dA_{ij}-1|\geq \epsilon\right)\leq \sum_{i,j}\mathbb{P}\left(|dA_{ij}-1|\geq \epsilon\right)\leq \sum_{i,j}\frac{\mathbb{E}[(dA_{ij}-1)^{2N+4}]}{\epsilon^{2N+4}}=O\left(d^{-(1+t)}\right)
\end{align*}
for all $\epsilon>0$. Therefore, $\max_{i,j}|dA_{ij}-1|\to 0$ (and hence $\max_{i,j} dA_{ij}\to 1$ and $\min_{i,j} dA_{ij}\to 1$) almost surely by the Borel-Cantelli lemma.

For (2) and (3), we again apply Lemma \ref{lem-ranDOCentry} to have
\begin{align*}
    \mathbb{E}\left[|d^{1+t'/2}B_{ij}|^{2p}\right]=\mathbb{E}\left[|d^{1+t'/2}C_{ij}|^{2p}\right]&=d^{p(2+t')}s(s+1)
    \cdots (s+p-1)\sum_{\beta\in \mathcal{S}} {\rm Wg}_{ds}(\beta)\\
    &=d^{p(2+t')}p!s^p(1+O(s^{-1}))(ds)^{-2p}(1+O((ds)^{-2})\\
    &=O(d^{pt'}s^{-p})=O(d^{-p(t-t')}).
\end{align*}
for $p\geq 0$. If we choose $N$ such that $N(t-t')\geq 4$, we further have
\begin{align*}
    \mathbb{P}\left(\max_{i,j}|d^{1+t'/2}B_{ij}|\geq \epsilon\right)&\leq \sum_{i,j}\mathbb{P}\left(|d^{1+t'/2}B_{ij}|\geq \epsilon\right)\\
    &\leq \sum_{i,j}\frac{\mathbb{E}[|d^{1+t'/2}B_{ij}|^{2N}]}{\epsilon^{2N}}=d^2 O(d^{-N(t-t')})=O(d^{-2}),
\end{align*}
and similarly $\mathbb{P}\left(\max_{i,j}|d^{1+t'/2}C_{ij}|\geq \epsilon\right)=O(d^{-2})$ for all $\epsilon>0$. This concludes the proof.
\end{proof}

\bibliography{references}

\newcommand{\etalchar}[1]{$^{#1}$}
\begin{thebibliography}{KNP{\etalchar{+}}21b}

\bibitem[Akh20]{akhiezer2020classical}
Naum~Ilich Akhiezer.
\newblock {\em The classical moment problem and some related questions in
  analysis}.
\newblock SIAM, 2020.

\bibitem[AN12]{aubrun2012realigning}
Guillaume Aubrun and Ion Nechita.
\newblock Realigning random states.
\newblock {\em Journal of Mathematical Physics}, 53(10):102210, 2012.

\bibitem[AN14]{Al14}
Muneerah Al~Nuwairan.
\newblock The extreme points of {SU}(2)-irreducibly covariant channels.
\newblock {\em Internat. J. Math.}, 25(6):1450048, 30, 2014.

\bibitem[ASY14]{ASY14}
Guillaume Aubrun, Stanis\l aw~J. Szarek, and Deping Ye.
\newblock Entanglement thresholds for random induced states.
\newblock {\em Comm. Pure Appl. Math.}, 67(1):129--171, 2014.

\bibitem[Aub12]{Aub12}
Guillaume Aubrun.
\newblock Partial transposition of random states and non-centered semicircular
  distributions.
\newblock {\em Random Matrices Theory Appl.}, 1(2):1250001, 29, 2012.

\bibitem[BCLY20]{BCLY20}
Michael Brannan, Beno{\^i}t Collins, Hun~Hee Lee, and Sang-Gyun Youn.
\newblock Temperley--lieb quantum channels.
\newblock {\em Communications in Mathematical Physics}, 376(2):795--839, Jun
  2020.

\bibitem[BCN16]{BCN16}
Serban~T. Belinschi, Beno\^{i}t Collins, and Ion Nechita.
\newblock Almost one bit violation for the additivity of the minimum output
  entropy.
\newblock {\em Comm. Math. Phys.}, 341(3):885--909, 2016.

\bibitem[BCS20]{BCS20}
Ivan Bardet, Beno\^{i}t Collins, and Gunjan Sapra.
\newblock Characterization of equivariant maps and application to entanglement
  detection.
\newblock {\em Ann. Henri Poincar\'{e}}, 21(10):3385--3406, 2020.

\bibitem[BCS{\.Z}09]{bruzda2009random}
Wojciech Bruzda, Valerio Cappellini, Hans-J{\"u}rgen Sommers, and Karol
  {\.Z}yczkowski.
\newblock Random quantum operations.
\newblock {\em Physics Letters A}, 373(3):320--324, 2009.

\bibitem[BO08]{BrOz}
Nathanial~P. Brown and Narutaka Ozawa.
\newblock {\em {$C^*$}-algebras and finite-dimensional approximations},
  volume~88 of {\em Graduate Studies in Mathematics}.
\newblock American Mathematical Society, Providence, RI, 2008.

\bibitem[Bra37]{Bra37}
Richard Brauer.
\newblock On algebras which are connected with the semisimple continuous
  groups.
\newblock {\em Annals of Mathematics}, 38(4):857--872, 1937.

\bibitem[BSM03]{berman2003completely}
Abraham Berman and Naomi Shaked-Monderer.
\newblock {\em Completely positive matrices}.
\newblock World Scientific, 2003.

\bibitem[Chr12]{PPTsq}
M.~Christandl.
\newblock P{PT} square conjecture.
\newblock {\em Banff International Research Station Workshop: \emph{Operator
  Structures in Quantum Information Theory}}, 2012.

\bibitem[CK06]{chruscinski2006class}
Dariusz Chru{\'s}ci{\'n}ski and Andrzej Kossakowski.
\newblock Class of positive partial transposition states.
\newblock {\em Physical Review A}, 74(2):022308, 2006.

\bibitem[CMHW18]{Christandl2018}
Matthias Christandl, Alexander M{\"u}ller-Hermes, and Michael~M. Wolf.
\newblock When do composed maps become entanglement breaking?
\newblock {\em Annales Henri Poincar{\'e}}, 20:2295--2322, 2018.

\bibitem[CN10]{collins2010randoma}
Beno{\^i}t Collins and Ion Nechita.
\newblock Random quantum channels {{I}}: Graphical calculus and the {{Bell}}
  state phenomenon.
\newblock {\em Communications in Mathematical Physics}, 297(2):345--370, 2010.

\bibitem[CN15]{CN16}
Benoît Collins and Ion Nechita.
\newblock {Random matrix techniques in quantum information theory}.
\newblock {\em Journal of Mathematical Physics}, 57(1):015215, 12 2015.

\bibitem[Col03]{collins2003moments}
Beno{\^i}t Collins.
\newblock Moments and cumulants of polynomial random variables on
  unitarygroups, the {{Itzykson-Zuber}} integral, and free probability.
\newblock {\em International Mathematics Research Notices}, 2003(17):953--982,
  2003.

\bibitem[COS18]{COS18}
Beno\^{i}t Collins, Hiroyuki Osaka, and Gunjan Sapra.
\newblock On a family of linear maps from {$M_n(\mathbb C)$} to
  {$M_{n^2}(\mathbb C)$}.
\newblock {\em Linear Algebra Appl.}, 555:398--411, 2018.

\bibitem[CP22]{CP22}
Benoît Collins and Félix Parraud.
\newblock {Concentration estimates for random subspaces of a tensor product and
  application to quantum information theory}.
\newblock {\em Journal of Mathematical Physics}, 63(10):102202, 10 2022.

\bibitem[C{\'S}06]{collins2006integration}
Beno{\^i}t Collins and Piotr {\'S}niady.
\newblock Integration with respect to the {{Haar}} measure on unitary,
  orthogonal and symplectic group.
\newblock {\em Communications in Mathematical Physics}, 264(3):773--795, 2006.

\bibitem[CW03]{chen2003matrix}
Kai Chen and Ling-An Wu.
\newblock A matrix realignment method for recognizing entanglement.
\newblock {\em Quantum Information \& Computation}, 3(3):193--202, 2003.

\bibitem[CYT19]{Chen2019}
Lin Chen, Yu~Yang, and Wai-Shing Tang.
\newblock Positive-partial-transpose square conjecture for $n=3$.
\newblock {\em Phys. Rev. A}, 99:012337, Jan 2019.

\bibitem[CYZ18]{collins2018ppt}
Beno{\^\i}t Collins, Zhi Yin, and Ping Zhong.
\newblock The {PPT} square conjecture holds generically for some classes of
  independent states.
\newblock {\em Journal of Physics A: Mathematical and Theoretical},
  51(42):425301, 2018.

\bibitem[DHS06]{DHS06}
NILANJANA DATTA, ALEXANDER~S. HOLEVO, and YURI SUHOV.
\newblock Additivity for transpose depolarizing channels.
\newblock {\em International Journal of Quantum Information}, 04(01):85--98,
  2006.

\bibitem[DS94]{diaconis1994eigenvalues}
Persi Diaconis and Mehrdad Shahshahani.
\newblock On the eigenvalues of random matrices.
\newblock {\em Journal of Applied Probability}, 31(A):49--62, 1994.

\bibitem[DSS{\etalchar{+}}00]{DVSS+00}
David~P. DiVincenzo, Peter~W. Shor, John~A. Smolin, Barbara~M. Terhal, and
  Ashish~V. Thapliyal.
\newblock Evidence for bound entangled states with negative partial transpose.
\newblock {\em Phys. Rev. A}, 61:062312, May 2000.

\bibitem[DTW16]{DTW16}
Nilanjana Datta, Marco Tomamichel, and Mark~M Wilde.
\newblock On the second-order asymptotics for entanglement-assisted
  communication.
\newblock {\em Quantum Information Processing}, 15(6):2569--2591, 2016.

\bibitem[GBW21]{GBW21}
Martina Gschwendtner, Andreas Bluhm, and Andreas Winter.
\newblock Programmability of covariant quantum channels.
\newblock {\em {Quantum}}, 5:488, June 2021.

\bibitem[GKS20]{girard2020convex}
Mark Girard, Seung-Hyeok Kye, and Erling St{\o}rmer.
\newblock Convex cones in mapping spaces between matrix algebras.
\newblock {\em Linear Algebra and its Applications}, 608:248--269, 2020.

\bibitem[Has09]{Has09}
M.~B. Hastings.
\newblock Superadditivity of communication capacity using entangled inputs.
\newblock {\em Nature Physics}, 5(4):255--257, Apr 2009.

\bibitem[Hol93]{holevo1993note}
Alexander~S Holevo.
\newblock A note on covariant dynamical semigroups.
\newblock {\em Reports on mathematical physics}, 32(2):211--216, 1993.

\bibitem[Hol96]{holevo1996covariant}
Alexander~S Holevo.
\newblock Covariant quantum markovian evolutions.
\newblock {\em Journal of Mathematical Physics}, 37(4):1812--1832, 1996.

\bibitem[Hol05]{Hol05}
A.~S. Holevo.
\newblock Additivity conjecture and covariant channels.
\newblock {\em International Journal of Quantum Information}, 03(01):41--47,
  2005.

\bibitem[HRF20]{hanson2020eventually}
Eric~P Hanson, Cambyse Rouz{\'e}, and Daniel~Stilck Fran{\c{c}}a.
\newblock Eventually entanglement breaking markovian dynamics: Structure and
  characteristic times.
\newblock {\em Ann. Henri Poincar{\'e}}, 21:1517--1571, 2020.

\bibitem[HSR03]{horodecki2003entanglement}
Michael Horodecki, Peter~W Shor, and Mary~Beth Ruskai.
\newblock Entanglement breaking channels.
\newblock {\em Reviews in Mathematical Physics}, 15(06):629--641, 2003.

\bibitem[HT05]{haagerup2005new}
Uffe Haagerup and Steen Thorbj{\o}rnsen.
\newblock A new application of random matrices:
  $\mathrm{{E}xt}({C}^*_{red}(\mathbb{F}_2))$ is not a group.
\newblock {\em Annals of Mathematics}, pages 711--775, 2005.

\bibitem[HW08]{HW08}
Patrick Hayden and Andreas Winter.
\newblock Counterexamples to the maximal p-norm multiplicativity conjecture for
  all $p>1$.
\newblock {\em Communications in Mathematical Physics}, 284(1):263--280, Nov
  2008.

\bibitem[JM19]{johnston2019pairwise}
Nathaniel Johnston and Olivia MacLean.
\newblock Pairwise completely positive matrices and conjugate local diagonal
  unitary invariant quantum states.
\newblock {\em Electronic Journal of Linear Algebra}, 35:156--180, 2019.

\bibitem[JY20]{jin2020investigation}
Ryan Jin and Yu~Yang.
\newblock Investigation of the {PPT} squared conjecture for high dimensions.
\newblock {\em preprint arXiv:2010.15554}, 2020.

\bibitem[Kin03]{Kin03}
C.~King.
\newblock The capacity of the quantum depolarizing channel.
\newblock {\em IEEE Transactions on Information Theory}, 49(1):221--229, 2003.

\bibitem[KKL11]{kiem2011existence}
Young-Hoon Kiem, Seung-Hyeok Kye, and Jungseob Lee.
\newblock Existence of product vectors and their partial conjugates in a pair
  of spaces.
\newblock {\em Journal of mathematical physics}, 52(12), 2011.

\bibitem[KMP17]{Kennedy2017}
Matthew Kennedy, Nicholas Manor, and Vern Paulsen.
\newblock Composition of {PPT} maps.
\newblock {\em Quantum Information and Computation}, 18, 10 2017.

\bibitem[KNP{\etalchar{+}}21a]{kukulski2021generating}
Ryszard Kukulski, Ion Nechita, {\L}ukasz Pawela, Zbigniew Pucha{\l}a, and Karol
  {\.Z}yczkowski.
\newblock Generating random quantum channels.
\newblock {\em Journal of Mathematical Physics}, 62(6):062201, 2021.

\bibitem[KNP{\etalchar{+}}21b]{KNPPZ21}
Ryszard Kukulski, Ion Nechita, Łukasz Pawela, Zbigniew Puchała, and Karol
  Życzkowski.
\newblock {Generating random quantum channels}.
\newblock {\em Journal of Mathematical Physics}, 62(6):062201, 06 2021.

\bibitem[KO12]{kye2012classification}
Seung-Hyeok Kye and Hiroyuki Osaka.
\newblock Classification of bi-qutrit positive partial transpose entangled edge
  states by their ranks.
\newblock {\em Journal of mathematical physics}, 53(5), 2012.

\bibitem[KW99]{keyl1999optimal}
Michael Keyl and Reinhard~F Werner.
\newblock Optimal cloning of pure states, testing single clones.
\newblock {\em Journal of Mathematical Physics}, 40(7):3283--3299, 1999.

\bibitem[KW09]{KW09}
Robert Koenig and Stephanie Wehner.
\newblock A strong converse for classical channel coding using entangled
  inputs.
\newblock {\em Physical Review Letters}, 103(7):070504, 2009.

\bibitem[LG15]{Lami2015entanglebreak}
L.~Lami and V.~Giovannetti.
\newblock Entanglement–breaking indices.
\newblock {\em Journal of Mathematical Physics}, 56(9):092201, Sep 2015.

\bibitem[Liu15]{liu2015unitary}
Chaobin Liu.
\newblock Unitary conjugation channels with continuous random phases.
\newblock {\em Quantum Studies: Mathematics and Foundations}, 2(2):177--181,
  2015.

\bibitem[LS15]{lopes2015generic}
Artur~O Lopes and Marcos Sebastiani.
\newblock Generic properties for random repeated quantum iterations.
\newblock {\em Quantum Studies: Mathematics and Foundations}, 2(4):389--402,
  2015.

\bibitem[LY22]{LY22}
Hun~Hee Lee and Sang-Gyun Youn.
\newblock Quantum channels with quantum group symmetry.
\newblock {\em Comm. Math. Phys.}, 389(3):1303--1329, 2022.

\bibitem[MS22]{MS22}
Laleh Memarzadeh and Barry~C. Sanders.
\newblock Group-covariant extreme and quasiextreme channels.
\newblock {\em Phys. Rev. Res.}, 4:033206, Sep 2022.

\bibitem[MSD17]{MSD17}
Marek Mozrzymas, Michał Studziński, and Nilanjana Datta.
\newblock {Structure of irreducibly covariant quantum channels for finite
  groups}.
\newblock {\em Journal of Mathematical Physics}, 58(5):052204, 05 2017.

\bibitem[MT07]{MT07}
Christian Mastrodonato and Roderich Tumulka.
\newblock Elementary proof for asymptotics of large haar-distributed unitary
  matrices.
\newblock {\em Letters in Mathematical Physics}, 82(1):51--59, Oct 2007.

\bibitem[NC10]{nielsen2010quantum}
Michael~A Nielsen and Isaac~L Chuang.
\newblock {\em Quantum computation and quantum information}.
\newblock Cambridge University Press, 2010.

\bibitem[NS06]{Nica_Speicher_2006}
Alexandru Nica and Roland Speicher.
\newblock {\em Lectures on the Combinatorics of Free Probability}.
\newblock London Mathematical Society Lecture Note Series. Cambridge University
  Press, 2006.

\bibitem[NS21]{nechita2021graphical}
Ion Nechita and Satvik Singh.
\newblock A graphical calculus for integration over random diagonal unitary
  matrices.
\newblock {\em Linear Algebra and its Applications}, 613:46 -- 86, 2021.

\bibitem[PJPY23]{PJPY23}
Sang-Jun Park, Yeong-Gwang Jung, Jeongeun Park, and Sang-Gyun Youn.
\newblock A universal framework for entanglement detection under group
  symmetry, 2023.

\bibitem[PR04]{PR04}
D{\'e}nes Petz and J{\'u}lia R{\'e}ffy.
\newblock On asymptotics of large haar distributed unitary matrices.
\newblock {\em Periodica Mathematica Hungarica}, 49(1):103--117, Sep 2004.

\bibitem[PY23]{PY23}
Sang-Jun Park and Sang-Gyun Youn.
\newblock $k$-positivity and schmidt number under orthogonal group symmetries.
\newblock Preprint, arXiv:2306.00654, 2023.

\bibitem[RJP18]{Rahaman2018}
Mizanur Rahaman, Samuel Jaques, and Vern Paulsen.
\newblock Eventually entanglement breaking maps.
\newblock {\em Journal of Mathematical Physics}, 59, 01 2018.

\bibitem[Rud05]{rudolph2005further}
Oliver Rudolph.
\newblock Further results on the cross norm criterion for separability.
\newblock {\em Quantum Information Processing}, 4:219--239, 2005.

\bibitem[SBSGS22]{BSGS22}
Marcel Seelbach~Benkner, Jens Siewert, Otfried G\"uhne, and Gael Sent\'{\i}s.
\newblock Characterizing generalized axisymmetric quantum states in
  $d\ifmmode\times\else\texttimes\fi{}d$ systems.
\newblock {\em Phys. Rev. A}, 106:022415, Aug 2022.

\bibitem[SDN22]{singh2022ergodic}
Satvik Singh, Nilanjana Datta, and Ion Nechita.
\newblock Ergodic theory of diagonal orthogonal covariant quantum channels.
\newblock {\em arXiv preprint arXiv:2206.01145}, 2022.

\bibitem[Sim95]{Simon95}
Barry Simon.
\newblock {\em Representations of Finite and Compact Groups (Graduate Studies
  in Mathematics ; V. 10)}.
\newblock American Mathematical Society, 12 1995.

\bibitem[SN21]{singh2021diagonal}
Satvik Singh and Ion Nechita.
\newblock Diagonal unitary and orthogonal symmetries in quantum theory.
\newblock {\em Quantum}, 5:519, 2021.

\bibitem[SN22a]{singh2022diagonal}
Satvik Singh and Ion Nechita.
\newblock Diagonal unitary and orthogonal symmetries in quantum theory {II}:
  Evolution operators.
\newblock {\em Journal of Physics A: Mathematical and Theoretical}, 2022.

\bibitem[SN22b]{singh2022ppt}
Satvik Singh and Ion Nechita.
\newblock The {PPT}$^2$ conjecture holds for all {C}hoi-type maps.
\newblock {\em Annales Henri Poincar{\'e}}, 23(9):3311--3329, 2022.

\bibitem[Sta11]{Stanley11}
Richard~P. Stanley.
\newblock {\em Enumerative Combinatorics: Volume 1 (Cambridge Studies in
  Advanced Mathematics, Series Number 49)}.
\newblock Cambridge University Press, paperback edition, 12 2011.

\bibitem[TAQ{\etalchar{+}}18]{tura2018separability}
Jordi Tura, Albert Aloy, Ruben Quesada, Maciej Lewenstein, and Anna Sanpera.
\newblock Separability of diagonal symmetric states: a quadratic conic
  optimization problem.
\newblock {\em Quantum}, 2:45, 2018.

\bibitem[TiB16]{tura2016characterizing}
Jordi Tura~i Brugu{\'e}s.
\newblock {\em Characterizing entanglement and quantum correlations constrained
  by symmetry}.
\newblock Springer, 2016.

\bibitem[VW01]{VW01}
K.~G.~H. Vollbrecht and R.~F. Werner.
\newblock Entanglement measures under symmetry.
\newblock {\em Phys. Rev. A}, 64:062307, Nov 2001.

\bibitem[Wat18]{watrous2018theory}
John Watrous.
\newblock {\em The Theory of Quantum Information}.
\newblock Cambridge University Press, 2018.

\bibitem[WH02]{WH02}
R.~F. Werner and A.~S. Holevo.
\newblock Counterexample to an additivity conjecture for output purity of
  quantum channels.
\newblock {\em J. Math. Phys.}, 43(9):4353--4357, 2002.
\newblock Quantum information theory.

\bibitem[Wil17]{wilde2017quantum}
Mark~M Wilde.
\newblock {\em Quantum information theory}.
\newblock Cambridge University Press, 2017.

\bibitem[WTB17]{WTM17}
Mark~M Wilde, Marco Tomamichel, and Mario Berta.
\newblock Converse bounds for private communication over quantum channels.
\newblock {\em IEEE Transactions on Information Theory}, 63(3):1792--1817,
  2017.

\bibitem[Yu16]{yu2016separability}
Nengkun Yu.
\newblock Separability of a mixture of dicke states.
\newblock {\em Physical Review A}, 94(6):060101, 2016.

\end{thebibliography}
\bibliographystyle{alpha}
\bigskip
\hrule
\bigskip

\end{document}